 \documentclass[onecolumn, draftcls,11pt]{IEEEtran}

\usepackage{amsmath}
\usepackage{amssymb}
\usepackage{mathrsfs}
\usepackage{tikz}
\usepackage[utf8]{inputenc}
\usepackage{amsfonts} 
\usepackage{epstopdf}
\usepackage{graphicx}
\usepackage{bbm}
\usepackage{enumerate}
\usepackage{epsf,epsfig,verbatim,amssymb,amsmath,array,cite,amsthm,subfigure,multicol,multirow}  
\usepackage{psfrag,bm,xspace}
\usepackage{hhline}

\newcommand{\beq}{\begin{equation}}
\newcommand{\eeq}{\end{equation}}

\newtheorem{lem}{Lemma}

\newtheorem{theorem}{Theorem}

\newtheorem{corollary}{Corollary}
\newtheorem{claim}{Claim}
\newtheorem{definition}{Definition}

\theoremstyle{definition}
\newtheorem{example}{Example}

\newtheorem*{prob*}{Problem}

\theoremstyle{remark}
\newtheorem{remark}{Remark}

\allowdisplaybreaks 

\global\long\def\RR{\mathbb{R}}

\global\long\def\ZZ{\mathbb{Z}}
\global\long\def\EE{\mathbb{E}}

\global\long\def\11{\mathbbm{1}}

\allowdisplaybreaks

\begin{document}
\title{How to Compute Modulo Prime-Power Sums ?}
\author{\IEEEauthorblockN{Mohsen Heidari, Farhad Shirani, and S. Sandeep
    Pradhan\thanks{This work was presented in part at IEEE International Symposium on Information Theory (ISIT), July 2016 and July 2017. }}\\
\IEEEauthorblockA{Department of Electrical Engineering and Computer Science,\\
University of Michigan, Ann Arbor, MI 48109, USA.\\
Email: \tt\small mohsenhd@umich.edu, fshirani@umich.edu, pradhanv@umich.edu}\\ August 2, 2017}

%
%

%
%




\maketitle
\begin{abstract}
A new class of structured codes called Quasi Group Codes (QGC) is introduced. A QGC is a subset
of a group code. In contrast with group codes, QGCs are not closed under group addition. The parameters of the QGC can be chosen such that the size of $\mathcal{C}+\mathcal{C}$ is equal to any number between $|\mathcal{C}|$  and $|\mathcal{C}|^2$ . We analyze the performance of a specific class of QGCs. This class of QGCs is constructed by assigning single-letter distributions to the indices of the codewords in a group code. Then, the QGC is defined as the set of codewords whose index is in the typical set corresponding to these single-letter distributions. The asymptotic performance limits of this class of QGCs is characterized using single-letter information quantities. Corresponding covering and packing bounds are derived. It is shown that the point-to-point channel capacity and optimal rate-distortion
function are achievable using QGCs. Coding strategies based on QGCs are introduced for three fundamental
multi-terminal problems: the K\"orner-Marton problem for modulo prime-power sums, computation over the multiple access channel (MAC), and MAC with distributed states.  For each problem a single-letter achievable rate-region is derived. It is shown, through examples, that the coding strategies improve upon the previous strategies based on unstructured codes, linear codes and group codes.

\end{abstract}


\section{Introduction}
\IEEEPARstart{T}{he}  conventional technique of deriving the performance limits for any
communication problem in information theory is  via
random coding \cite{Cover2006} involving so-called  Independent Identically Distributed
(IID) random codebooks. Since such a code possesses only single-letter empirical
properties, coding techniques are constrained to exploit only these
for enabling efficient communication. We refer to them as unstructured
codes. These techniques have been  proven to achieve capacity for point-to-point (PtP) channels, multiple-access channel (MAC) and particular multi-terminal channels such as degraded
broadcast channels. Based on these initial successes,  
it was widely believed that  one can achieve the capacity of any network 
communication problem using IID codebooks.

Stepping beyond this conventional technique, K\"orner and Marton
\cite{korner-marton} proposed a technique based on
statistically correlated codebooks (in particular, identical random linear codes) 
possessing algebraic closure
properties, henceforth referred to as (random) structured codes, that
outperformed all techniques based on (random) unstructured codes. 
This technique was proposed for the problem of distributed computation of the modulo two sum of two correlated symmetric binary sources \cite{korner-marton}.
Applications of structured codes  were also studied for various multi-terminal communication systems, including, but not limited to, distributed source coding \cite{ Dinesh_dist_source_coding , Ahlswede-Han, Han_Kobayashi_DSC, Han1987}, computation over MAC \cite{Nazer_Gasper_Comp_MAC, Nazer_Gaspar_2016, ISIT15_transversal, ISIT16_QGC, Arun_comp_over_MAC_ISIT13, Zhan2013, Appuswamy2013}, MAC with side information  \cite{Philosof2011, Philosof-Zamir, Arun-MAC-with-States, Ahlswede-Han, ISIT17_MAC_States}, the joint source-channel coding over MAC \cite{ISIT16_MAC_Correlated_Sources}, multiple-descriptions   \cite{ISIT15Lattices}, interference channel  \cite{Vishwanath_Jafar_Shamai_2008,hong_caire,Bresler2010a, Niesen2013, Jafarian2012, Ordentlich2012, ISIT16_QLC_Interference},  broadcast channel  \cite{Aron-BC-ISIT13} and MAC with Feedback \cite{ISIT17_Arxiv_MAC_Feedback}. In these works, algebraic structures are exploited to design new coding schemes which outperform all coding schemes solely based on random unstructured codes.
The emerging opinion in this regard is that even if computational complexity is a non-issue, algebraic structured codes may be
necessary,  in a deeply fundamental way, to achieve optimality in transmission and storage of information in networks. 

There are several algebraic structures such as fields, ring and groups. Linear codes are defined over finite fields. The focus of this work is on structured codes defined over the ring of modulo-$m$ integers, that is $\ZZ_{m}$. Group codes are a class of structured codes constructed over $\ZZ_{m}$, and were first studied by Slepian \cite{Slepian1968} for the Gaussian channel.  A group code over $\ZZ_m$ is defined as a set of codeswords that is closed under the element-wise modulo-$m$ addition. Linear codes are a special case of group codes (the case when $m$ is a prime). There are two  main incentives to study group codes. First, linear codes are defined only over finite fields, and finite fields exists only when alphabet sizes equal to a prime power, i.e., $\ZZ_{p^r}$. Second, there are several communications problems in which group codes have superior performance limits compared to linear codes. As an example, group codes over $\ZZ_8$ have better error correcting properties than linear codes for communications over an additive white Gaussian noise channel with 8-PSK constellation \cite{Loeliger1991}. As an another example, construction of polar codes over alphabets of size equal to a prime power $p^r$, is more efficient with a module structure rather than a vector space structure \cite{Sasoglu_Polar_codes_book2009,sahebi_polar_codes, Abbe2012, Park2013}. Bounds on the achievable rates of group codes in PtP communications were studied in  \cite{Ahlswede1971, Ahlswede1971a, Ahlswede1971c, Como2009, Loeliger1991, Aria_group_codes}.  Como \cite{Como2009} derived the largest achievable rate using group codes for certain PtP channels.  In \cite{Ahlswede1971}, Ahlswede showed that group codes do not achieve the capacity of a general discrete memoryless channel. In \cite{Aria_group_codes},  Sahebi et.al., unified the previously known works, and characterized the ensemble of all group codes over finite commutative groups. In addition, the authors derived the optimum asymptotic performance limits of group codes for PtP channel/source coding problems.



It appears that there is a trade-off between cooperation and communication/compression  in networks.  To see this consider the following observations.  K\"orner and Marton suggested the use of identical linear codes
to effect binning of two correlated binary sources when the
objective is to reconstruct the modulo-two sum of the sources
at the decoder. A similar approach  has been used in interference alignment 
using lattices and linear codes in channel coding over interference
channels   \cite{Jafar_Vishwanath_2010, Vishwanath_Jafar_Shamai_2008}.          
The aligning users must use identical linear/lattice codes (modulo shifts). 
In summary, to achieve network cooperation the users must use identical linear codes. 
A linear code, group code or lattice code $\mathcal{C}$ is completely structured in the sense
that the size of $\mathcal{C}+\mathcal{C}$ \emph{equals the size} of $\mathcal{C}$.
However, if the objective is to have the full
reconstruction of both the sources at the decoder (Slepian-Wolf
setting \cite{Slepian-Wolf}), then it has been shown that  using
identical binning can be strictly suboptimal. In general,
to achieve the Slepian-Wolf performance limit, one needs to
use independent unstructured binning of the two sources using
Shannon-style unstructured code ensembles \cite{Cover2006}. A similar observation was made recently 
regarding the interference channels \cite{ISIT16_QLC_Interference}: each cooperating transmitter using identical 
linear codes must pay some penalty  in terms of sacrificing her/his rate for the overall
good of the network. A selfish user intent on maximizing individual throughput must use
essentially independent Shannon-style unstructured code ensembles. 
A  code $\mathcal{C}$ used in random coding in
Shannon ensembles is completely unstructured (complete lack of structure) in the sense that 
the size of $\mathcal{C}+\mathcal{C}$ nearly \emph{equals the square of the size} of $\mathcal{C}$.

This gap between the completely structured codes and the completely unstructured codes leads to the following question: Is there a spectrum of strategies involving partially structured codes or partially unstructured codes that lie between these two extremes? Based
on this line of thought, we consider a new class of  codes which are not fully closed with respect to any algebraic structure but maintain a
degree of ``closedness" with respect to some. In our earlier works \cite{ISIT15_transversal, ISIT16_QGC}, it was observed  that adding a certain set of codewords to a group code improves the performance of the code. Based on these observations\footnote{The motivation for this work comes from our earlier work on multi-level polar codes based on $\ZZ_{p^r}$ \cite{sahebi_polar_codes}. A multi-level polar code is not a group code. But it is a subset a nontrivial group code. }, we introduce a new class of structured code ensembles
called Quasi Group Codes (QGC)
 whose \emph{closedness can be controlled}. A QGC is
a subset of a group code. The degree of closedness of a QGC can be controlled in the sense that the size of $\mathcal{C}+\mathcal{C}$ can be any number between the size of $\mathcal{C}$ and the square of the size of $\mathcal{C}$.
We provide a method for constructing specific subsets of these codes by
putting single-letter distributions on the indices of the codewords.
We are able to analyze the performance of the resulting code ensemble,
and  characterize the asymptotic performance using
single-letter information quantities. By choosing the single-letter
distribution on the indices one can operate anywhere
in the spectrum between the two extremes: group codes
and unstructured codes. 

The contributions of this work are as follows. A new class of  codes over groups called Quasi Group Codes (QGC) is introduced. These codes are constructed by taking subsets of group codes. This work considers QGCs over cyclic groups $\mathbb{Z}_{p^r}$. One can use the fundamental theorem of finitely generated Abelian groups to generalize the results of this paper to QGCs over non-cyclic finite Abelian groups. Information-theoretic characterizations for the asymptotic performance limits and properties of QGCs for source coding and channel coding problems are derived in terms of single-letter information quantities. Covering and packing bounds are derived for an ensemble of QGCs. Next, a binning technique for the QGCs is developed by constructing nested QGCs. As a result of these bounds,  the PtP channel capacity and optimal rate-distortion function of sources are shown to be achievable using nested QGCs. The applications of QGCs in some multi-terminal communications problems are considered. More specifically our study includes the following problems:


\paragraph*{\bf Distributed Source Coding}
A more general version of K\"orner-Marton problem is considered. In this problem, there are two distributed sources taking values from $\ZZ_{p^r}$.  The sources are to be compressed in a distributed fashion. The decoder wishes to compute the modulo $p^r$-addition of the sources losslessly.  

\paragraph*{\bf Computation over MAC}
In this problem, two transmitters wish to communicate independent information to a receiver over a MAC. The objective is to decode the modulo-$p^r$ sum of the codewords sent by the transmitters at the receiver. This problem is of interest in its own right. Moreover, this problem finds applications as an intermediate step in the study of other fundamental problems such as the interference channel and broadcast channel \cite{Aron-BC-ISIT13, Padakandla_IC_2012}. 

\paragraph*{\bf MAC with Distributed States}
In this problem,  two transmitters wish to communicate independent information to a  receiver over a MAC. The transition probability  between the output and the inputs depends on states $S_1$, and  $S_2$ corresponding to the two transmitters. The state sequences are generated IID according to some fixed joint probability distribution. Each encoder observes the corresponding state sequence non-causally. The objective of the receiver is to decode the messages of both transmitters. 

These problems are formally defined in the sequel. For each of these problems,  a coding scheme based on (nested) QGCs is introduced. It is shown, through examples, that the coding scheme improves upon the best-known coding strategies based on unstructured codes, linear codes and group codes.   In addition, for each problem  a new single-letter achievable rate-region is derived. These rate-regions strictly subsume all the previously known rate-regions for each of these problems. 

The rest of this paper is organized as follows: Section \ref{sec: preliminaries} provides the preliminaries and notations. In Section \ref{sec: proposed scheme} we introduce QGC's and define an ensemble of QGCs. Section \ref{sec: Properties of QGC} characterizes basic properties of QGCs. Section \ref{sec: binning for QGCs} describes a method for binning using QGCs. In Section  \ref{sec: dist} and Section \ref{sec: comp_over_mac}, we discuss the applications of QGC's in distributed source coding and computation over MAC, respectively. In Section \ref{sec: MAC with States } we investigate applications of nested QGCs in the problem of MAC with states. Finally, Section \ref{sec: conclusion} concludes the paper.

\section{Preliminaries} \label{sec: preliminaries}
\subsection{Notations}
We denote (i) vectors using lowercase bold letters such as $\mathbf{b}, \mathbf{u}$, (ii) matrices using uppercase bold letters such as $\mathbf{G}$, (iii) random variables using capital letters such as $X,Y$, (iv) numbers, realizations of random variables and elements of sets using lower case letters such as $a, x$. Calligraphic letters such as $\mathcal{C}$ and  $\mathcal{U}$  are used to represent sets. For shorthand, we denote the set $\{1, 2, \dots, m\}$ by $[1:m]$.

\subsection{Definitions}
A group is a set equipped with a binary operation denoted by “$+$”. Given a prime power $p^r$, the group of integers modulo-$p^r$ is denoted by $\ZZ_{p^r}$, where the underlying set is $\{0,1,\cdots, p^r-1\}$, and the addition  is modulo-$p^r$ addition. Given a group $M$, a subgroup is a subset $H$ which is closed under the group addition.  For $s \in [0: r]$, define $$H_{s}=p^{s}\ZZ_{p^r}=\{0, p^{s}, 2p^{s}, \cdots,  (p^{r-s}-1)p^{s}\},$$ and $T_s=\{0,1, \cdots
 , p^s-1\}$. For example, $H_0=\ZZ_{p^r}, T_0=\{0\}$, whereas $H_r=\{0\}, T_r=\ZZ_{p^r}$. Note, $H_s$ is a subgroup of $\ZZ_{p^r}$, for $s \in [0: r]$. Given $H_{s}$ and $T_s$, each element $a$ of $\ZZ_{p^r}$ can be represented uniquely as a sum $a=t+h$, where $h\in H_{s}$ and $t \in T_s$. We denote such $t$ by $[a]_{s}$. Therefore, with this notation, $[\cdot ]_s$ is a function from $\ZZ_{p^r}\rightarrow T_s$. Note that this function satisfies the distributive property:
 \begin{align*}
 [a + b]_s=\Big[ [a]_s + [b]_s \Big]_s
 \end{align*}
 
For any elements $a,b \in \ZZ_{p^r}$, we define the multiplication $a\cdot b$ by adding $a$ with itself $b$ times. Given a positive integer $n$,  denote $\ZZ_{p^r}^n=\bigotimes_{i=1}^n\ZZ_{p^r}$. Note that $\ZZ_{p^r}^n$ is a group, whose addition is element-wise and its underlying set is $\{0,1, \dots, p^r-1\}^n$.   We follow the definition of shifted group codes on $\ZZ_{p^r}$ as in \cite{Aria_group_codes} \cite{Dinesh_dist_source_coding}.
\begin{definition}[Shifted Group Codes]\label{def: group codes}
An $(n,k)$-\textit{shifted group code} over $\ZZ_{p^r}$ is defined as
\begin{equation}\label{eq: group code}
\mathcal{C}=\{ \mathbf{u} \mathbf{G}+\mathbf{b}: \mathbf{ u} \in \ZZ_{p^r}^{k}\},
\end{equation}
where $\mathbf{b}\in \ZZ^n_{p^r}$ is the translation (dither) vector and  $\mathbf{G}$ is a $k\times n$ generator matrix with elements in $\ZZ_{p^r}$.
\end{definition}

We follow the definition of typicality as in \cite{Csiszar}. 
\begin{definition}
For any probability distribution $P$ on $\mathcal{X}$ and $\epsilon >0$, a sequence $\mathbf{x}^n\in \mathcal{X}^n$ is said to be $\epsilon$-typical with respect to $P$ if
\begin{align*}
\Big| \frac{1}{n} N(a|\mathbf{x}^n)-P(a)\Big| \leq \frac{\epsilon}{|\mathcal{X}|},~ \forall a\in \mathcal{X},
\end{align*}
and, in addition, no $a\in \mathcal{X}$ with $P(a)=0$ occurs in $\mathbf{x}^n$. Note $N(a|x^n)$ is the number of the occurrences of $a$ in the sequence $\mathbf{x}^n$. The set of all $\epsilon$-typical sequences with respect to a probability distribution $P$ on $\mathcal{X}$ is denoted by $A_\epsilon^{(n)}(X)$.
\end{definition}

\section{Quasi Group Codes}\label{sec: proposed scheme}
Linear codes and group codes are two classes of structured codes. These codes are closed under the addition of the underlying group or field.  It is known in the literature that coding schemes based on linear codes and group codes  improve upon unstructured random coding strategies \cite{korner-marton}. In this section, we propose a new class of structured codes called \textit{quasi-group codes}.

 A QGC is defined as a subset of a group code. Therefore, QGCs are not necessarily closed under the addition of the underlying group. An $(n,k)$ shifted group code over $\ZZ_{p^r}$ is defined as the image of a linear mapping from $\ZZ_{p^r}^k$ to $\ZZ_{p^r}^n$ as in Definition \ref{def: group codes}. Let $\mathcal{U}$ be an arbitrary subset of $\ZZ_{p^r}^k$. Then a QGC is defined as
  \begin{align} \label{eq: QGC codebook}
 \mathcal{C}=\{\mathbf{u}\mathbf{G}+\mathbf{b}: \mathbf{u}\in \mathcal{U}\}, 
\end{align}
where $\mathbf{G}$ is a $k\times n$ matrix and $\mathbf{b}$ is an element of $\ZZ_{p^r}^n$.   If $\mathcal{U}=\ZZ_{p^r}^k$, then $\mathcal{C}$ is a shifted group code.  As we will show, by changing the subset $\mathcal{U}$, the code $\mathcal{C}$ ranges from completely structured codes (such as group codes and linear codes) where $|\mathcal{C} + \mathcal{C}|=|\mathcal{C}|$ to completely unstructured codes where $|\mathcal{C} + \mathcal{C}|=|\mathcal{C}|^2$. For a general subset $\mathcal{U}$, it is difficult to derive a single-letter characterization of the asymptotic performance of such codes. To address this issue, we present a special type of subsets $\mathcal{U}$ for which single-letter characterization of their performance is possible.  

\begin{example}
Let $U$ be a random variable over $\ZZ_{p^r}$ with PMF $P_U$. For $\epsilon>0$, set $\mathcal{U}$ to be the set of all $\epsilon$-typical sequences $\mathbf{u}^k$. More precisely, define $\mathcal{U}= A_{\epsilon}^{(k)}(U)$. In this case, the set $\mathcal{U}$ is determined by the PMF $P_U$ and $\epsilon$. For instance, if $U$ is uniform over $\ZZ_{p^r}$, then $\mathcal{U}=\ZZ_{p^r}^k$.  
\end{example}
Next, we provide a more general construction of $\mathcal{U}$: 
\paragraph*{\bf Construction of $\mathcal{U}$}
Given a positive integer $m$, consider $m$ mutually independent random variables $U_1, U_2, \cdots, U_m$. Suppose each $U_i$ takes values from $\ZZ_{p^r}$ with distribution $P_{U_i}, i\in [1:m]$. For $\epsilon>0$, and positive integers $k_i$, define $\mathcal{U}$ as a Cartesian product of the $\epsilon$-typical sets of $U_i, i\in [1:m]$. More precisely,
\begin{align}\label{eq: U}
\mathcal{U}\triangleq \bigotimes_{i=1}^m A_{\epsilon}^{(k_i)}(U_i).
\end{align} 
In this  construction,  set $\mathcal{U}$ is determined by $m$, $k_i, \epsilon$, and the PMFs $P_{U_i}, i\in[1:m]$. 

For more convenience, we use a different representation for this construction. Let $k\triangleq \sum_{i=1}^m k_i$. Denote $q_i \triangleq \frac{k_i}{k}$. Note that $q_i\geq 0$ and $\sum_i q_i=1$. Therefore, we can define a random variable $Q$ with $P(Q=i)=q_i$. Define a random variable $U$ with the conditional distribution $P(U=a|Q=i)=P(U_i=a)$ for all $a\in \ZZ_{p^r}, i\in [1:m]$. With this notation the set $\mathcal{U}$ in the above construction is characterized by a finite set $\mathcal{Q}$, a pair of random variables $(U,Q)$ distributed over $\ZZ_{p^r}\times \mathcal{Q}$, an integer $k$, and $\epsilon>0$.  The joint distribution of $U$ and $Q$ is denoted by $P_{UQ}$. Note that we assume $P_Q(q)>0$ for all $q\in \mathcal{Q} $. For a more concise notation, we identify the set $\mathcal{U}$ without explicitly specifying $\epsilon$.  With the notation given for the construction of $\mathcal{U}$, we define its corresponding QGC.
\begin{definition}\label{def: QGC}
An  $(n,  k)$- QGC $\mathcal{C}$ over $\ZZ_{p^r}$ is defined as in (\ref{eq: QGC codebook}) and \eqref{eq: U}, and is characterized by a matrix $\mathbf{G}\in \ZZ_{p^r}^{k\times n}$, a translation $\mathbf{b}\in \ZZ^n_{p^r}$, and a pair of random variables $(U, Q)$ distributed over the finite set $\ZZ_{p^r} \times \mathcal{Q}$. The set $\mathcal{U}$ in \eqref{eq: U} is defined as the index set of $\mathcal{C}$.
\end{definition} 

\begin{remark}
 Any shifted group code over $\ZZ_{p^r}$ is a QGC.
\end{remark}
\begin{remark}
Let $\mathcal{C}$ be an $(n,k)$-QGC with randomly selected matrix and translation. In contrast to linear codes, codewords of $\mathcal{C}$ are not necessary pairwise independent. 
\end{remark}

Fix $ (n,  k)$ and random variables $(U, Q)$. We create an ensemble of codes by taking the collection of all $ (n, k)$-QGCs with random variables $(U, Q)$,  for all matrices $\mathbf{G}$ and translations $\mathbf{b}$. We call such a collection as the ensemble of $(n,k)$-QGCs with random variables $(U,Q)$. A random codebook $\mathcal{C}$ from this ensemble is chosen by selecting the elements of $\mathbf{G}$ and $\mathbf{b}$ randomly and uniformly over $\ZZ_{p^r}$. In order to characterize the asymptotic performance limits of QGCs, we need to define sequences of ensembles of QGCs.  For any positive integer $n$, let $k_n=cn$, where  $c>0$ is a constant. Consider the sequence of the ensembles of $(n,k_n)$-QGCs with random variables $(U,Q)$. In the next two lemmas, we characterize the size of randomly selected codebooks from these ensembles. 
\begin{lem}\label{lem: size of U}
Let $\mathcal{U}_n$ be the index set associated with the ensemble of $(n,k_n)$-QGCs with random variables $(U,Q)$ and $\epsilon>0$, where $k_n=cn$ for a constant $c>0$. Then there exists $N>0$, such that for all $n>N$, 
\begin{align*}
\Big|\frac{1}{k_n}\log_2 |\mathcal{U}_n|- H(U|Q)\Big| \leq \epsilon',
\end{align*}
where $\epsilon'$ is a continuous function of $\epsilon$, and $\epsilon'\rightarrow 0$ as $\epsilon \rightarrow 0$.
\end{lem}
\begin{proof}
The proof is Given in Appendix \ref{APP: proof of lemma size of U}
\end{proof}

\begin{remark}\label{rem: size of a random QGC}
Let $\mathcal{C}_n$ be an $(n,k_n)$-QGC with random variables $(U, Q)$. Then, using Lemma \ref{lem: size of U}, for large enough $n$,
\begin{align}\label{eq: QGC_rate}
\frac{1}{n}\log_2|\mathcal{C}_n| \leq  \frac{k_n}{n}H(U|Q) + \epsilon'.
\end{align} 
\end{remark}

\begin{lem}\label{lem: injective map}
Let $\mathcal{U}_n$ be the index set associated with the ensemble of $(n,k_n)$-QGCs with random variables $(U,Q)$, where $k_n=cn$ for a constant $c>0$. Define a map $\Phi_n: \mathcal{U}_n \rightarrow \ZZ_{p^r}^{n}$, $\Phi_n(\mathbf{u})=\mathbf{uG}_n$ for all $\mathbf{u}\in \mathcal{U}_n$, where $\mathbf{G}_n$ is a $k_n \times n$ matrix whose elements are chosen randomly and uniformly from $\ZZ_{p^r}$. Suppose $H(U|[U]_s,Q) < \frac{1}{c} (r-s) \log_2p 	$ for all $s\in [0 : r-1]$.  Then, for any $\delta>0$, there exists $N>0$ such that for each $n>N$  and for any randomly selected $\mathbf{U}\in \mathcal{U}_n$, the size of inverse image $|\Phi_n^{-1}(\Phi(\mathbf{U}))|=1$ with probability at least $(1-\delta)$.
\end{lem}

\begin{proof}
The proof is provided in Appendix \ref{APP: proof of lemma injective map}.
\end{proof}

In the case of linear codes ($r=0$), suppose $\Phi_n$ is  the map induced by an $(n,k)$-linear code. If $k\geq n$, then the inverse image of any vector by the map $\Phi_n$ has more than one candidate. However, based on Lemma \ref{lem: injective map}, this is not the case for a the map induced by a  $(n,k)$-QGC. 

In our earlier work, we considered a special class of QGCs which is called transversal group codes  \cite{ISIT15_transversal}. 

\begin{definition}[Transversal Group Codes]
 An $(n,k_1,k_2, \dots, k_r)$-transversal group code over $\ZZ_{p^r}$ is defined as 
$$\mathcal{C}=\Big \{ \sum_{s=1}^r \mathbf{u}_s \mathbf{G}_s+\mathbf{b}: \mathbf{u}_s \in T_s^{k_s}, s\in [1:r] \Big \},$$ where  $T_s=[0:p^s-1]$, $\mathbf{b}\in \ZZ^n_{p^r}$ and $\mathbf{G}_s$ is a $k_s \times n$ matrix with elements in $\ZZ_{p^r}$.
\end{definition}
A transversal group code is a code created by removing a certain set of codewords from a group code. Based our results for transversal group codes, we introduce QGCs.  
\section{Properties of Quasi Group Codes}\label{sec: Properties of QGC}
It is known that if $\mathcal{C}$ is a random unstructured codebook, then $|\mathcal{C}+\mathcal{C}|\approx|\mathcal{C}|^2$ with high probability. Group codes on the other hand are closed under the addition, which means $|\mathcal{C}+\mathcal{C}|=|\mathcal{C}|$. Comparing to unstructured codes, when the structure of the group codes matches with that of a multi-terminal channel/source coding problem, it turns out that higher/lower transmission rates are obtained. However, in certain problems, the structure of the group codes is too restrictive. More precisely, when the underlying group is $\ZZ_{p^r}$ for $r\geq 2$, there are several nontrivial subgroups. These subgroups cause a penalty on the rate of a group code. This results in lower transmission rates in channel coding and higher transmission rates in source coding. 

Quasi group codes balance the trade-off between the structure of the group codes and that of the unstructured codes. More precisely, when $\mathcal{C}$ is a QGC, then $|\mathcal{C}+\mathcal{C}|$ is a number between $|\mathcal{C}|$ and $|\mathcal{C}|^2$. This results in a more flexible algebraic structure to match better with the structure of the channel or source. This trade-off is shown more precisely in the following lemma. 

\begin{lem}\label{lem:sum of two quasi group code}
Let $\mathcal{C}_i, i=1,2$ be an $(n,k_i)$-QGC over $\ZZ_{p^r}$ with random variables $(U_i, Q)$. Consider the joint distribution among $(U_1, U_2, Q)$ that is consistent with marginals $(U_1, Q)$ and $(U_2, Q)$, and that satisfies the Markov chain  $U_1 \leftrightarrow Q \leftrightarrow U_2$.
%
\begin{enumerate} 
\item Suppose $k_1=k_2=k$, and let $\mathcal{D}$ be an $(n, k)$-QGC with random variable $(U_1+U_2, Q)$. The generator matrices of $\mathcal{C}_1,\mathcal{C}_2$ and $\mathcal{D}$ are identical. Suppose $(\mathbf{U}_1, \mathbf{U}_2)$ are chosen randomly and uniformly from $\mathcal{U}_1 \times \mathcal{U}_2$. Let $\mathbf{X}_i$ be the codewrod of $\mathcal{C}_i$ corresponding to $\mathbf{U}_i, i=1,2$. Then, for all $\epsilon>0$ and all sufficiently large n, $$P\{\mathbf{X}_1+\mathbf{X}_2 \in \mathcal{D}\}\geq 1-\delta(\epsilon),$$
where $\delta(\epsilon)\rightarrow 0$ as $\epsilon \rightarrow 0$.

\item $\mathcal{C}_1+ \mathcal{C}_2$ is an  $(n,k_1+k_2)$-QGC with random variables $\left(U_I, (Q,I)\right)$, where $I\in \{1,2\}$. If $I=i$, then $U_I=U_i, i=1,2$. In addition,  $P(I=i, Q=q, U_I=a)=\frac{k_i}{k_1+k_2} P(Q=q) P(U_i=a | Q=q)$, for all $a\in \ZZ_{p^r}, q\in \mathcal{Q}$ and $i=1,2$.
%
\end{enumerate}
\end{lem}

\begin{proof}
Using \eqref{eq: QGC codebook}, suppose $\mathcal{U}_i$ is the index set, $\mathbf{G}_i$ is the matrix, and $\mathbf{b}_i$ is the translation of $\mathcal{C}_i, i=1,2$.  For the first statement, since $k_1=k_2$ and $\mathbf{G}_1=\mathbf{G}_2$, then $\mathbf{X}_i=\mathbf{U}_i \mathbf{G}+\mathbf{b}_i, i=1,2$. With this notation, $\mathbf{X}_1 + \mathbf{X}_2 =( \mathbf{U}_1+\mathbf{U}_2) \mathbf{G}+\mathbf{b}_1+\mathbf{b}_2$. By definition, $\mathcal{U}_i$ is the product of typical sets as in \eqref{eq: U}. By $\mathcal{U}_d$ denote the index set of $\mathcal{D}$. By Lemma \ref{lem: sum of typical sets }, $\mathcal{U}_d \subseteq (\mathcal{U}_1+\mathcal{U}_2)$. Thus, $\mathcal{D} \subseteq (\mathcal{C}_1+\mathcal{C}_2)$.  Since $\mathbf{U}_1, \mathbf{U}_2$ are independent random variables with uniform distribution over $\mathcal{U}_1 \times \mathcal{U}_2$, then $\mathbf{U}_1+\mathbf{U}_2 \in \mathcal{U}_d$ with probability at least $(1-\delta(\epsilon))$. This follows from standard arguments on typical sets \cite{ElGamal-book}. As a result, $\mathbf{X}_1+\mathbf{X}_2 \in \mathcal{D}$ with probability at least $(1-\delta(\epsilon))$.


For the second statement, we have 
\begin{align*}
\mathcal{C}_1+\mathcal{C}_2 = \{[\mathbf{u}_1, \mathbf{u}_2] \begin{bmatrix}
\mathbf{G}_1\\ \mathbf{G}_2 \end{bmatrix}+\mathbf{b}_1+\mathbf{b}_2: \mathbf{u}_i \in \mathcal{U}_i, i=1,2 \}.
\end{align*}
Therefore,  $\mathcal{C}_1+\mathcal{C}_2$ is an $(n, k_1+k_2)$-QGC. Note that $\mathcal{U}_1 \times \mathcal{U}_2$ is the index set associated with this codebook. The statement follows, since each subset $\mathcal{U}_i, i=1,2$ is a Cartesian product of $\epsilon$-typical sets of $U_{i, q}, q\in \mathcal{Q}$. The random variables $(U_I, (Q,I))$ describes such a Cartesian product. 

For the third statement, the inequalities follow from standard counting arguments. 
\end{proof}
{ We explain the intuition behind the lemma. Suppose  $\mathcal{C}_1, \mathcal{C}_2$ and $\mathcal{D}$ are  QGCs with identical generator matrices and with random variables $U_1, U_2$ and $U_1+U_2$, respectively. Then $\mathcal{D} = \mathcal{C}_1+\mathcal{C}_2$ with probability approaching one. }
\begin{remark}
If $\mathcal{C}_1$ and $\mathcal{C}_2$ are the QGCs as in Lemma \ref{lem:sum of two quasi group code}, then from standard counting arguments we have  $$ \max\{| \mathcal{C}_1|, |\mathcal{C}_2| \} \leq |\mathcal{C}_1+ \mathcal{C}_2| \leq \min\{p^{rn}, |\mathcal{C}_1| \cdot |  \mathcal{C}_2|\} $$
\end{remark}


In what follows, we derive a packing bound and a covering bound for a QGC with matrices and translation chosen randomly and uniformly. Fix a PMF $P_{XY}$, and suppose an $\epsilon$-typical sequence $\mathbf{y}$ is given with respect to the marginal distribution $P_Y$.  Consider the set of all codewords in a QGC that are jointly typical with $\mathbf{y}$ with respect to $P_{XY}$. In the packing lemma, we characterize the conditions under which the probability of this set is small. This implies the existence of a  ``good-channel" code which is also a QGC. In the covering lemma, we derive the conditions for which, with high probability, there exists at least one such codeword in a QGC. In this case a ``good-source" code exists which is also a QGC. These conditions are provided in the next two lemmas. 

For any positive integer $n$, let $k_n=cn$, where  $c>0$ is a constant. Let $\mathcal{C}_n$ be a sequence of $(n,k_n)$-QGCs with random variables $(U,Q)$, $\epsilon>0$. By $R_n$ denote the rate of $\mathcal{C}_n$. Suppose the elements of the generator matrix and the translation of $\mathcal{C}_n$ are chosen randomly and uniformly from $\ZZ_{p^r}$.

\begin{lem}[Packing]\label{lem: packing}
Let $(X,Y)\sim P_{XY}$.  By $\mathbf{c}_n(\theta)$ denote the $\theta$th codeword of $\mathcal{C}_n$. Let $\tilde{\mathbf{Y}}^n$ be a random sequence distributed according to $\prod_{i=1}^n P_{Y|X}(\tilde{y}_i|c_{n, i}(\theta))$. Suppose, conditioned on $\mathbf{c}_n(\theta)$, $\tilde{\mathbf{Y}}^n$ is independent of all other codewords in $\mathcal{C}_n$. Then, for any $\theta \in [1:|\mathcal{C}_n|]$, and $\delta>0$, $\exists N>0$ such that for all $n>N$, $$P\{\exists \mathbf{x}\in \mathcal{C}_n: (\mathbf{x}, \tilde{\mathbf{Y}}^n)\in A_{\epsilon}^{(n)} (X,Y), \mathbf{x}\neq \mathbf{c}_n(\theta)\}<\delta,$$ if the following bounds hold
 \begin{align}\label{eq: packing bound}
R_n < \min_{0 \leq s\leq r-1} \frac{H(U|Q)}{H(U|Q,[U]_s)}\big( \log_2p^{r-s}-H(X|Y[X]_s) +\eta(\epsilon)\big),
\end{align}
where $\eta(\epsilon)\rightarrow 0$ as $\epsilon \rightarrow 0$ .
\end{lem}
\begin{proof}
See Appendix \ref{sec: proof of the packing lemma}.
\end{proof}

\begin{lem}[Covering]\label{lem: covering}
 Let $(X,\hat{X})\sim P_{X \hat{X}}	$,  where $\hat{X}$ takes values from $\ZZ_{p^r}$.  Let $\mathbf{X}^n$ be a random sequence distributed according to $\prod_{i=1}^n P_X(x_i)$.  Then, for any $\delta>0$, $\exists N>0$ such that for all $n>N$,  $$ P\{ \exists \hat{\mathbf{x}} \in \mathcal{C}_{n}: (\mathbf{X}^n,\mathbf{\hat{x}})\in A_{\epsilon}^{(n)} (X,\hat{X})\} > 1-\delta$$
 if the following inequalities hold
\begin{align}\label{eq: covering bound}
R_n > \max_{1 \leq s\leq r} \frac{H(U|Q)}{H([U]_s|Q)}\big( \log_2 p^s-H([\hat{X}]_s|X) +\eta(\epsilon)\big).
\end{align}
	
\end{lem}
\begin{proof}
See Appendix \ref{sec: proof of the covering lemma}.
\end{proof}


Lemma \ref{lem:sum of two quasi group code},  \ref{lem: packing} and Lemma \ref{lem: covering} provide a tool to derive inner bounds for achievable rates using quasi group codes in multi-terminal channel coding and source coding problems.

\section{Binning Using QGC}\label{sec: binning for QGCs}
Note that in a randomly generated QGC, all codewords have uniform distribution over $\ZZ^n_{p^r}$. However, in many communication setups we require application of codes with non-uniform distributions. In addition, we require binning techniques for various multi-terminal communications. In this section, we present a method for random binning of QGCs. In the next sections, we will use random binning of QGCs to propose coding schemes for various multi-terminal problems.

We introduce nested quasi group codes using which we propose a random binning technique. A QGC $\mathcal{C}_I$ is said to be nested in a QGC $\mathcal{C}_O$, if $\mathcal{C}_I \subset \mathcal{C}_O +\mathbf{b}$, for some translation $\mathbf{b}$. Suppose $\mathcal{C}_O$ is an $(n,k+l)$-QGC with the following structure, 
\begin{align}\label{eq: nested QGC codebook}
\mathcal{C}_O\triangleq \{\mathbf{u}\mathbf{G}+ \mathbf{v}\mathbf{\tilde{G}} +\mathbf{b}: \mathbf{u}\in \mathcal{U}, \mathbf{v}\in \mathcal{V}\},
\end{align}
where $\mathcal{U}$ and $\mathcal{V}$ are subsets of $\ZZ_{p^r}^k$, and $\ZZ_{p^r}^l$, respectively. Define the inner code as 
\begin{align*}
\mathcal{C}_I \triangleq \{\mathbf{u}\mathbf{G}+\mathbf{b}: \mathbf{u}\in \mathcal{U}\}.
\end{align*}
By Definition \ref{def: QGC}, $\mathcal{C}_I$ is an $(n,k)$-QGC. In addition $\mathcal{C}_I \subset \mathcal{C}_O+\mathbf{b}$. The pair $(\mathcal{C}_I, \mathcal{C}_O)$ is called a nested QGC.  For any fixed element $\mathbf{v}\in \mathcal{V}$, we define its corresponding bin as the set 
\begin{align}\label{eq: bin for nested QGC}
\mathcal{B}(\mathbf{v})\triangleq \{\mathbf{u}\mathbf{G}+ \mathbf{v}\mathbf{\tilde{G}} +\mathbf{b}:  \mathbf{u}\in \mathcal{U}\}.
\end{align} 

\begin{definition}\label{def: nested QGC}
An $(n, k, l)$-nested QGC is defined as a pair $(\mathcal{C}_I, \mathcal{C}_O)$, where $\mathcal{C}_I$ is an $(n,k)$-QGC, and  $\mathcal{C}_O=\{\mathbf{x}_I+ \mathbf{\bar{x}}: \mathbf{x}_I \in \mathcal{C}_I, \mathbf{\bar{x}}\in \bar{\mathcal{C}}\},$ where $ \bar{\mathcal{C}}$ is an $(n,l)$-QGC. Let the random variables corresponding to $\mathcal{C}_I$ and $\bar{\mathcal{C}}$ are $(U,Q)$ and $(V,Q)$, respectively. Then, $\mathcal{C}_O$ is characterized by $(U,V, Q)$.   
\end{definition}
In a nested QGC both the outer-code and the inner code are themselves QGCs. More precisely we have the following remark.
\begin{remark}\label{rem: nested QGC are QGC}
Let $(\mathcal{C}_I,\mathcal{C}_O)$ be an $(n,k_1,k_2)$-nested QGC with random variables $(U_1,U_2,Q)$. Suppose the joint distribution among $(U_1,U_2,Q)$ is the one that satisfies  the Markov chain $U_1 \leftrightarrow Q  \leftrightarrow U_2$. Then by Lemma \ref{lem:sum of two quasi group code} $\mathcal{C}_O$ is an $(n, k_1+k_2)$-QGC with random variables $(U_I, (Q,I))$.
\end{remark}
\begin{remark}\label{rem: nested QGC rate}
Suppose $(\mathcal{C}_I, \mathcal{C}_O)$ is an $(n, k_1,k_2)$-nested QGC with random matrices and translations. By $R_O$ and $R_I$ denote the rates of $\mathcal{C}_O$ and $\mathcal{C}_I$, respectively. Let $\rho$ denote the rate of the $\bar{\mathcal{C}}$ associated with $(\mathcal{C}_I, \mathcal{C}_O)$ as in Definition \ref{def: nested QGC}. Using Remark \ref{rem: nested QGC are QGC} and \ref{rem: size of a random QGC}, for large enough $n$, with probability close to one, $|R_O-R_I-\rho| \leq o(\epsilon)$.
\end{remark}
Intuitively, as a result of this remark, $R_O\approx R_I + \rho$. This implies that the bins $\mathcal{B}(\mathbf{v})$ corresponding to different $\mathbf{v}\in \bar{\mathcal{C}}$ are ``almost disjoint". In this method for binning, since both the inner-code and the outer-code are QGCs, the structure of the inner-code, bins and the outer-code can be determined using the PMFs of the related random variables (that is $U, V$ and $Q$ as in the definition of nested QGCs). We show that nested QGCs improve upon the previously known schemes in certain multi-terminal problems. Such codes are also used to induce non-uniform distributions on the codewords, for instance, in PtP source coding as well as channel coding. In the following, it is shown that nested QGC achieve the Shannon performance limits for PtP channel and source coding problem.

\noindent{\bf Channel Model:}
A discrete memoryless channel is  characterized by the triple $(\mathcal{X},\mathcal{Y},P_{Y|X})$, where the two finite sets $\mathcal{X}$ and $\mathcal{Y}$ are the input and output alphabets, respectively, and $P_{Y|X}$ is the channel transition probability matrix.

\begin{definition}
An $(n,\Theta)$-code for a channel $(\mathcal{X},\mathcal{Y},P_{Y|X})$ is a pair of mappings $(e,f)$ where  $e:[1: \Theta]\to \mathcal{X}^{n}$ and $f:\mathcal{Y}^n \to [1:\Theta]$. 
\end{definition}

\begin{definition}
For a given channel $(\mathcal{X},\mathcal{Y},P_{Y|X})$, a  rate $R$ is said to be achievable  if  for any $\epsilon>0$ and for all sufficiently large $n$, there  exists an $(n,\Theta)$-code  such that :
\begin{align*}
&\frac{1}{\Theta}\sum_{i=1}^{\Theta}P_{Y|X}^n(f(Y^n)\neq i|X^n=e(i))< \epsilon, \quad \frac{1}{n}\log{\Theta} > R-\epsilon.
\end{align*}
The channel capacity is defined as the supremum of all achievable rates.
\end{definition}

\noindent{\bf Source Model:}
A discrete memoryless source is a tuple $(\mathcal{X},\hat{\mathcal{X}}, P_X,d)$, where the two finite sets $\mathcal{X}$ and $\hat{\mathcal{X}}$ are the source and reconstruction alphabets,  respectively, $P_X$ is the source probability distribution, and  $d:\mathcal{X}\times\hat{\mathcal{X}}\to \mathbb{R}^+$ is the (bounded) distortion function.

\begin{definition}
An $(n,\Theta)$-code for a source $(\mathcal{X},\hat{\mathcal{X}}, P_X,d)$ is a pair of mappings $(e,f)$ where   $f:\mathcal{X}^n \to [1:\Theta]$ and $e:[1: \Theta]\to \hat{\mathcal{X}}^{n}$.
\end{definition}

\begin{definition}
For a given source $(\mathcal{X},\hat{\mathcal{X}}, P_X,d)$, a  rate-distortion pair $(R,D)$ is said to be achievable  if  for any $\epsilon>0$ and for all sufficiently large $n$, there  exists an $(n,\Theta)$-code  such that :
\begin{align*}
&\frac{1}{n}\sum_{i=1}^{n}d(X_i,\hat{X}_i)< D+\epsilon, \quad \frac{1}{n}\log{\Theta} < R+\epsilon,
\end{align*}
where $\hat{X}^n=e(f(X^n))$.
The optimal rate-distortion region is defined as the set of all achievable rate-distortion pairs.
\end{definition}

\begin{definition}
An $(n, \Theta)$-code is said to be based on nested QGCs, if there exists an $(n, k, l)$-nested QGC with random variables  $(U,V,Q)$ such that a) $\Theta=|\mathcal{V}|$, where $\mathcal{V}$ is the index set associated with the codebook $\bar{\mathcal{C}}$ (see Definition \ref{def: nested QGC}), b) for any $\mathbf{v}\in \mathcal{V}$, the output of the mapping $e(\mathbf{v})$ is in  $\mathcal{B}(\mathbf{v})$, where $\mathcal{B}(\mathbf{v)}$ is the bin associated with $\mathbf{v}$, and is defined as in  \eqref{eq: bin for nested QGC}.
\end{definition}

\begin{definition}
For a channel, a rate $R$ is said to be achievable using nested QGCs if for any $\epsilon>0$ and all sufficiently large $n$, there exists an $(n,\Theta)$-code based on nested QGCs such that: 
\begin{align*}
&\frac{1}{\Theta}\sum_{i=1}^{\Theta}P(f(Y^n)\neq i|X^n=e(i))< \epsilon, \quad \frac{1}{n}\log{\Theta}> R-\epsilon.
\end{align*}
For a source, a  rate-distortion pair $(R,D)$ is said to be achievable using nested QGSs, if  for any $\epsilon>0$ and for all sufficiently large $n$, there  exists an $(n,\Theta)$-code based on nested QGCs such that:
\begin{align*}
&\frac{1}{n}\sum_{i=1}^{n}d(X_i,\hat{X}_i)< D+\epsilon, \quad \frac{1}{n}\log{\Theta} < R+\epsilon,
\end{align*}
where $\hat{X}^n=e(f(X^n))$.
\end{definition}

\begin{lem}
The PtP channel capacity and the optimal rate-distortion region of
sources are achievable using nested QGCs.
\end{lem}
 
 \begin{proof}[Outline of the proof]
Consider a memoryless channel with input alphabet $\mathcal{X}$ and conditional distribution $P_{Y|X}$. Let the prime power $p^r$ be such that $|\mathcal{X}| \leq p^r$. Fix a PMF $P_X$ on $\mathcal{X}$, and set $l={nR}$, were $R$ will be determined later. Let $(\mathcal{C}_I, \mathcal{C}_{O})$ be an $(n,k,l)$ nested QGC with random variables $(U,V,Q)$. Let $Q$ be a trivial random variable, and $U$ and $V$ be independent with uniform distribution over $\{0,1\}$. 

Suppose the messages are drawn randomly and uniformly from $\{0,1\}^l$. Upon receiving a message $\mathbf{v}$, the encoder first calculates its bin, that is $\mathbf{B}(\mathbf{v})$. Then it finds $\mathbf{x} \in \mathcal{B}(\mathbf{v})$ such that $\mathbf{x}\in A_\epsilon^{(n)}(X)$.  Then $\mathbf{x}$ is sent to the channel. Upon receiving $\mathbf{y}$ from the channel, the decoder finds all $\tilde{\mathbf{c}} \in \mathcal{C}_O$ such that $(\tilde{\mathbf{c}}, \mathbf{y})\in A_\epsilon^{(n)}(X,Y)$. Then, the decoder lists the bin number for any of such $\tilde{\mathbf{c}} $. If the bin number is unique, it is declared as the decoded message. Otherwise, an encoding error will be declared.   Note that the effective rate of transmission is $R$.


Let $R_{I}$ be the rate of $\mathcal{C}_I$. Then, using Lemma \ref{lem: covering}, the probability of the error at the encoder approaches zero, if $R_{I} \geq \log p^r-H(X)$. Using Lemma \ref{lem: packing}, we can show that the average probability of error at the decoder approaches zero, if $R_{I}+R \leq \log p^r -H(X|Y)$. As a result the rate $R \leq I(X;Y)$ is achievable. 
 
For the source coding problem, given a distortion level $D$, consider a random variable $\hat{X}$ such that $\EE\{d(X, \hat{X})\}\leq D$. 
Let $\mathbf{x}$ be a typical sequence from the source.  The encoder finds $\mathbf{c}\in \mathcal{C}_O$ such that $\mathbf{c}$ is jointly $\epsilon$-typical with $\mathbf{x}$ with respect to $P_XP_{\hat{X}|X}$. If no such $\mathbf{c}$ are found, an encoding error will be declared. Otherwise, the encoder finds $\textbf{v}$ for which $\mathbf{c}\in \mathcal{B}(\mathbf{v})$. Then, it sends $\textbf{v}$.  Given $\mathbf{v}$, the decoder finds $\tilde{\mathbf{c}}\in \mathcal{B}(\mathbf{v})$ such that  $\tilde{\mathbf{c}}$ is $\epsilon$-typical with respect to $P_{\hat{X}}$. An error occurs, if no unique codeword $\tilde{\mathbf{c}}$ is found. Using Lemma \ref{lem: covering}, it can be shown that the encoding error approaches zero, if $R+R_{in} \geq \log p^r-H(\hat{X}|X)$. Using Lemma \ref{lem: packing}, the decoding error approaches zero, if $R_{in} \leq \log p^r -H(\hat{X})$. As a result the rate $R \geq I(X;\hat{X})$ and distortion $D$ is achievable. 
\end{proof}

\section{Distributed Source Coding} \label{sec: dist}
In this section, we consider a special distributed source coding problem. Suppose $X_1$ and $X_2$ are sources over $\ZZ_{p^r}$ with joint PMF $P_{X_1X_2}$. The $j$th encoder compresses $X_j$ and sends it to a central decoder. The decoder wishes to reconstruct $X_1+ X_2$ losslessly. Figure \ref{fig: dist scr diagram} depicts the diagram of such a setup.

\begin{figure}[hbtp]
\centering
\includegraphics[scale=1.3]{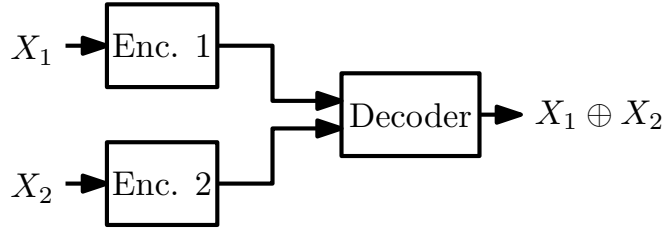}
\caption{An example for the problem of distributed source coding. In this setup, the sources $X_1$ and $X_2$ take values from $\ZZ_{p^r}$. The decoder reconstructs $X_1+ X_2$ losslessly. }
\label{fig: dist scr diagram}
\end{figure}

Consider a pair of sources with joint distribution $P_{XY}$ defined on $\ZZ_{p^r}\times \ZZ_{p^r}$.  The source sequences $(X^n, Y^n)$ are generated randomly and independently with the joint distribution $$P(\mathbf{X^n=x^n, X^n=y^n})= \prod _{i=1}^n P_{XY}(x_i,y_i).$$ 




\begin{definition}
 An $(n,\Theta_1,\Theta_2)$-code consists of two encoding functions 
\begin{align*}
&f_i: \ZZ_{p^r}^n \rightarrow \{1,2,\cdots, \Theta_i\},  \quad i=1,2,
\end{align*}

and a decoding function
\begin{equation*}
g: \{1,2,\cdots, \Theta_1\} \times \{1,2,\cdots, \Theta_2\} \rightarrow \ZZ_{p^r}^n
\end{equation*}
\end{definition} 

\begin{definition}
Given a pair of sources $(X_1, X_2) \sim P_{X_1X_2}$ with values over $\ZZ_{p^r}\times \ZZ_{p^r}$, a pair $(R_1,R_2)$ is said to be achievable if for any $\epsilon>0$ and sufficiently large $n$ , there exists an $(n,\Theta_1,\Theta_2)$-codes such that, 
\begin{align*}
\frac{1}{n}\log_2 M_i < R_i+\epsilon \quad for  \quad i=1,2,
\end{align*}
and
\begin{equation*}
 P\{\mathbf{X_1}^n+ \mathbf{X_2}^n\neq g(f_1(\mathbf{X_1}^n),f_2(\mathbf{X_2}^n))\} \leq \epsilon.
\end{equation*}

\end{definition}

For this problem, we use nested QGCs to propose a new coding scheme. We use two nested QGCs one for each encoder. The inner-codes are identical. 

\begin{theorem}\label{them: distributed source coding}
Given a pair of sources  $(X_1,X_2)\sim P_{X_1X_2}$ distributed over $\ZZ_{p^r}\times \ZZ_{p^r}$, the following rate-region is achievable
\begin{align}\label{eq: achievable bounds dist}
R_i \geq \log_2p^r- \min_{0\leq s\leq r-1} \frac{H(W_i|Q)}{H(W_1+ W_2|[W_1+ W_2]_sQ)} (\log_2p^{(r-s)}-H(X_1+ X_2|[X_1 + X_2]_s)),
\end{align}
where $i=1,2$, and $W_1, W_2$ take values from  $\ZZ_{p^r}$, and the Markov chain $W_1-Q-W_2$ holds. In addition, $|\mathcal{Q}|\leq r$ is sufficient to achieve the above bounds.
 \end{theorem}

\begin{proof}
Fix a positive integer $n$, and define $ l_1 \triangleq c_1 n, l_2\triangleq c_2n$, and $k\triangleq \tilde{c}n$, where $\tilde{c}, c_1$ and $c_2$ are positive constant real numbers. {  Let $\mathcal{C}_{I,1},\mathcal{C}_{I,2}$ and $\mathcal{C}_d$ be three $(n,k)$-QGC's (as in Definition \ref{def: QGC}) with identical matrices and translation. By $\mathbf{G}$ and $\mathbf{b}$ denote the generator matrix and translation, respectively.} 
The random variables associated with $\mathcal{C}_{I, 1}$ and $\mathcal{C}_{I, 2}$ are $(W_1, Q)$ and $(W_2, Q)$, respectively. { The random variable associated with $\mathcal{C}_d$ is $(W_1+W_2, Q)$.} Let $\bar{\mathcal{C}}_i$ be an $(n, l_i)$-QGC with random variables $(V_i, Q)$, where $i=1,2$. The random variable $V_i$ is uniform over $\{0,1\}$, and is independent of  $Q$.  The matrix used for $\bar{\mathcal{C}}_1$ is identical to the one used for $\bar{\mathcal{C}}_2$, and is denoted by $\bar{\mathbf{G}}$. The translation defined for $\bar{\mathcal{C}}_i$ is denoted by $\mathbf{\bar{b}}_i, i=1,2$. Suppose that the elements of   $\mathbf{G}, \bar{\mathbf{G}}, \mathbf{b}$, and $\mathbf{\bar{b}}_i, i=1,2$ are generated randomly and independently from $\ZZ_{p^r}$. Also, conditioned on $Q$ the random variables $W_1, W_2, V_1$, and $V_2$ are mutually independent.  By $R_i$ denote the rate of $\bar{\mathcal{C}}_i$, and let $R_{I, i}$ be the rate of $\mathcal{C}_{I,i}$, where $i=1,2$.
 
\paragraph*{\textbf{Codebook Generation}} 

%
We use two nested QGC's, one for each encoder. The codebook for the first encoder is $( \mathcal{C}_{I, 1}, \mathcal{C}_{O,1})$ which is an $(n,k, l_1)$ nested QGC (as in Definition  \ref{def: nested QGC}) that is characterized by $\mathcal{C}_{I,1}$ and $\bar{\mathcal{C}}_1$.  For the second encoder, we use $(\mathcal{C}_{I, 2}, \mathcal{C}_{O,2})$ which is an $(n,k, l_2)$ nested QGC  characterized by $\mathcal{C}_{I,2}$ and $\bar{\mathcal{C}}_2$ . With this notation, the random variables corresponding to $(\mathcal{C}_{I, i}, \mathcal{C}_{O,i})$ are $(W_i,V_i, Q), i=1,2$. {The codebook at the decoder is $\mathcal{C}_{d}$.}

\paragraph*{\textbf{Encoding}}
Suppose $\mathbf{x}_1$ and $\mathbf{x}_2$ are a IID realization of $(X^n_1,X^n_2)$. The first encoder checks if $\mathbf{x}_1$ is $\epsilon$-typical and $\mathbf{x}_1\in \mathcal{C}_{O, 1}$. If not, an encoding error $E_1$ is	declared. In the case of no encoding error, by Definition \ref{def: nested QGC},  $\mathbf{x}_1=\mathbf{c}_{I,1} + \bar{\mathbf{c}}_1$, where $\mathbf{c}_{I,1}\in \mathcal{C}_{I,1}$ and $\bar{\mathbf{c}}_1\in \bar{\mathcal{C}}_1$.  The first encoder sends the index of $\bar{\mathbf{c}}_1$. Note $\bar{\mathbf{c}}_1$ determines the index of the bin which contains $\mathbf{x}_1$.  Similarly, if $\mathbf{x}_2\in A_\epsilon^{(n)}(X_2)$ and $\mathbf{x}_2\in \mathcal{C}_{O,2}$,  the second encoder sends finds $\mathbf{c}_{I,2}\in \mathcal{C}_{I,2}$ and $\bar{\mathbf{c}}_2\in \bar{\mathcal{C}}_2$ such that $\mathbf{x}_2=\mathbf{c}_{I,2} + \bar{\mathbf{c}}_2$. Then it sends the index of   $\bar{\mathbf{c}}_2$. If no such $\mathbf{c}_{I, 2}$ and $\bar{\mathbf{c}}_2$ are found, an error event $E_2$ is declared.
%
%

\paragraph*{\textbf{Decoding}}
The decoder wishes to reconstruct $\mathbf{x}_1+\mathbf{x}_2$. Assume there is no encoding error.  Upon receiving the bin numbers from the encoders, the decoder calculates $\bar{\mathbf{c}}_1$ and $\bar{\mathbf{c}}_2$. Then, it finds $\tilde{\mathbf{c}}\in  \mathcal{C}_d$ such that $\tilde{\mathbf{c}}+\bar{\mathbf{c}}_1+\bar{\mathbf{c}}_2 \in A_{\epsilon}^{(n)}(X_1+ X_2)$. If $\tilde{\mathbf{c}}$ is unique, then $\tilde{\mathbf{c}}+\bar{\mathbf{c}}_1+\bar{\mathbf{c}}_2$ is declared as a reconstruction of $\mathbf{x}_1+\mathbf{x}_2$. An error event $E_d$ occurs, if no unique $\tilde{\mathbf{c}}$ was found. 

Using standard arguments for large enough $n$, the event that $\mathbf{x}_i$ is not $\epsilon$-typical is small. Next we use Lemma \ref{lem: covering} to bound $P(E_i), i=1,2$.  Note that the event $E_i$ is the same as the event of interest in Lemma \ref{lem: covering}, where $\hat{X}=X=X_i$, and $\mathcal{C}_{n}=\mathcal{C}_{O,i}$. In addition, by Remark  \ref{rem: nested QGC are QGC}, $\mathcal{C}_{O, i}$ is an $(n, k+l_i)$-QGC. Let $R_{O,i}$ denote  the rate of $\mathcal{C}_{O,i}$. By Remark \ref{rem: nested QGC rate}, with probability close to one, $|R_{O,i}-R_i-R_{l,i}|\leq o(\epsilon)$.  Therefore, applying Lemma \ref{lem: covering},  $P(E_i)\rightarrow 0$ if  (\ref{eq: covering bound}) holds for $R_n=R_i+R_{I, i}-o(\epsilon), i=1,2$.  Next we bound $P(E_d|E_1^c \cap E_2^c)$. Given $\bar{\mathbf{c}}_1$ and $\bar{\mathbf{c}}_2$, consider the codebook defined by ${ \mathcal{D}\triangleq \mathcal{C}_{d} +\bar{\mathbf{c}}_1 +\bar{\mathbf{c}}_2}$. 
We use Lemma \ref{lem: packing} to bound the probability of $E_d$ for fixed $\bar{\mathbf{c}}_1$ and $\bar{\mathbf{c}}_2$. Note that this event is the same as the event of interest in Lemma \ref{lem: packing}, where $Y$ is a trivial random variable, $ X=X_1+ X_2$, and $\mathcal{C}_n$ is replaced with $\mathcal{D}$. Therefore,  we can show that $P(E_d \cap E_1^c \cap E_2^c) \rightarrow 0$ as $n\rightarrow \infty$, if the bounds in (\ref{eq: packing bound}) are satisfied. Using the above argument, and noting that the effective transmission rate of the $i$th encoder is  $R_i$, we can derive the bounds in (\ref{eq: achievable bounds dist}). The cardinality bound on $\mathcal{Q}$ and the complete proof of the theorem are given in Appendix \ref{sec: proof dist}. 

\end{proof}


Every linear code,  group code and transversal group code is a QGC. Therefore, the achievable rate region of any coding scheme which uses these codes is included in the achievable rate region of that coding scheme using QGCs.  We show, through the following example, that the inclusion is strict.
%

\begin{example}\label{ex: dist src z_4}
Consider a distributed source coding problem in which $X_1$ and $X_2$ are sources over $\ZZ_4$ and lossless reconstruction of $X_1\oplus_4 X_2$ is required at the decoder.  Assume $X_1$ is uniform over $\ZZ_4$. $X_2$ is related to $X_1$ via the equation $X_2=N-X_1$, where $N$ is a random variable which is independent of $X_1$. The distribution of $N$ depends on a parameter denoted by $\delta_N $, where $0\leq \delta_N \leq 1$, and is presented in Table \ref{tab: N}.

\begin{table}[h]
\caption {Distribution of $N$}\label{tab: N}
\begin{center}
\begin{tabular}{|c|c|c|c|c|}
\hline
N & 0 & 1 & 2 & 3\\
\hline
$P_N$ & $0.1\delta_N$ & $0.9\delta_N$ & $0.1(1-\delta_N)$ & $0.9(1-\delta_N)$\\
\hline
\end{tabular}
\end{center}
\end{table}

Using random unstructured codes, the rates $(R_1,R_2)$ such that $R_1+R_2\geq H(X_1,X_2)$ are achievable \cite{Slepian-Wolf}. It is also possible to use linear codes for the reconstruction of $X_1\oplus_4 X_2$. For that, the decoder first reconstructs the modulo-$7$ sum of $X_1$ and $X_2$, then from $X_1\oplus_7 X_2$ the modulo-$4$ sum is retrieved. This is because linear codes are built only over finite fields, and $\ZZ_7$ is the smallest field in which the modulo-$4$ addition can be embedded. Therefore, the rates $R_1=R_2 \geq H(X_1\oplus_7 X_2)$ is achievable using linear codes over the field $\ZZ_7$ \cite{korner-marton}. As is shown in \cite{Aria_group_codes}, group codes in this example outperform linear codes. The largest achievable region using group codes is described by all rate pair $(R_1,R_2)$ such that $R_i \geq \max \{H(Z), 2 H(Z|[Z]_1)\}, ~ i=1,2, $ where $Z=X_1\oplus_4 X_2$. It is shown in \cite{ISIT15_transversal} that using transversal group codes the rates $(R_1,R_2)$ such that $R_i\geq \max \{H(Z), 1/2 H(Z)+ H(Z|[Z]_1)\}$ are achievable. An achievable rate region using nested QGC's can be obtained from Theorem \ref{them: distributed source coding}. Let $Q$ be a trivial random variable and set $P(W_1=0)=P(W_2=0)=0.95$ and  $P(W_1=1)=P(W_2=1)=0.05$. As a result one can verify that the following is achievable: 
$$R_j \geq 2- \min\{ 0.6(2-H(Z)), 5.7 (2-2H(Z|[Z]_1)\}.$$

We compare the achievable rates of these schemes for the case where $\delta_N=0.6$. The result are presented in Table \ref{tab: dist src. achievable rates }.

\begin{table}[th]
\caption {Achievable sum-rate using different coding schemes for Example \ref{ex: dist src z_4}. Note that $Z \triangleq X_1 \oplus_4 X_2$.}\label{tab: dist src. achievable rates }
\begin{center}
\begin{tabular}{|c|c|c|c|}
\hline
Scheme & Achievable Rate & $\delta_N=0.6$\\\hline
Unstructured Codes & $H(X_1,X_2)$ & $3.44$\\ \hline
Linear Codes   & $H(X_1\oplus_7 X_2)$ & $4.12$\\ \hline
Group Codes & $\max \{H(Z), 2 H(Z|[Z]_1)\}$ & $3.88$\\\hline
QGCs & $2- \min\{ 0.6(2-H(Z)), 5.7 (2-2H(Z|[Z]_1)\}$ & $3.34$\\\hline
\end{tabular}
\end{center}
\end{table}
\end{example}

\section{Computation Over MAC}\label{sec: comp_over_mac}

In this section, we consider the problem of computation over MAC. Figure \ref{fig: comp_over_mac} depicts an example of this problem. In this setup $X_1$ and $X_2$ are the channel's inputs, and take values from $\ZZ_{p^r}$.  Two distributed encoders map their messages to $X^n_1$ and $X^n_2$. Upon receiving the channel output the decoder wishes to decode $X^n_1+ X^n_2$ losslessly. The definition of a code for computation over MAC, and an achievable rate are given in Definition \ref{def: code for comp over MAC} and \ref{def: comp over MAC achievable rate}, respectively.  Applications of this problem are found in various multi-user communication setups such as interference and broadcast channels. 

\begin{figure}[ht]
\centering
\includegraphics[scale=1.2]{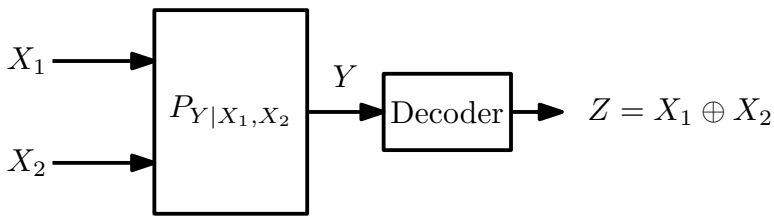} 
\caption{An example for the problem of computation over MAC. The channel input alphabets belong to $\ZZ_{p^r}$. The receiver decodes $X_1+ X_2$ which is the modulo-$p^r$ sum of the inputs of the MAC.}
\label{fig: comp_over_mac}
\end{figure}

\begin{definition}
A two-user MAC is a tuple $(\mathcal{X}_1, \mathcal{X}_2, \mathcal{Y}, P_{Y|X_1X_2})$, where the finite sets $\mathcal{X}_1, \mathcal{X}_2$ are  the inputs alphabets,  $\mathcal{Y}$ is the output alphabet, and  $P_{Y|X_1 X_2}$ is the channel transition probability matrix. Without loss of generality, it is assumed that $\mathcal{X}_1=\mathcal{X}_2=\ZZ_{p^r}$, for a prime-power $p^r$.
\end{definition}


\begin{definition} [Codes for computation over MAC]\label{def: code for comp over MAC}
 An $(n, \Theta_1, \Theta_2)$-code for computation over a MAC $(\ZZ_{p^r},\ZZ_{p^r}, \mathcal{Y}, P_{Y|X_1X_2})$ consists of two encoding functions and one decoding function $f_i:[1:\Theta_i]\rightarrow \ZZ_{p^r}^n$, for $i=1,2$, and $g: \mathcal{Y}^n \rightarrow \ZZ_{p^r}^n$, respectively.
\end{definition}

\begin{definition}[Achievable Rate] \label{def: comp over MAC achievable rate}
$(R_1,R_2)$ is said to be achievable, if for any $\epsilon >0$,  there exists for all sufficiently large $n$ an $(n, \Theta_1, \Theta_2)$-code such that  
\begin{align*}
&P\{g(Y^n)\neq f_1(M_1)+f_2(M_2)\}\leq \epsilon, \quad 
R_i-\epsilon \leq \frac{1}{n}\log \Theta_i,
\end{align*}
where $M_1$ and $M_2$ are independent random variables and $P(M_i=m_i)=\frac{1}{\Theta_i}$ for all $ m_i \in [1:\Theta_i], i=1,2$.
\end{definition}

For the above setup, we use QGCs to derive an achievable rate region.
\begin{theorem}\label{thm: comp MACnon uniform}
Given a MAC $(\ZZ_{p^r},\ZZ_{p^r}, \mathcal{Y}, P_{Y|X_1X_2})$, the following rate-region is achievable
\begin{align*}
R_i \leq \min_{0 \leq s \leq r} \frac{H(V_i |Q)}{H(V|[V]_s, Q)} \left( \log_2 p^{r-s}- H(X|Y[X]_s ) - \max_{\substack{1 \leq t \leq r\\j=0,1}} \frac{H(W|Q, [W]_s)}{H([W_j]_t|Q)} \left( \log_2p^t-H([X_j]_t) \right) \right)
\end{align*}
where $i=1,2$, $W=W_1+W_2, V=V_1+V_2, X=X_1+X_2$, and the joint PMF of the above random variables factors as $$P_{QX_1X_2V_1V_2W_1W_2Y} = P_{X_1}P_{X_2} P_Q P_{Y|X_1X_2}\prod_{i=1}^2 P_{V_i|Q} P_{W_i|Q} .$$
\end{theorem}

\begin{remark}
The cardinality bound $|\mathcal{Q}|\leq r^2$ is sufficient to achieve the rate region in the theorem.
\end{remark}

\begin{proof}[Outline of the proof]
Fix positive integer $n$, and define $ l \triangleq c n$, and $k\triangleq \tilde{c} n$, where $\tilde{c}$ and $c$ are positive constant real numbers. Let $\mathcal{C}_{I,1}$ and $\mathcal{C}_{I,2}$ be two $(n,k)$-QGC's with identical matrices and translations. 
The random variables defined for $\mathcal{C}_{I, 1}$ and $\mathcal{C}_{I, 2}$ are $(W_1, Q)$ and $(W_2, Q)$, respectively. Let $\bar{\mathcal{C}}_i$ be an $(n, l)$-QGC with random variables $(V_i, Q)$, where $i=1,2$.  The matrix used for $\bar{\mathcal{C}}_1$ is identical to the one used for $\bar{\mathcal{C}}_2$. The translation used by $\mathcal{C}_i$ is denoted by $\bar{\mathbf{b}}_i$.  
Suppose that the elements of the matrices and the translations are generated randomly and independently from $\ZZ_{p^r}$. Also, conditioned on $Q$ the random variables $W_1, W_2, V_1$, and $V_2$ are mutually independent.  By $R_i$ denote the rate of $\bar{\mathcal{C}}_i$, and let $R_{I, i}$ be the rate of $\mathcal{C}_{I,i}$, where $i=1,2$.

\textbf{Codebook Generation:}
We use two nested QGC's, one for each encoder. The codebook used for the $i$th encoder is $\mathcal{C}_{O,i}$ which is an $(n,k, l)$ nested QGC  characterized by $\mathcal{C}_{I,i}$ and $\bar{\mathcal{C}}_i$.  With this notation, the random variables corresponding to $\mathcal{C}_{O,i}$ are $(W_i,V_i, Q), i=1,2$. For the decoder, we use $\mathcal{C}_{O, 1}+\mathcal{C}_{O, 2}$ as a codebook.

\textbf{Encoding:} Index the codewords of $ \bar{ \mathcal{C}}_{i}, i=1,2$.  Upon receiving a message index $\theta_i$, the $i$th encoder finds the codeword $\mathbf{c}_i \in  \bar{ \mathcal{C}}_{i}$ with that index. Then it finds $\mathbf{c}_{I,i}\in \mathcal{C}_{I,i}$ such that $\mathbf{c}_i+\mathbf{c}_{I, i}$ is $\epsilon$-typical with respect to $P_{X_i}$.  If such codeword was found, the encoder $i$ sends $\mathbf{x}_i=\mathbf{c}_i+\mathbf{c}_{I, i}, i=1,2$. Otherwise, an error event $E_i, i=1,2$ is declared. 

\textbf{Decoding:}
The channel takes $\mathbf{x}_1$ and $\mathbf{x}_2$ and produces $\mathbf{y}$. Upon receiving $\mathbf{y}$ from the channel, the decoder wishes to decode $\mathbf{x}=\mathbf{x}_1+\mathbf{x}_2$. It finds $\tilde{\mathbf{x}} \in \mathcal{C}_{O, 1}+\mathcal{C}_{O, 2}$ such that $\tilde{\mathbf{x}}$ and $\mathbf{y}$ are jointly $\tilde{\epsilon}$-typical with respect to the distribution $P_{X_1+X_2, Y}$. An error event $E_d$ is declared, if no unique $\tilde{\mathbf{x}}$ was found.

Note that given the message the bin number $\mathbf{c}_i$ is determined. Then the encoder finds an $\epsilon$-typical codeword in the corresponding bin, i.e., $\mathcal{C}_{I,i}+\mathbf{c}_i$. Therefore, the inner-code $\mathcal{C}_{I, i}, i=1,2$ needs to be a ``good covering" code. We use Lemma \ref{lem: covering} to bound $P(E_i), i=1,2$.  Note that the event $E_i$ is the same as the event of interest in this lemma, where $X$ is trivial, $\hat{X}=X_i$, and $\mathcal{C}_{n}=\mathcal{C}_{I,i}+\mathbf{c}_i$. The rate of such code equals $R_{I, i}$. Therefore, $P(E_i)\rightarrow 0$ as $n\rightarrow \infty$, if  (\ref{eq: covering bound}) holds for $R_n= R_{I, i}, i=1,2$. Next, we find the conditions that $P(E_d)$ approaches zero as $n\rightarrow \infty$.  Note that using Lemma \ref{lem:sum of two quasi group code} the codebook defined by $\mathcal{C}_{O, 1}+ \mathcal{C}_{O, 2}$ is an $(n, k+l)$-QGC. We apply Lemma \ref{lem: packing} for $E_d$ and this codebook. In this lemma $X=X_1+X_2$, and $R_n$ is the rate of $\mathcal{C}_{O, 1}+ \mathcal{C}_{O, 2}$. Note that the effective rate of transmission for each encoder is $R_i, i=1,2$. Finally, we derive the bounds in the theorem using these covering and packing bounds, and the relation between the rate of $\mathcal{C}_{O, 1}+ \mathcal{C}_{O, 2}$ and $R_i, R_{I, i}, i=1,2$. The complete proof is provided in Appendix \ref{sec: proof of comp mac nonuniform}. 
\end{proof}

\begin{corollary}\label{cor: com_over_mac}
A special case of the theorem is when $X_1$ and $X_2$ are distributed uniformly over $\ZZ_{p^r}$. In this case, the following is achievable
\begin{align}\label{eq: comp_over MAC ahiev}
R_i \leq \min_{0 \leq s \leq r} \frac{H(V_i|Q)}{H(V_1+V_2|[V_1+V_2]_s, Q)} I(X_1+X_2;Y|[X_1+X_2]_s), \quad i=1,2,
\end{align}
\end{corollary}

We show, through the following example, that  QGC outperforms the previously known schemes.

%

\begin{example}\label{ex: comp_z4}
Consider the MAC described by $Y=X_1+ X_2 + N,$ where $X_1$ and $X_2$ are the channel inputs with alphabet $\ZZ_4$. $N$ is independent of $X_1$ and $X_2$ with the distribution given in Table \ref{tab: N}, where $0\leq \delta_N \leq 1$.

Using standard unstructured codes the rate pair $(R_1, R_2)$ satisfying $R_1+R_2\leq I(X_1X_2;Y)$ are achievable. { Note that the modulo-$4$ addition can be embedded in a larger field such as $\ZZ_7$.  For that linear codes over $\ZZ_7$ can be used. In this case, the following rates are achievable: 
\begin{align*}
R_1=R_2=\max_{P_{X_1}P_{X_2} : X_1, X_2 \in \ZZ_4 } \min\{H(X_1), H(X_2)\} -H(X_1\oplus_7 X_2 | Y),
\end{align*}
where the maximization is taken over all probability distribution $P_{X_1}P_{X_2}$ on $\ZZ_7 \times \ZZ_7$ such that $P(X_i \in \ZZ_4)=1,, i=1,2$. This is because, $\ZZ_4$ is the input alphabet of the channel.  

}
It is shown in \cite{Aria_group_codes} that the largest achievable region using group codes is  $$R_i \leq  \min \{  I(Z;Y), 2  I(Z;Y|[Z]_1)\},$$ where $Z=X_1+X_2$ and $X_1$ and $X_2$ are uniform over $\ZZ_4$.  
Using Corollary \ref{cor: com_over_mac}, QGC's achieve $R_i \leq \min \{ 0.6 I(Z;Y), 5.7 I(Z; Y| [Z]_1)\}.$ This can be verified by checking \eqref{eq: comp_over MAC ahiev} when $Q$ is a trivial random variable, $P(V_1=0)=P(V_2=0)=0.95$ and  $P(V_1=1)=P(V_2=1)=0.05$.
We compare the achievable rates of these schemes for the case where $\delta_N=0.6$. The result are presented in Table \ref{tab: comp MAC achievable rates }.

\begin{table}[th]
\caption {Achievable rates using different coding schemes for Example \ref{ex: comp_z4}. Note that $Z \triangleq X_1 +X_2$.}\label{tab: comp MAC achievable rates }
\begin{center}
\begin{tabular}{|c|c|c|c|}
\hline
Scheme & Achievable Rate $(R_1=R_2)$ & $\delta_N=0.6$\\\hline
Unstructured Codes & $ I(X_1X_2;Y)/2$ & $0.28$\\ \hline
Linear codes & $\min\{H(X_1), H(X_2)\} -H(X_1\oplus_7 X_2 | Y)$ & $0.079$\\ \hline
Group Codes & $\min \{  I(Z;Y), 2  I(Z;Y|[Z]_1)\}$ & $0.06$\\\hline
QGCs & $\min \{ 0.6 I(Z;Y), 5.7 I(Z; Y| [Z]_1)\}$ & $0.33$\\\hline
\end{tabular}
\end{center}
\end{table}

\end{example}

\section{MAC with States}\label{sec: MAC with States }
\subsection{Model}
Consider a two-user discrete memoryless  MAC with input alphabets $\mathcal{X}_1,  \mathcal{X}_2$, and output alphabet $\mathcal{Y}$. The transition probabilities between the input and the output of the channel depends on a random vector $(S_1, S_2)$ which is called state. Figure \ref{fig: MAC with states} demonstrates such setup.  Each state $S_i$ takes values from a set $\mathcal{S}_i$, where $i=1,2$. The sequence of the states is generated randomly according to the probability distribution $\prod_{i=1}^nP_{S_1S_2}$. The entire sequence of the state $S_i$ is known at the $i$th transmitter, $i=1,2$, non-causally. The conditional distribution of $Y$ given the inputs and the state is $P_{Y|X_1X_2S_1S_2}$.  
Each input $X_i$  is associated with a state dependent cost function $c_i:\mathcal{X}_i \times \mathcal{S}_i \rightarrow [0, +\infty)$\footnote{We use a cost function for this problem because, in many cases without a cost function the problem has a trivial solution.}. The cost associated with the sequences ${x}_i^n$ and $s_i^n$ is given by 
\begin{align*}
\bar{c}_i({x}_i^n, {s}_i^n)=\frac{1}{n}\sum_{j=1}^n c_i(x_{ij}, s_{ij}).
\end{align*}

\begin{figure}[hbtp]
\centering
\includegraphics[scale=1]{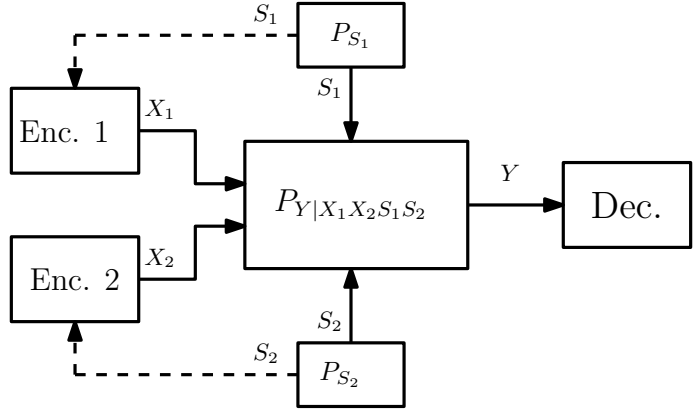}
\caption{A two-user MAC with distributed states. The states $(S_1, S_2)$ are generated randomly according to $P_{S_1S_2}$. The entire sequence of each state $S_i$ is available non-casually at the $i$th transmitter, where $ i=1,2$.}
\label{fig: MAC with states}
\end{figure}

\begin{definition}
An $(n, \Theta_1, \Theta_2)$-code for reliable communication over a given two-user MAC with states is defined by two encoding functions
\begin{align*}
f_i: \{1, 2, \dots, \Theta_i\} \times \mathcal{S}^n_i \rightarrow \mathcal{Y}^n, \quad i=1,2,
\end{align*}
and a decoding function 
\begin{align*}
g: \mathcal{Y}^n\rightarrow \{1, 2, \dots, \Theta_1\} \times \{1, 2, \dots, \Theta_2\} .
\end{align*}
\end{definition}

\begin{definition}\label{def: MAC with state achievable rate}
For a given MAC with state, the rate-cost tuple$(R_1,R_2, \tau_1, \tau_2)$ is said to be achievable, if for any $\epsilon >0$, and for all large enough $n$ there exist an $(n, \Theta_1, \Theta_2)$-code such that  
\begin{align*}
&P\{g(Y^n)\neq (M_1,M_2)\}\leq \epsilon, \quad 
\frac{1}{n}\log \Theta_i \geq R_i-\epsilon,\quad
\EE\{\bar{c}_i (f_i(M_i), {S}_i^n)\} \leq \tau_i+\epsilon,
\end{align*}
for $i=1,2$, where a) $M_1,M_2$ are independent random variables with distribution $P(M_i=m_i)=\frac{1}{\Theta_i}$ for all $ m_i \in [1:\Theta_i]$, b) $(M_1, M_2)$ is independent of the states $(S_1,S_2)$. Given $\tau_1, \tau_2$, the capacity region $\mathcal{C}_{\tau_1, \tau_2}$ is defined as the set of all rates $(R_1,R_2)$ such that the rate-cost $(R_1,R_2, \tau_1, \tau_2)$ is achievable.
\end{definition}

\subsection{Achievable Rates}
We  propose a structured coding scheme that builds upon QGC. Then we present the single-letter characterization of the achievable region of this coding scheme. 
Using this binning method, a rate region is given in the following theorem. 
\begin{theorem}\label{thm: QGC MAC with state achievable}
For a given MAC $(\mathcal{X}_1,\mathcal{X}_2, \mathcal{Y}, P_{Y|X_1X_2})$ with independent states $(S_1,S_2)$ and cost functions $c_1,c_2$ the following rates are achievable using nested-QGC 
\begin{align*}
R_1+R_2 &\leq r\log_2p - H(Z_1+ Z_2|YQ) - \max_{\substack{i=1,2\\ 1\leq t \leq r}}\Big\{\frac{H(V_1+ V_2|Q)}{H([V_i]_t|Q)} \Big(\log_2 p^t-H([Z_i]_t| Q S_i)\Big)\Big\},
\end{align*}
where the joint distribution of the above random variables factors as $$P_{S_1S_2}P_Q P_{Y|X_1X_2} \prod_{i=1,2} P_{V_i|Q} P_{Z_i|QS_i} P_{X_i|QZ_iS_i}.$$
\end{theorem}

\begin{proof}
Let $\mathcal{C}_{I,j}$ be an $(n,k)$-QGC with matrix $\mathbf{G}_j$, translation $\mathbf{b}_j$, and random variables $(W_j, Q)$, where $W_j$ is uniform over $\{0,1\}$, and $j=1,2$. Denote $\mathcal{W}_1$ and $\mathcal{W}_2$ as the index sets associated with $\mathcal{C}_{I,1}$ and $\mathcal{C}_{I,1}$, as in (\ref{eq: QGC codebook}). { Let $\bar{\mathcal{C}}_1, \bar{\mathcal{C}}_2$ and $\bar{\mathcal{D}}$ be three $(n,l)$ QGC with identical matrices $\mathbf{\bar{G}}$ and identical translations $\mathbf{\bar{b}}$. Suppose $(V_j, Q)$ are the random variables associated with $\bar{\mathcal{C}}_j$, where $j=1,2$. Furthermore, let $(V_1+ V_2, Q)$ is the random variable associated with $\bar{\mathcal{D}}$.}  Suppose that the elements of all the matrices and the translations are selected randomly and uniformly from $\ZZ_{p^r}$. Rate of $\bar{\mathcal{C}}_i$ is denoted by $\rho_i$, rate of $\bar{\mathcal{D}}$ is denoted by $\rho$, and that of $\mathcal{C}_{I, i}$ is $R_i, i=1,2$. For each, sequence $\mathbf{z}_i$ and $\mathbf{s}_i$, generate a sequence $\mathbf{x}_i$ randomly with IID distribution according to $P^n_{X_i|Z_iS_i}, i=1,2$. Denote such sequence by $x_i(\mathbf{s}_i, \mathbf{z}_i)$.

\textbf{Codebook Construction:}
For each encoder we use a nested QGC. For the first encoder, we use the $(n,k,l)$nested QGC generated by $\mathcal{C}_{I,1}$ and $\bar{\mathcal{C}}_1$. For the second encoder, we use the $(n,k,l)$nested QGC characterized by $\mathcal{C}_{I,2}$ and $\bar{\mathcal{C}}_2$. {  The codebook used in the decoder is $\mathcal{C}_{I, 1} + \mathcal{C}_{I, 2} + \mathcal{\bar{D}}$. By Lemma \ref{lem:sum of two quasi group code} , this codebook is an $(n, 2k+l)$-QGC.  In addition, the rate of such code is $R_1+R_2+\rho$  }

\textbf{Encoding:}
 For $i=1,2$, the $i$th encoder is given a message $\mathbf{\theta
}_i$, and an state sequence $\mathbf{s}_i$. The encoder first calculates the bin associated with $\mathbf{\theta}_i$. Then it finds a  codeword $\mathbf{z}_i$ in that bin such $(\mathbf{z}_i, \mathbf{s}_i)$ are jointly $\epsilon$-typical with respect to $P_{Z_iS_i}$.  If no such sequence was found, the error event $E_i$ will be declared. The encoder calculates $\mathbf{x}_i(\mathbf{s}_i, \mathbf{z}_i)$, and sends it through the channel. Define the event $E_c$ as the event in which $(\mathbf{Z}_1, \mathbf{Z}_2, \mathbf{s}_1, \mathbf{s}_2)$ are not jointly $\epsilon'$- typical with respect to the joint distribution $P_{Z_1Z_2S_1S_2}$. 
 
\textbf{Decoding:} 
 The decoder receives $y^n$ from the channel. Then it finds $\mathbf{\tilde{w}}_1 \in \mathcal{W}_1, \mathbf{\tilde{w}}_2 \in \mathcal{W}_2$, and ${  \mathbf{\tilde{v}}\in   A_\epsilon^{(n)}(V_1+ V_2)}$ such that the corresponding codeword defined as $$\mathbf{\tilde{z}}=\mathbf{\tilde{w}}_1\mathbf{G}_1+\mathbf{\tilde{w}}_2\mathbf{G}_2+\mathbf{\tilde{v}}\mathbf{\bar{G}}+\mathbf{b}_1+\mathbf{b}_2+\bar{\mathbf{b}}$$ is jointly $\tilde{\epsilon}$-typical with  $ \mathbf{Y}$ with respect to $P_{Z_1+ Z_2, Y}$. If $\mathbf{\tilde{w}}_1, \mathbf{\tilde{w}}_2$ are unique, then they are considered as the decoded messages. Otherwise an error event $E_d$ will be declared.
 
 \textbf{Error Analysis:}
 We use Lemma \ref{lem: covering} for $E_1$ and $E_2$. For that in the covering bound given in (\ref{eq: covering bound}) set $R=\rho_i, U=V_i, Q=\bar{Q}, \hat{X}=X_i$, and $X=S_i$, where $i=1,2$. As a result, $P(E_1)$ and $P(E_2) $ approaches zero as $n\rightarrow \infty$, if the covering bound holds:
\begin{align*}
  \rho_i > \max_{1\leq t \leq r} \frac{H(V_i|\bar{Q})}{H([V_i]_t|\bar{Q})} (\log_2 p^t-H([Z]_t|S_i)).
\end{align*}
Note that by Remark \ref{rem: size of a random QGC}, $\rho_i \leq \frac{l}{n} H(V_i|\bar{Q})+ \delta(\epsilon)$. Thus, the above bound gives the following bound
\begin{align}\label{eq: covering bound for encoder i}
 \frac{l}{n}  H([V_i]_t|\bar{Q}) > \log_2 p^t-H([Z]_t|S_i), ~ 1\leq t \leq r, ~ i=1,2.
\end{align}

{  
\paragraph{ \bf Analysis of $E_c \cap E^c_1 \cap E^c_2 $}
 Define the set 
\begin{align*}
\mathcal{E}_{\mathbf{s}_1, \mathbf{s}_2} \triangleq \Big\{   (\mathbf{z}_1, \mathbf{z}_2)\in \ZZ^n_{p^r} \times \ZZ_{p^r}^n :  ( \mathbf{z}_i,  \mathbf{s}_i)   \in A_{\epsilon}^{(n)}(Z_iS_i),  (\mathbf{z}_1, \mathbf{z}_2, \mathbf{s}_1, \mathbf{s}_2) \notin A_\epsilon^{(n)}(Z_1Z_2S_1S_2) , i=1,2   \Big \}.
\end{align*}
Therefore, probability of $E_c \cap E^c_1 \cap E^c_2$ can be written as  
\begin{align*}
P(E_c \cap E^c_1 \cap E^c_2 ) =  \sum_{(\mathbf{s}_1, \mathbf{s}_2) \in A_\epsilon^{(n)}(S_1, S_2)}P^n_{S_1,S_2}(\mathbf{s}_1, \mathbf{s}_2) \sum_{(\mathbf{z}_1, \mathbf{z}_2)\in \mathcal{E}_{\mathbf{s}_1, \mathbf{s}_2}} P( e_1(\Theta_1, \mathbf{s}_1)=\mathbf{x}_1,e_2(\Theta_2, \mathbf{s}_2)=\mathbf{x}_2 ),
\end{align*}
where $e_i$ is the output of the $i$th encoder, and $\Theta_i$ is the random message to be transmitted by encoder $i$, where $i=1,2$. To bound $P(E_c \cap E^c_1 \cap E^c_2)$, we use a similar argument as in the proof of Theorem \ref{thm: comp MACnon uniform}. We can show that, $\EE\{P(E_c \cap E^c_1 \cap E^c_2)\} \rightarrow 0$ as $n \rightarrow \infty$.

 \paragraph{ \bf Analysis of $E_d \cap (E_c \cup E_1 \cup E_2)^c $}

Next, we use Lemma \ref{lem: packing} to provide an upper-bound on $P(E_d \cap (E_c \cup E_1 \cup E_2)^c)$. Conditioned on $E_1^c\cap E_2^c$, the event $E_d$ is the same as the event of interest in Lemma \ref{lem: packing}. Set $\mathcal{C}_n=\mathcal{C}_{I, 1} + \mathcal{C}_{I, 2} + \mathcal{\bar{D}}$, and $R=R_1+R_2+\rho$. It can be shown that $P(E_d \cap (E_c \cup E_1 \cup E_2)^c)$ approaches zero, if the packing bound in (\ref{eq: packing bound}) holds. Since $W_i$ is uniform over $\{0,1\}$, then $H(W_i|Q, [W_i]_t)=0$ for all $t>0$. Therefore, the packing bound is simplified to 
\begin{align}\label{eq MAC st: packing}
R_1+R_2+\rho \leq \log_2p^r-H(Z_1+ Z_2|Y).
\end{align}
Note that  $\rho \leq \frac{l}{n} H(V_1+ V_2|{Q}) $. Therefore, if the bound  
\begin{align}\label{eq MAC st: Lower bound on R_1+R_2}
R_1+R_2 \leq \log_2p^r-H(Z_1+ Z_2|Y)-\frac{l}{n} H(V_1+ V_2|{Q}),
\end{align}
holds on $R_1+R_2$, then  \eqref{eq MAC st: packing} holds too.  Using \eqref{eq: covering bound for encoder i}, we establish a lower-bound on  $\frac{l}{n} H(V_1+ V_2|{Q})$. We have 
\begin{align}\label{eq MAC st: bound on H(V_1+V_2)}
 \frac{l}{n}  H(V_1+ V_2|{Q}) >  \frac{H(V_1+ V_2|{Q})}{ H([V_i]_t|\bar{Q}) } \left(\log_2 p^t-H([Z]_t|S_i),\right) ~ 1\leq t \leq r, ~ i=1,2.
\end{align}

Then combining \eqref{eq MAC st: Lower bound on R_1+R_2} and \eqref{eq MAC st: bound on H(V_1+V_2)} gives the following:
\begin{align*}
R_1+R_2 \leq \log_2p^r-H(Z_1+ Z_2|Y)- \frac{H(V_1+ V_2|{Q})}{ H([V_i]_t|\bar{Q}) } \left(\log_2 p^t-H([Z]_t|S_i) \right).
\end{align*}
Since these bounds hold for $i=1,2$, and $1\leq t \leq r$, we get the bound in the theorem. 
}
\end{proof}

\begin{corollary}
The rate region given in Theorem \ref{thm: QGC MAC with state achievable} contains the achievable rate region using group codes and linear codes. For that let $V_i, i=1,2$ be distributed uniformly over $\ZZ_{p^r}$. Therefore, we get the bound 
\begin{align*}
R_1+R_2 &\leq \min_{\substack{i=1,2\\ 1\leq t \leq r}}\{H([Z_i]_t| Q S_i)\} - H(Z_1+ Z_2|YQ) .
\end{align*}
\end{corollary}

Jafar \cite{Jafar-MAC-state} used the Gel'fand-Pinsker approach for the point-to-point channel coding with states, and proposed a coding scheme using unstructured random codes. Using this scheme a  single-letter and  computable rate region is characterized.


\begin{definition}\label{def:  mac with states gelfand pinsker achievable rate}
For a MAC $(\mathcal{X}_1,\mathcal{X}_2, \mathcal{Y}, P_{Y|X_1X_2})$ with  states $(S_1,S_2)$ and cost functions $c_1,c_2$, define $\mathscr{R}_{GP}$ as 
\begin{align}
\max \Big \{ I(U_1 U_2; Y|Q)-I(U_1;S_1|Q)-I(U_2;S_2|Q)\Big\},\label{eq: sum-rate}
\end{align}
where the maximization is taken over all joint probability distributions $P_{S_1S_2QU_1U_2X_1X_2Y}$ satisfying  $\EE\{c_i(X_i,S_i)\}\leq \tau_i$ for $i=1,2$, and factoring as   $$ P_Q P_{S_1S_2}P_{Y|X_1X_2} \prod_{i=1,2} P_{U_iX_i|S_iQ}.$$
The collection of all such PMFs $P_{S_1S_2QU_1U_2X_1X_2Y}$ is denoted by $\mathscr{P}_{GP}$.
\end{definition}

To the best of our knowledge, $\mathscr{R}_{GP}$ is the current largest achievable rate region using unstructured codes for the problem of MAC with states \cite{Jafar-MAC-state}. 

%

\subsection{An Example}
We present a MAC with state setup for which $\mathscr{R}_{GP}$ is strictly contained in the region characterized in Theorem \ref{thm: QGC MAC with state achievable}.
\begin{example}\label{ex: MAC states}
Consider a noiseless MAC given in the following $$Y=X_1\oplus_4 S_1 \oplus_4 X_2 \oplus_4 S_2,$$
where $X_1, X_2$ are the inputs, $Y$ is the output, and $S_1,S_2$ are the states. All the random variables take values from $\ZZ_4$. The states $S_1$ and $S_2$ are mutually independent, and are distributed uniformly over $\ZZ_4$. The cost function at the first encoder is defined as 
\begin{align*}
c_1(x)\triangleq \left\{
                \begin{array}{ll}
                  1 & if x\in \{1,3\} \\
                  0 & otherwise,
                  \end{array}
              \right.
\end{align*}
whereas, for the second encoder the cost function is 
\begin{align*}
c_2(x)\triangleq \left\{
                \begin{array}{ll}
                  1 & if x\in \{2,3\} \\
                  0 & otherwise.
                  \end{array}
              \right.
\end{align*}
We are interested in satisfying the  cost constraints $\EE\{c_1(X_1)\}=\EE\{c_2(X_2)\}=0$. This implies that, with probability one, $X_1\in \{0, 2\}$, and  $X_2\in \{0, 1\}$.
\end{example}

We proceed using two lemmas. 

\begin{lem}\label{lem: suboptimality of Gelfand-Pinsker}
For the setup in Example \ref{ex: MAC states},  an outer-bound for $\mathscr{R}_{GP}$  is the set of all rate pairs $(R_1, R_2)$ such that $R_1+R_2< 1$.
\end{lem}

\begin{proof}
See Appendix \ref{sec: proof of Lemma MAC with states}.
\end{proof}
Using numerical analysis, we can provide a tighter bound on the sum-rate which is $R_1+R_2 \leq 0.32$. However, the bound in Lemma \ref{lem: suboptimality of Gelfand-Pinsker} is sufficient  for the purpose of this paper.

\begin{lem}
For the MAC with states problem in Example  \ref{ex: MAC states}, the rate pairs $(R_1,R_2)$ satisfying $R_1+R_2= 1$ is achievable. 
\end{lem}
\begin{proof}
We use the proposed scheme presented in the proof of Theorem \ref{thm: QGC MAC with state achievable}. Similar to the proof of the Theorem, two $(n, k, l)$ nested QGC are used, one for each encoder.  Set $W_1$ and $W_2$, the random variables associated with the QGC, to be distributed uniformly over $\{0,1\}$. Suppose $\mathbf{v}_1, \mathbf{v}_2$ are the output of the nested-QGC at encoder 1 and encoder 2, respectively. Encoder 1 sends $\mathbf{x}_1=\mathbf{v}_1\ominus \mathbf{s}_1$, where $\mathbf{s}_1$ is the realization of the state $S_1$. Similarly, the second encoder sends $\mathbf{x}_2=\mathbf{v}_2\ominus \mathbf{s}_2$, where $\mathbf{s}_2$ is the realization of the state $S_2$. The conditional distribution of $v_1$ given $s_1$ is 
\begin{align*}
p(v_1|s_1)\triangleq \left\{
                \begin{array}{ll}
                  1/2 & \text{if} ~ v_1=-s_1, \text{or} ~ v_1=-s_1+ 2 \\
                  0 & \text{otherwise},
                  \end{array}
              \right.
\end{align*}
The distribution of $V_2$ conditioned of $S_2$ is
\begin{align*}
p(v_2|s_2)\triangleq \left\{
                \begin{array}{ll}
                  1/2 & \text{if} ~ v_2=-s_2, \text{or} ~ v_2=-s_2+ 1 \\
                  0 & \text{otherwise},
                  \end{array}
              \right.
\end{align*}
As a result, $X_1\in \{0,2\}, X_2\in \{0,1\}$. Hence, the cost constraints are satisfied. In this situation,  $H([V_i]_1)=H(V_i)=1,$ for $i=1,2$, and $H(V_1+ V_2)=\frac{3}{2}$.  Therefore, assuming $Q$ is trivial, the sum-rate given in the Theorem is simplified to 
\begin{align*}
R_1+R_2 &\leq   \frac{3}{2}\min\{H(V_1|S_1), H(V_2|S_2)\}\\& - H(V_1+ V_2|Y)-\frac{1}{2} =1,
\end{align*}
where the last equality holds, because $H(V_i|S_i)=1$, and $H(V_1+ V_2|Y)=H(X_1+ S_1+  X_2+ S_2|Y)=0$. As a result the sum-rate $R_1+R_2=1$ is achievable.
\end{proof}

%
%

\section{Conclusion}\label{sec: conclusion}
A new class of structured codes called Quasi Group Codes was introduced, and basic properties and performance limits of such codes were investigate. The asymptotic performance limits of QGCs was characterized using single-letter information quantities. The PtP channel capacity and optimal rate-distortion function are achievable using QGCs. coding strategies based on QGCs were studied for  three multi-terminal problems:   the K\"orner-Marton problem for modulo prime-power sums, computation over MAC, and MAC with States. For each problems, a coding scheme based on (nested) QGCs was introduced, and  a single-letter achievable rate-region was derived. The results show that the coding scheme improves upon coding strategies based on unstructured codes, linear codes and group codes. 
\appendices
\section{ }
\subsection{Proof of Lemma \ref{lem: size of U}}\label{APP: proof of lemma size of U}
\begin{proof}
Using \eqref{eq: U} we get $	\mathcal{U}_n = \bigotimes_{q\in \mathcal{Q}} A_{\epsilon}^{(k_{q,n})}(U_q) $, where $k_{q,n}=P_Q(q) k_n$, and the distribution of $U_q$ is the same as the conditional distribution of $U$ given $Q=q$. Using well-known results on the size of $\epsilon$-typical sets we can provide a bound on $|A_{\epsilon}^{(k_{q,n})}(U_q)|$. More precisely, there exists $N_{q}$ such that for all $k_{q,n}> c N_{q}$, we have  $  |\frac{1}{k_{q,n}} \log_2 |A_{\epsilon}^{(k_{q,n})}(U_q)|-H(U_q) |\leq 2\epsilon'_q$, where using the same argument as in \cite{Csiszar} $$\epsilon'_q = -\frac{\epsilon}{p^r} \sum_{a\in \ZZ_{p^r}, P(U_q=a)>0}\log_2 P(U_q=a).$$ Therefore,
\begin{align*}
\frac{1}{k_n}\log_2 |\mathcal{U}_n| &= \frac{1}{k_n}\sum_{q\in \mathcal{Q}} \log_2 |A_{\epsilon}^{(k_{q,n})}(U_q)|\\
&\leq \sum_{q\in \mathcal{Q}} \frac{k_{q,n}}{k_n} (H(U_q)+2\epsilon'_q)\\
& \stackrel{(a)}{=}H(U|Q)+\sum_{q\in \mathcal{Q}} P_Q(q) 2 \epsilon'_q\leq H(U|Q)+2\epsilon', 
\end{align*}
where $\epsilon'=2 \max_{q\in Q} \epsilon'_q$. Note $(a)$ holds as $P_Q(q)=k_{q,n}/k_n$. Using a similar argument we can show that $\frac{1}{k_n}\log_2 | \mathcal{U}_n| \geq H(U|Q)-\epsilon'$. Finally, by setting $N= \max_q N_q$, and combining the bounds on $\frac{1}{k_n}\log_2 |\mathcal{U}_n|$ the proof is completed.
\end{proof}

\subsection{Proof of Lemma \ref{lem: injective map}}\label{APP: proof of lemma injective map}
\begin{proof}
As $\mathbf{G}_n$ is a random matrix, then $\Phi_n$ is a randomly selected map. Fix an arbitrary $\mathbf{u}\in \mathcal{U}_n$. We have
\begin{align}\label{eq: p injective}
P\{\exists \mathbf{u}'\in \mathcal{U}_n: \mathbf{u'}\neq \mathbf{u}, ~ \Phi_n(\mathbf{u'})= \Phi_n(\mathbf{u})  \} & \leq \sum_{\substack{ \mathbf{u}' \in \mathcal{U}_n\\ \mathbf{u}'\neq \mathbf{u}}} P\{\Phi_n(\mathbf{u'})= \Phi_n(\mathbf{u})\},
\end{align}
where the inequality follows from the union bound. Let $H_s=p^s\ZZ_{p^r}$ be a subgroup of $\ZZ_{p^r}$, where $s\in [0:r-1]$. If $a\in \ZZ_{p^r} -\{0\}$, then there exits a maximum $s\in [0:r-1]$ such that $a\in H_s$. That is $a \in H_s$ and  $ a \notin  H_{t}$ for all $t>s$. As a result, for any $\mathbf{u'}\in \mathcal{U}_n$ there are $r$ cases for the maximum $s$ such that $u-u' \in H^{k_n}_s$. Considering these cases, we obtain
 \begin{align}\label{eq: p injective bound 1}
\sum_{\substack{ \mathbf{u}' \in \mathcal{U}_n\\ \mathbf{u}'\neq \mathbf{u}}} P\{\Phi_n(\mathbf{u'})= \Phi_n(\mathbf{u})\} &=\sum_{s=0}^{r-1} \sum_{\substack{ \mathbf{u}' \in \mathcal{U}_n \\ \mathbf{u}'-\mathbf{u} \in H^{k_n}_s \backslash H^{k_n}_{s+1} }} P\{\Phi_n(\mathbf{u'})= \Phi_n(\mathbf{u})\}
\end{align}
Since $\Phi_n$ is a linear map, we have $P\{\Phi_n(\mathbf{u}')=\Phi_n(\mathbf{u})\}=P\{\Phi_n(\mathbf{u}'-\mathbf{u})=0\}$. Next, we use Lemma \ref{lem: P(phi)} (see Appendix \ref{sec: useful lemmas}). Since $\mathbf{u'}-\mathbf{u}\in H^{k_n}_s \backslash H^{k_n}_{s+1}$, then $P\{\Phi_n(\mathbf{u}'-\mathbf{u})=0\}=p^{-n(r-s)}$. Therefore, using \eqref{eq: p injective} and \eqref{eq: p injective bound 1} we get
\begin{align}\label{eq: p injective bound 2}
P\{\exists \mathbf{u}'\neq \mathbf{u} : \Phi_n(\mathbf{u'})= \Phi_n(\mathbf{u})  \} & \leq  \sum_{s=0}^{r-1}\sum_{\substack{ \mathbf{u}' \in \mathcal{U}_n\\ \mathbf{u}'-\mathbf{u} \in H^{k_n}_s }} p^{-n(r-s)}
\end{align}
Next, we replace the summation over $\mathbf{u}'$ with the size of the set  $\mathcal{U}_n \bigcap (\mathbf{u} + H^{k_n}_s )$. Since $\mathcal{U}_n$ is a Cartesian product of typical sets, we use Lemma \ref{lem: typical set intersection subgroup} (see Appendix \ref{sec: useful lemmas}) to obtain the following bound  $$|\mathcal{U}_n \bigcap (\mathbf{u} + H^{k_n}_s )| \leq \prod_q 2^{k_{q,n}\left(H(U_q|[U_q]_s)+\epsilon'_q\right) },$$ where $k_{q,n}=P_Q(q)k_n$. Therefore the right-hand side of \eqref{eq: p injective bound 2} is bounded by  
\begin{align}\label{eq: p injective bound 3}
\sum_{s=0}^{r-1}\sum_{\substack{ \mathbf{u}' \in \mathcal{U}_n\\ \mathbf{u}'-\mathbf{u} \in H^{k_n}_s }} p^{-n(r-s)} & \leq \sum_{s=0}^{r-1} 2^{k_n\left(H(U|Q [U]_s)+ \epsilon'\right)}  p^{-n(r-s)}
\end{align}

By assumption of the lemma, suppose $\mathbf{U}\in \mathcal{U}_n$ is chosen randomly and uniformly. Then using \eqref{eq: p injective} and \eqref{eq: p injective bound 3} we obtain 
\begin{align}\label{eq: p injective bound 4}
\sum_{\mathbf{u}\in \mathcal{U}_n} \frac{1}{|\mathcal{U}_n|}P\{\exists \mathbf{u}'\in \mathcal{U}_n: \mathbf{u'}\neq \mathbf{U}, ~ \Phi_n(\mathbf{u'})= \Phi_n(\mathbf{u})  \}\leq \sum_{s=0}^{r-1} 2^{k_n\left(H(U|Q [U]_s)+ \epsilon'\right)}  p^{-n(r-s)}
\end{align}

Since by assumption,  $H(U|[U]_s,Q) < \frac{1}{c} (r-s) \log_2p, \forall s\in [0 : r-1]$, the right-hand side of \eqref{eq: p injective bound 4} approaches zero as $n\rightarrow \infty$. This implies that for any randomly selected $\mathbf{U} \in \mathcal{U}_n$, the size of the inverse image $|\Phi_n^{-1}(\Phi(\mathbf{U}))|=1$ with probability at least $(1-\delta)$.
\end{proof}

\section{Proof of Lemma \ref{lem: packing}} \label{sec: proof of the packing lemma}

\begin{proof}
Let $\mathcal{C}_n$ be the random $(n,k_n)$-QGC as in Lemma \ref{lem: packing}. For shorthand, for any $\mathbf{u} \in \mathcal{U}_n$, denote $\Phi_n(\mathbf{u})= \mathbf{u}\mathbf{G}_n$, where $\mathbf{G}_n$ is the random matrix corresponding to $\mathcal{C}_n$. Fix $\mathbf{u}_0 \in \mathcal{U}_n$. Without loss of generality assume $\mathbf{c}(\theta)=\Phi_n(\mathbf{u}_0)+B$, where $B$ is the translation associated with $\mathcal{C}_n$. Define the event $\mathcal{E}_n(\mathbf{u}):=\{(\Phi_n(\mathbf{u})+B,\tilde{\mathbf{Y}})\in A_\epsilon^{(n)}(X,Y)\}$,  and let $\mathcal{E}_n$ be the event of interest as given in the lemma. Then $\mathcal{E}_n$ is the union of $\mathcal{E}_n(\mathbf{u})$ for all $\mathbf{u}\in \mathcal{U}_n\backslash \{\mathbf{u}_0\}$. By the union bound, the probability of $\mathcal{E}_n$ is bounded as 
\begin{align}\label{eq: union bound on P(E)}
P(\mathcal{E}_n)\leq \sum_{\substack{\mathbf{u}\in \mathcal{U}_n\\ \mathbf{u}\neq \mathbf{u}_0}} P(\mathcal{E}_n(\mathbf{u}))
\end{align}

For any $\mathbf{u}\in \mathcal{U}_n$, the probability of  $\mathcal{E}_n(\mathbf{u})$, can be calculated as,
\begin{align}\label{eq: inner_probability}
P(\mathcal{E}_n(\mathbf{u})) &= \sum_{\mathbf{x}_0 \in \ZZ_{p^r}^n}  \sum_{\mathbf{y} \in  \mathcal{Y}^n }P(\Phi_n(\mathbf{u}_0)+B=\mathbf{x}_0,\tilde{\mathbf{Y}}=\mathbf{y}, \mathcal{E}_n(\mathbf{u}))\\\label{eq: inner_probability_eq2}
&= \sum_{\mathbf{x}_0 \in \ZZ_{p^r}^n}  \sum_{\mathbf{y} \in  A_{\epsilon}^{(n)}(Y) }  \sum_{\mathbf{x}: (\mathbf{x},\mathbf{y}) \in A_\epsilon^{(n)}(X,Y)} P(\Phi_n(\mathbf{u}_0)+B=\mathbf{x}_0,\tilde{\mathbf{Y}}=\mathbf{y}, \Phi_n(\mathbf{u})+B=\mathbf{x} ) 
\end{align} 
By assumption, conditioned on $\Phi_n(\mathbf{u}_0)+B$, the random variable $\tilde{\mathbf{Y}}$ is independent of $\Phi_n(\mathbf{u})+B$.  Therefore, the summand in (\ref{eq: inner_probability_eq2}) is simplified to 
\begin{equation}
P(\Phi_n(\mathbf{u}_0)+B=\mathbf{x}_0, \Phi_n(\mathbf{u})+B=\mathbf{x}) P^n_{Y|X}(\mathbf{y}|\mathbf{x}_0).
\end{equation}
Since $B$ is uniform  over $\ZZ_{p^r}^n$, and is independent of other random variables, 
\begin{equation}
P(\Phi_n(\mathbf{u}_0)+B=\mathbf{x}_0, \Phi_n(\mathbf{u})+B=\mathbf{x})=p^{-nr}P(\Phi_n(\mathbf{u}-\mathbf{u}_0)=\mathbf{x}-\mathbf{x}_0).
\end{equation}

Using Lemma \ref{lem: P(phi)}, if $\mathbf{u}-\mathbf{u}_0 \in H^{k_n}_s\backslash H^{k_n}_{s+1}$, then $P(\Phi_n(\mathbf{u}-\mathbf{u}_0)=\mathbf{x}-\mathbf{x}_0) =p^{-n(r-s)} \mathbbm{1}\{\mathbf{x}-\mathbf{x}_0 \in H^{k_n}_s\} $. Therefore, using \eqref{eq: inner_probability_eq2}, and for $\mathbf{u}-\mathbf{u}_0 \in H^{k_n}_s\backslash H^{k_n}_{s+1}$ we obtain

 \begin{align*}
P(\mathcal{E}_n(\mathbf{u}))&=  \sum_{\mathbf{x}_0 \in \ZZ_{p^r}^n}  \sum_{\mathbf{y} \in  A_{\epsilon}^{(n)}(Y) } \sum_{\substack{\mathbf{x}: \\ (\mathbf{x}, \mathbf{y}) \in A_\epsilon^{(n)}(X,Y) \\ \mathbf{x}-\mathbf{x}_0\in H_s^n}} p^{-nr}P^n_{Y|X}(\mathbf{y}|\mathbf{x}_0)p^{-n(r-s)}
\end{align*} 
{ Denote $\mathcal{A}\triangleq \{\mathbf{x}:  (\mathbf{x}, \mathbf{y}) \in A_\epsilon^{(n)}(X,Y),~ \mathbf{x}-\mathbf{x}_0 \in  H_s^n\}$. Note that if $([\mathbf{x}_0]_s, \mathbf{y}) \notin A_\epsilon^{(n)}([X]_s Y)$, then $\mathcal{A}=\emptyset$. Therefore,
 \begin{align}\label{eq: inner_probability_cont}
P(\mathcal{E}_n(\mathbf{u}))&=  \sum_{\substack{ (\mathbf{x}_0, \mathbf{y}):\\ ([\mathbf{x}_0]_s, \mathbf{y}) \in A_\epsilon^{(n)}([X]_sY)}} \sum_{\substack{\mathbf{x}\in \mathcal{A}}} p^{-nr}P^n_{Y|X}(\mathbf{y}|\mathbf{x}_0)p^{-n(r-s)}
\end{align} 

 Next, we replace the summation over $\mathbf{x}$ with the size of the set $\mathcal{A}$.  We bound the size of $\mathcal{A}$ using Lemma  \ref{lem: typical set intersection subgroup}. Therefore, an upper-bound on \eqref{eq: inner_probability_cont} is 
\begin{align}\nonumber          
P(\mathcal{E}_n(\mathbf{u})) &\leq \left(  \sum_{\substack{ (\mathbf{x}_0, \mathbf{y}):\\ ([\mathbf{x}_0]_s, \mathbf{y}) \in A_\epsilon^{(n)}([X]_sY)}} p^{-nr}P^n_{Y|X}(\mathbf{y}|\mathbf{x}_0) \right) p^{-n(r-s)} 2^{n\left( H(X|Y [X]_{s})+\delta(4\epsilon) \right)}\\
&\leq  \left(\sum_{\mathbf{x}_0 \in \ZZ_{p^r}^n}  \sum_{\mathbf{y} \in  \mathcal{Y}^n } p^{-nr}P^n_{Y|X}(\mathbf{y}|\mathbf{x}_0)\right) p^{-n(r-s)} 2^{n\left( H(X|Y [X]_{s})+\delta(4\epsilon) \right)}\\\label{eq: bound on E(u)}
& \leq p^{-n(r-s)} 2^{n \left( H(X|Y [X]_{s})+\delta(4\epsilon) \right) }.
\end{align} 
}
Note that if $\mathbf{a} \in \ZZ_{p^r}^k, \mathbf{a}\neq \mathbf{0}$ then there exists $s\in [0:r-1]$ such that $\mathbf{a}\in H_s^k \backslash H_{s+1}^k$. Therefore, there are $r$ different cases for each value of $s$.  Using (\ref{eq: bound on E(u)}), and considering these cases, we obtain 

\begin{align*}
P(\mathcal{E}_n) & \leq \sum_{s=0}^{r-1}\sum_{\substack{\mathbf{u}\in \mathcal{U}_n\\ \mathbf{u}-\mathbf{u}_0 \in  H^{k_n}_s \backslash H^{k_n}_{s+1}  }} P(\mathcal{E}_n(\mathbf{u})) \leq \sum_{s=0}^{r-1}\sum_{\substack{\mathbf{u}\in \mathcal{U}_n\\ \mathbf{u}-\mathbf{u}_0 \in  H^{k_n}_s \backslash H^{k_n}_{s+1}  }}  2^{n(H(X|Y[X]_s)+\delta(4\epsilon))}p^{-n(r-s)}\\
&\leq  \sum_{s=0}^{r-1}|\mathcal{U}_n\bigcap (\mathbf{u}_0+H_s^k)| 2^{n(H(X|Y[X]_s)+\delta(4\epsilon))}p^{-n(r-s)}
\end{align*}  

Note that  $\mathcal{U}_n$ is the Cartesian product of $\epsilon$-typical sets $A_{\epsilon}^{(p(q)k_n)}(U_q)$, $q\in \mathcal{Q}$. For each component $q$ of $\mathcal{U}_n$, we can apply Lemma \ref{lem: typical set intersection subgroup}. Therefore,
$$|\mathcal{U}_n\cap (\mathbf{u}_0+H_s^k)|\leq 2^{\sum_q p(q)k_n (H(U_q|[U_q]_s)+\delta(2 \epsilon))}=2^{k_n (H(U|Q[U]_s)+\delta(2\epsilon))}.$$
Finally, 
\begin{align*}
P(\mathcal{E}_n) & \leq \sum_{s=0}^{r-1}2^{n\big( \frac{k_n}{n}(H(U|Q[U]_s) +H(X|Y[X]_s) + \frac{k_n}{n}\delta(2\epsilon)+\delta(4\epsilon)\big)}p^{-n(r-s)}
\end{align*}

As a result $\lim_{n\rightarrow \infty}P(\mathcal{E}_n)= 0$, if the inequality $$c H(U|Q[U]_s) \leq \log_2p^{r-s}-H(X|Y[X]_s) -2(2+c)\delta(\epsilon),$$ holds for all $0\leq s\leq r-1$.  
Multiply each side of this inequality by $\frac{H(U|Q)}{H(U|Q[U]_s)}$. This gives the following bound
$$cH(U|Q) \leq \frac{H(U|Q)}{H(U|Q[U]_s)}(\log_2p^{r-s}-H(X|Y[X]_s)-2(2+c)\delta(\epsilon))$$  
By definition $R_n=\frac{1}{n}\log_2|\mathcal{C}_n|\leq c H(U|Q)+\epsilon'$. Therefore, 
$$R_n \leq \frac{H(U|Q)}{H(U|Q[U]_s)}(\log_2p^{r-s}-H(X|Y[X]_s)-2(2+c)\delta(\epsilon)),$$ and the proof is completed.
\end{proof}

\section{Proof of Lemma \ref{lem: covering}}\label{sec: proof of the covering lemma}
\begin{proof}
We use the same notation as in the proof of Lemma \ref{lem: packing}. For any typical sequence $\mathbf{x}$ define
\begin{align*}
\lambda_n(\mathbf{x})= \sum_{\mathbf{\hat{x}}\in A_\epsilon^{(n)}(\hat{X}|\mathbf{x})}  \sum_{\mathbf{u} \in \mathcal{U}_n}  \mathbbm{1}\{\Phi_n(\mathbf{u})+B=\hat{x}\}.
\end{align*}
Note $\lambda_n(\mathbf{x})$	counts the number of codewords that are conditionally typical with $\mathbf{x}$ with respect to $p(\hat{\mathbf{x}}| \mathbf{x})$. We show that $\lim_{n\rightarrow \infty}P(\lambda_n(\mathbf{x})=0)=0$ for any $\epsilon$-typical sequence $\mathbf{x}$. This implies that $\lim_{n\rightarrow \infty}P(\lambda_n(\mathbf{X}^n)=0)=0$,  where $\mathbf{X}^n\sim \prod_{i=1}^n p(x)$. This proves the statements of the Lemma. Hence, it suffices to show that $\lim_{n\rightarrow \infty }P(\lambda_n(\mathbf{x})=0)=0$. We have,
\begin{align*}
P\{\lambda_n(\mathbf{x})=0\} \leq  P\Big\{\lambda_n(\mathbf{x})\leq \frac{1}{2} E(\lambda_n(x))\Big\}  \leq P\Big\{|\lambda_n(x)-E(\lambda_n(x)) |\geq \frac{1}{2} E(\lambda_n(x))\Big\}
\end{align*}
Hence, by Chebyshev's inequality, $P\{\lambda_n(\mathbf{x})=0\}  \leq \frac{4 Var(\lambda_n(x))}{E(\lambda_n(x))^2}$.
Note that 
\begin{align} \label{eq: expectation_delta1}
E(\lambda_n(x))=\sum_{\mathbf{\hat{x}}\in A_\epsilon^{(n)}(\hat{X}|\mathbf{x})}  \sum_{\mathbf{u} \in \mathcal{U}_n}     P\{\Phi(\mathbf{u})+B=\hat{\mathbf{x}}\}
\end{align}
Since $B$ is uniform over $\ZZ_{p^r}^n$, we get
\begin{align} \label{eq: expectation_delta2}
E(\lambda_n(x))=  |A_\epsilon^{(n)}(X|\mathbf{\hat{x}})| | \mathcal{U}_n| p^{-rn}.
\end{align}
Note   $2^{k_n(H(U|Q)- 2\epsilon')}\leq |\mathcal{U}_n| \leq  2^{k_n(H(U|Q)+2\epsilon')}$, where $$\epsilon' =  -  \frac{\epsilon}{p^r} \sum_{q \in \mathcal{Q}}P_Q(q)  \sum_{ \substack{  a \in \ZZ_{p^r}: P_{U|Q}(a|q)>0 }}\log P_{U|Q}(a|q).$$  Therefore,
\begin{align} \label{eq: expectation_delta}
                        2^{n(H(\hat{X}|X)- 2\tilde{\epsilon})} 2^{k_n(H(U|Q)-2\epsilon')} p^{-rn}   \leq   E(\lambda_n(x))\leq  2^{n(H(\hat{X}|X)+2\tilde{\epsilon})} 2^{k_n(H(U|Q)+2\epsilon')} p^{-rn},
\end{align}
To calculate the variance, we start with
\begin{align*}
E(\lambda_n(x)^2)&= \sum_{\mathbf{\hat{x},\hat{x}'}\in A_\epsilon^{(n)}(\hat{X}|\mathbf{x})} \sum_{\mathbf{u}, {\mathbf{u}'} \in \mathcal{U}_n}    P\{\Phi(\mathbf{u})+B=\mathbf{\hat{x}}, \Phi({\mathbf{u}'})+B=\mathbf{\hat{x}'}\}.
\end{align*}

Since $B$ is independent of other random variables, the most inner term in the above summations is simplified to $p^{-nr}P\{\Phi(\mathbf{u}-\mathbf{u'})=\mathbf{\hat{x}}-\mathbf{\hat{x}'}\}$. Using Lemma \ref{lem: P(phi)}, if $\mathbf{u}-\mathbf{u'} \in H_s^{k_n}\backslash H_{s+1}^{k_n}$, then
\begin{align*}
P\{\Phi(\mathbf{u}-\mathbf{u'})=\mathbf{\hat{x}}-\mathbf{\hat{x}'}\}=p^{-n(r-s)}\mathbbm{1}\{\mathbf{\hat{x}}-\mathbf{\hat{x}'} \in H_s^n\}
\end{align*}
Considering all the cases for the values of $s$, we get
\begin{align*}
E(\lambda_n(x)^2)&= \sum_{s=0}^{r} \sum_{\substack{\mathbf{u}, {\mathbf{u}'} \in \mathcal{U}_n\\ \mathbf{u}-{\mathbf{u}'}\in H^{k_n}_{s}\backslash H_{s+1}^{k_n}}}\sum_{\substack{\mathbf{\hat{x},\hat{x}'}\in A_\epsilon^{(n)}(\hat{X}|\mathbf{x})\\ \mathbf{\hat{x}}-\mathbf{\hat{x}'} \in H_s^n}}   p^{-nr}p^{-n(r-s)}
\end{align*}
 Since the innermost terms in the above summations do not depend on the individual values of $\mathbf{x}, \hat{\mathbf{x}}, \mathbf{u}, {\mathbf{u}'}$, the corresponding summations can be replaced by the size of the associated sets. Moreover, we provide an upperbound on the summation over $\mathbf{u}, {\mathbf{u}'}$ by replacing $H_s^{k_n}\backslash H_{s+1}^{k_n}$ with $H_s^{k_n}$. Using Lemma \ref{lem: typical set intersection subgroup} for $\mathbf{x}, \hat{\mathbf{x}}$, we get
\begin{align*}
E(\lambda_n(x)^2)&\leq  \sum_{s=0}^{r} \sum_{\mathbf{u} \in \mathcal{U}_n} \sum_{\substack{{\mathbf{u}'} \in \mathcal{U}_n\\ \mathbf{u}-{\mathbf{u}'}\in H^{k_n}_{s}}}2^{n(H(\hat{X}|X)+\tilde{\epsilon}+H(\hat{X}|X[\hat{X}]_s)+\delta(4\epsilon))} p^{-nr}p^{-n(r-s)}
\end{align*}

%

For any $\mathbf{u}\in \mathcal{U}_n$, by applying Lemma \ref{lem: typical set intersection subgroup} we get $|\mathcal{U}_n \bigcap (\mathbf{u}+H_s^{k_n})| \leq 2^{{k_n}(H(U|Q[U]_s)+\delta(4\epsilon))}.$
 As a result, 
\begin{align*}
E(\lambda_n(x)^2)&\leq  \sum_{s=0}^{r}   2^{{k_n}(H(U|Q[U]_s)+\delta(4\epsilon))} 2^{{k_n}(H(U|Q)+\epsilon')} 2^{n(H(\hat{X}|X)+\tilde{\epsilon}+H(\hat{X}|X[\hat{X}]_s)+\delta(4\epsilon))} p^{-nr}p^{-n(r-s)}.
\end{align*}
Note that the case $s=0$ gives $E^2(\lambda_n(x))$. Therefore,

\begin{align} \label{eq: variance_delta}
Var(\lambda_n(x)^2)& \leq  p^{-nr} \sum_{s=1}^{r}   2^{{k_n}(H(U|Q)+H(U|Q[U]_s)) } 2^{n(H(\hat{X}|X)+H(\hat{X}|X[\hat{X}]_s))} 2^{n(1+c)(\epsilon+\delta(4\epsilon))}   p^{-n(r-s)}
\end{align}

Finally, using \eqref{eq: expectation_delta}, \eqref{eq: variance_delta} and  the Chebyshev's inequality as argued before, we get

\begin{align*}
P\{\lambda_n(\mathbf{x})=0\}&\leq 4 \sum_{s=1}^r 2^{{k_n}(-H(U|Q)+H(U|Q[U]_s)) }2^{n(-H(\hat{X}|X)+H(\hat{X}|X[\hat{X}]_s))}2^{n(1+c)(\epsilon+\delta(4\epsilon))}  p^{nr}p^{-n(r-s)} \\
&= 4 ~ 2^{n(1+c)(\epsilon+\delta(4\epsilon))} \sum_{s=1}^r 2^{-k_nH([U]_s|Q)} 2^{-nH([\hat{X}]_s|X)}  p^{ns}.
\end{align*}
The second equality follows, because $H(V|W)-H(V|[V]_sW)=H([V]_s|W)$ holds for any random variables $V$ and $W$. Therefore, $P\{\lambda_n(\mathbf{x})\}$ approaches zero, as $n\rightarrow \infty$, if 
\begin{align*}
c H([U]_{s}|Q) \geq  \log_2 p^s -H([\hat{X}]_s|X)+(1+c)(\epsilon+\delta(4\epsilon)), \quad  \mbox{for} ~~1\leq s\leq r.
\end{align*}
By the definition of rate and the above inequalities the proof is completed.

\end{proof}


\section{ Proof of Theorem \ref{them: distributed source coding}} \label{sec: proof dist}
We need to find conditions for which the probability of the error events $E_1, E_2$ and $E_d$ approach zero. By $\mathcal{W}_i$ denote the index set of $\mathcal{C}_{I, i}$, and let $\mathcal{V}_i$ be the index set of $\bar{\mathcal{C}}_{i}, i=1,2$.  
%
\subsection{Analysis of $E_1, E_2$}
Fix $\mathbf{G}$, $\mathbf{\bar{G}}, \mathbf{b}$ and $\mathbf{\bar{b}}_i$. For any sequence $\mathbf{x}_i \in \ZZ_{p^r}^n$ define $$\lambda_i(\mathbf{x}_i)=\sum_{\mathbf{w}_i \in \mathcal{W}_i }\sum_{\mathbf{v}_i \in \mathcal{V}_i}  \mathbbm{1}\{\mathbf{x}_i=\mathbf{w}_i\mathbf{G}+\mathbf{v}_i\mathbf{\bar{G}}+\mathbf{b}+\mathbf{\bar{b}}_i\},$$ where $i=1,2$. Therefore, $E_i$ occurs if $\lambda_i(x_i)=0$, where $(\mathbf{x}_1, \mathbf{x}_2)$ is a realization of the sources. For more convenience, we consider a superset of the event  $E_i$. We say $E'_i$ occurs, if $\lambda_i(\mathbf{x}_i)< \frac{1}{2}E(\lambda_i(x_i))$. We  show that $P(E_i)\rightarrow 0$ as $n\rightarrow \infty$. Note that $\mathcal{C}_{O, i}$ is the $(n,k,l_i)$-nested QGC characterized by $\mathcal{C}_{I, i}$ and $\bar{\mathcal{C}}_i$. By Lemma \ref{lem:sum of two quasi group code}, $(\mathcal{C}_{I, i}, \mathcal{C}_{O, i})$ is also an $(n, k+l_i)$-QGC. In addition, similar to the random variables in this lemma, the random variables defined for $\mathcal{C}_{O, i}$ are $(U_i, (Q,J_i))$, where given $J_i=1$ we have $U_i=W_i$, and given $J_i=2$ we get $U_i=V_i$. In addition, $P(J_i=0)=\frac{k}{l_i+k}$, and $P(J_i=1)=\frac{l_i}{l_i+k}$.  We apply Lemma \ref{lem: covering} to bound the probability of $E_i$. In this lemma set $\hat{X}=X=X_i$ with probability one, $\mathcal{C}_{n}=\mathcal{C}_{O, i}$, and $R_n=R_{O, i}, i=1,2$.  Therefore, $P(E'_i)\rightarrow 0$ as $n\rightarrow \infty$, If 
\begin{align*}
R_{O,i} \geq \max_{1 \leq s \leq r} \frac{ H(U_i|Q,J_i)}{ H([U_i]_s|Q,J_i)}  ( \log_2 p^s +o(\epsilon)). 
\end{align*}

Using Remark \ref{rem: size of a random QGC}, and the above bound we get $\frac{k+l_i}{n} H([U_i]_s|Q,J_i) \geq  \log_2 p^s +o(\epsilon)$ for $s\in [1:r]$. Therefore, by the definition of $U_i$ and $J_i$, we get
\begin{align*}
\frac{k}{n} H([W_i]_s|Q)+\frac{l_i}{n}H(V_i|Q) \geq \log_2 p^s +o(\epsilon), ~ 1 \leq s \leq r.
\end{align*}
Note that in this bound we use the equality $H([V_i]_s)=H(V_i)$. This equality holds because $V_i $ takes values from $\{0,1\}$. Again using Remark \ref{rem: size of a random QGC}, we get $|R_i - \frac{l_i}{n}H(V_i|Q)|\leq o(\epsilon)$.  Hence, if the following holds
\begin{align}\label{eq: covering_dist}
\frac{k}{n} H([W_i]_s|Q)+R_i \geq \log_2 p^s +o(\epsilon), ~ 1 \leq s \leq r, ~ i=1,2,
\end{align}
then $P(E'_i)\rightarrow 0$ as $n\rightarrow \infty$.
%

\subsection{Analysis  of $E_d$}
Suppose there is no error in the encoding stage. Upon receiving the bin numbers, the decoder calculates $\mathbf{\bar{c}}_1$ and $\mathbf{\bar{c}}_2$. The decoding error $E_d$ occurs, if there exist more than one  $\tilde{\mathbf{c}} \in \mathcal{C}_{I, 1}+\mathcal{C}_{I, 2}$ such that  $\tilde{\mathbf{c}}+\mathbf{\bar{c}}_1+\mathbf{\bar{c}}_2$  is $\epsilon$-typical with respect to $P_{X_1+X_2}$. 

Since there is no error at the encoding stage, $\mathbf{x}_i \in \mathcal{C}_{O, i}, i=1,2$. By Definition \ref{def: nested QGC}, every codeword in $\mathcal{C}_{O, i}$ is characterized by a pair $(\mathbf{v}_i, \mathbf{w}_i)$, where $\mathbf{v}_i\in \mathcal{V}_i, \mathbf{w}_i \in \mathcal{W}_i, i=1,2$. Given $\mathbf{x}_i$, if more than one pair was found at the $i$th encoder, select one randomly and uniformly.  By $P(\mathbf{v}_i, \mathbf{w}_i | \mathbf{x}_i)$ denote the probability that $(\mathbf{v}_i, \mathbf{w}_i)$ is selected at the $i$th encoder.  Then, $P(\mathbf{v}_i, \mathbf{w}_i|\mathbf{x_i})=\frac{1}{\lambda_i(\mathbf{x}_i)}\mathbbm{1}\{\mathbf{w}_{ i}\mathbf{G}+\mathbf{v}_i\mathbf{\bar{G}} + \mathbf{b}+\mathbf{\bar{b}}_i=\mathbf{x}_i	\}.$ Fix $\mathbf{G}$, $\mathbf{\tilde{G}}_i, \mathbf{b}$ and $\mathbf{\bar{b}}_i, i=1,2$. Suppose $\mathbf{x}_1$ and $\mathbf{x}_2$ are the realizations of the sources $X_1$ and $X_2$, respectively. Moreover, suppose $(\mathbf{x}_1, \mathbf{x}_2)\in A_\epsilon^{(n)}(X_1,X_2)$. Therefore, the probability of $(E_d \cap E_1^c \cap E_2^c)$ equals
\begin{align*}
P(E_d \cap & E_1^c \cap E_2^c|\mathbf{x}_1, \mathbf{x}_2)=\\
&\mathbbm{1}\Big\{\lambda_i(\mathbf{x_i})\geq E(\lambda_i(\mathbf{x}_i)), i=1,2  \Big\}  \left[ \prod_{j=1}^{2}   \sum_{\mathbf{v}_j \in \mathcal{V}_j} \sum_{\mathbf{w}_{ j} \in \mathcal{	W}_j} P(\mathbf{v}_j, \mathbf{w}_{ j}| \mathbf{x}_j) \right] P(E_d| \mathbf{x}_i, \mathbf{v}_i, \mathbf{w}_i, i=1,2)
\end{align*}

 In what follows, we bound $P(E_d| \mathbf{x}_i, \mathbf{v}_i, \mathbf{w}_i, i=1,2), P(\mathbf{v}_1, \mathbf{w}_{ 1}| \mathbf{x}_1)$, and $P(\mathbf{v}_2, \mathbf{w}_{ 2}| \mathbf{x}_2)$. Conditioned on $\mathbf{x}_1, \mathbf{x}_2, \bar{\mathbf{c}}_1$ and $\bar{\mathbf{c}}_2$, the probability of $E_d$ equals 
\begin{align*}
P(E_d| \mathbf{x}_1, \mathbf{x}_2, \bar{\mathbf{c}}_1,\bar{\mathbf{c}}_2)=\mathbbm{1}\{\exists \tilde{z}\in A_\epsilon^{(n)}(X_1+X_2): \tilde{z}\neq \mathbf{x}_1+\mathbf{x}_2, \tilde{z}\in \mathcal{C}_{I, 1}+\mathcal{C}_{I, 2}+\bar{\mathbf{c}}_1+\bar{\mathbf{c}}_2\}
\end{align*}
 Let  $\mathcal{W}=\mathcal{W}_1+\mathcal{W}_2$, and define $Z \triangleq X_1+X_2$. Recall, $\bar{\mathbf{c}}_i=\mathbf{v}_i \bar{\mathbf{G}}+\bar{\mathbf{b}}_i, i=1,2$. Using the union bound, we have
\begin{align}\nonumber
P(E_d| \mathbf{x}_i, \mathbf{v}_i, \mathbf{w}_i, i=1,2)&\leq   \sum_{\substack{ \tilde{\mathbf{w}} \in \mathcal{W}}}\sum_{\substack{ \tilde{\mathbf{z}}\in A_{\epsilon}^{(n)}(Z)\\ \mathbf{\tilde{z}}\neq \mathbf{x}_1+\mathbf{x}_2}} \mathbbm{1}\{\tilde{\mathbf{w}}\mathbf{G}+(\mathbf{v}_1+\mathbf{v}_2)\mathbf{\bar{G}}+2\mathbf{b}+\mathbf{\bar{b}}_1+\mathbf{\bar{b}}_2=\tilde{\mathbf{z}}\}\\\label{eq: bound on P(E|X_1X_2...)}
 &\leq   \sum_{\substack{ \tilde{\mathbf{w}} \in \mathcal{W}\\ \tilde{\mathbf{w}} \neq \mathbf{w}_{ 1}+\mathbf{w}_{ 2}}}\sum_{ \tilde{\mathbf{z}}\in A_{\epsilon}^{(n)}(Z)} \mathbbm{1}\{\tilde{\mathbf{w}}\mathbf{G}+(\mathbf{v}_1+\mathbf{v}_2)\mathbf{\bar{G}}+2\mathbf{b}+\mathbf{\bar{b}}_1+\mathbf{\bar{b}}_2=\tilde{\mathbf{z}}\}
\end{align}
 
The second inequality follows, because in general $\mathbf{\tilde{w}}\neq \mathbf{w}_1+\mathbf{w}_2$ does not  imply $\tilde{\mathbf{z}}\neq \mathbf{x}_1+\mathbf{x}_2$. This is due to the fact that $\mathbf{G}$ is not injective necessarily.   Since there is no encoding error, $\lambda_i(\mathbf{x}_i) \geq \frac{1}{2}E(\lambda_i(\mathbf{x}_i))$. As a result, 
\begin{align}\label{eq: bound on P(i, u|x_1)}
P(\mathbf{v}_i, \mathbf{w}_i|\mathbf{x_i})\leq \frac{2}{E(\lambda_i(\mathbf{x}_i))}\mathbbm{1}\{\mathbf{w}_{ i}\mathbf{G}+\mathbf{v}_i\mathbf{\bar{G}} + \mathbf{b}+\mathbf{\bar{b}}_i=\mathbf{x}_i	\}
\end{align}
%
 Using the bounds given in (\ref{eq: bound on P(E|X_1X_2...)}) and (\ref{eq: bound on P(i, u|x_1)}), we get

\begin{align*}
P(E_d \cap E_1^c \cap E_2^c|\mathbf{x}_1, \mathbf{x}_2)&\leq  \left[ \prod_{j=1}^{2}\sum_{ \substack{  \mathbf{v}_j \in \mathcal{V}_j\\ \mathbf{w}_{ j} \in \mathcal{W}_j}}  \frac{2}{E(\lambda_j(\mathbf{x}_j))}\mathbbm{1}\{\mathbf{w}_{ j}\mathbf{G}+\mathbf{v}_j\mathbf{\bar{G}} + \mathbf{b}+\mathbf{\bar{b}}_j=\mathbf{x}_j	\} \right] \\
&\sum_{\substack{ \tilde{\mathbf{w}} \in \mathcal{W}\\ \tilde{\mathbf{w}} \neq \mathbf{w}_{ 1}+\mathbf{w}_{ 2}}}\sum_{ \tilde{z}\in A_{\epsilon}^{(n)}(Z)} \mathbbm{1}\{\tilde{\mathbf{w}}\mathbf{G}+(\mathbf{v}_1+\mathbf{v}_2)\mathbf{\bar{G}}+2\mathbf{b}+\mathbf{\bar{b}}_1+\mathbf{\bar{b}}_2=\tilde{\mathbf{z}}\}
\end{align*}

Next, we average $P(E_d \cap E_1^c \cap E_2^c|\mathbf{x}_1, \mathbf{x}_2)$ over all possible choices of $\mathbf{G}, \mathbf{\bar{G}}, \mathbf{b}, \mathbf{\bar{b}}_1$, and $\mathbf{\bar{b}}_2$. We obtain
\begin{align*}
\EE\{P(E_d \cap E_1^c & \cap E_2^c|\mathbf{x}_1, \mathbf{x}_2)\}\leq \sum_{ \substack{  \mathbf{v}_1 \in \mathcal{V}_1\\ \mathbf{w}_{ 1} \in \mathcal{W}_1}}\frac{2}{E(\lambda_1(\mathbf{x}_1))}
\sum_{ \substack{  \mathbf{v}_2 \in \mathcal{V}_2\\ \mathbf{w}_{ 2} \in \mathcal{W}_2}}\frac{2}{E(\lambda_2(\mathbf{x}_2))} \sum_{\substack{ \tilde{\mathbf{w}} \in \mathcal{W}\\ \tilde{\mathbf{w}} \neq \mathbf{w}_1+\mathbf{w}_2}}\sum_{ \tilde{z}\in A_{\epsilon}^{(n)}(Z)}\\
& P\{\tilde{\mathbf{w}}\mathbf{G}+(\mathbf{v}_1+\mathbf{v}_2)\mathbf{\bar{G}}+2\mathbf{B}+\mathbf{\bar{B}}_1+\mathbf{\bar{B}}_2=\tilde{\mathbf{z}}, \mathbf{w}_{ i}\mathbf{G}+\mathbf{v}_i\mathbf{\bar{G}} + \mathbf{B}+\mathbf{\bar{B}}_i=\mathbf{x}_i, i=1,2\}
\end{align*}
Note $\mathbf{\bar{B}}_1$ and $\mathbf{\bar{B}}_2$ are independent random variables with uniformly distributed over $\ZZ_{p^r}^n$. Therefore, the innermost term in the above summations equals 
\begin{align}\label{eq dist: probability}
{p^{-2nr}}P\{(\tilde{\mathbf{w}}-\mathbf{w}_{ 1}-\mathbf{w}_{ 2})\mathbf{G}=\tilde{\mathbf{z}}-\mathbf{x}_1-\mathbf{x}_2\}.
\end{align}
We apply Lemma \ref{lem: P(phi)}, to calculate the above probability. If $\tilde{\mathbf{w}}-\mathbf{w}_{ 1}-\mathbf{w}_{ 2} \in H_s^k \backslash H_{s+1}^k$, then \eqref{eq dist: probability} equals to
\begin{align}\label{eq dist: probability 2}
p^{-2nr}p^{-n(r-s)} \mathbbm{1}\{ \tilde{\mathbf{z}}-\mathbf{x}_1-\mathbf{x}_2 \in H_s^k  \}.
\end{align}

As a result, we have
\begin{align*}
\EE\{P(E_d \cap E_1^c \cap E_2^c|\mathbf{x}_1, \mathbf{x}_2)\}&\leq  \sum_{ \substack{  \mathbf{v}_1 \in \mathcal{V}_1\\ \mathbf{w}_{ 1} \in \mathcal{W}_1}}\frac{2}{E(\lambda_1(\mathbf{x}_1))}
\sum_{ \substack{  \mathbf{v}_2 \in \mathcal{V}_2\\ \mathbf{w}_{ 2} \in \mathcal{W}_2}}\frac{2}{E(\lambda_2(\mathbf{x}_2))}\\
&\sum_{s =0}^{r-1} \sum_{\substack{ \tilde{\mathbf{w}} \in \mathcal{W}\\ \tilde{\mathbf{w}} - \mathbf{w}_{ 1}-\mathbf{w}_{ 2}\in H^{k}_{s} \backslash H_{s+1}^k }}\sum_{ \substack{ \tilde{\mathbf{z}}\in A_{\epsilon}^{(n)}(Z)\\ \tilde{\mathbf{z}}-\mathbf{x}_1-\mathbf{x}_2 \in H^n_{s}}} p^{-2nr}p^{-n(r-s)}
\end{align*}
Since the most inner terms in the above summations depend only on $s$, we can replace the  summations over $\mathbf{\tilde{w}}$ and $\mathbf{\tilde{z}}$ with the size of the associated sets. We apply Lemma \ref{lem: typical set intersection subgroup} to bound the size of these sets. Also, we can replace the summations over $\mathbf{v}_i$ and $\mathbf{w}_i, i=1,2$ with the size of the related sets. Define $W \triangleq W_1+W_2$, we get,
\begin{align*}
\EE\{P(E_d \cap E_1^c \cap E_2^c|\mathbf{x}_1, \mathbf{x}_2)\}&\leq  |\mathcal{W}_1||\mathcal{V}_1|\frac{2}{E({\lambda_1}(\mathbf{x}_1))}   |\mathcal{W}_2||\mathcal{V}_2|\frac{2}{E(\lambda_2(\mathbf{x}_2))}\\ &\sum_{s=0}^{r-1} 2^{n(H(Z|[Z]_{s})+o(\epsilon))}2^{k (H(W|Q [W]_s)+o(\epsilon))} p^{-2nr}p^{-n(r-s)}.
\end{align*}
Note that from \eqref{eq: expectation_delta2} in the proof of Lemma \ref{lem: covering}, $E({\lambda_i}(\mathbf{x}_i))=|\mathcal{W}_i||\mathcal{V}_i|p^{-nr}, i=1,2$. Therefore, we have 
\begin{align*}
\EE\{P(E_d \cap E_1^c \cap E_2^c|\mathbf{x}_1, \mathbf{x}_2)\}\leq 4 \sum_{s =0}^{r-1}  2^{n(H(Z|[Z]_{s})+o(\epsilon))} 2^{k (H(W|Q, [W]_s)+o(\epsilon))} p^{-n(r-s)}.
\end{align*}
Note that the above bound does not depend on $\epsilon$-typical sequences $\mathbf{x}_1$ and $\mathbf{x}_2$. Using standard arguments for $\epsilon$-typical sets, the probability that $(\mathbf{X}_1^n, \mathbf{X}_2^n)\notin A_\epsilon^{(n)}(X_1,X_2)$ is upper-bounded by $\frac{c}{n\epsilon^2}$, where $c=\frac{p^{6r}}{4}$.  Hence, we have
\begin{align*}
\EE\{P(E_d \cap E_1^c \cap E_2^c)\} \leq \frac{c}{n\epsilon^2}+4(1-\frac{c}{n\epsilon^2}) \sum_{s =0}^{r-1}  2^{n(H(Z|[Z]_{s})+o(\epsilon))} 2^{k (H(W|Q, [W]_s)+o(\epsilon))} p^{-n(r-s)}.
\end{align*}
Therefore, $\EE\{P(E_d \cap E_1^c \cap E_2^c)\}$ tends to zero as $n\rightarrow \infty$, if for any $s \in [0:r-1]$,
\begin{align}\label{eq: packing_dist}
\frac{k}{n} H(W| Q, [W]_{s}) < \log_2 p^{(r-s)}-H(Z|[Z]_{s})-o(\epsilon).
\end{align}
Next, we use \eqref{eq: packing_dist} to show that the bounds in (\ref{eq: covering_dist}) are redundant except the following:
\begin{align}\label{eq: covering_dist_simplified}
R_i+\frac{k}{n} H(W_i|Q) =  \log_2 p^r.
\end{align}

For that, we compare \eqref{eq: covering_dist_simplified} with the bounds in \eqref{eq: covering_dist} for different values of $s$. Noting that $H(W_i|Q)=H([W_i]_s|Q)+H(W_i|Q [W_i]_s)$, it is sufficient to show that  $\frac{k}{n} H(W_i|Q, [W_i]_s) \leq \log_2 p^{r-s}.$ To show this inequality,  we first prove that 
\begin{align}\label{eq dist: inequality of H(w1+w2)}
H(W_i|Q, [W_i]_s) \leq H(W_1+W_2|Q, [W_1+W_2]_s), ~ i=1,2, ~ 0 \leq s \leq r .
\end{align}
 Then, using \eqref{eq: packing_dist}, we get  $\frac{k}{n} H(W_i|Q, [W_i]_s) \leq \log_2 p^{r-s}$. In what follows, we prove \eqref{eq dist: inequality of H(w1+w2)}. We have
\begin{align*}
H(W_1+W_2|Q, & [W_1+W_2]_s)=H(W_1+W_2|Q, [[W_1]_s+[W_2]_s]_s)\\
&\geq H(W_1+W_2|Q, [W_1]_s,[W_2]_s) \\
&= H(W_1, W_2|Q, [W_1]_s,[W_2]_s)- H(W_1|Q, [W_1]_s,[W_2]_s, W_1+W_2)\\
&\stackrel{(a)}{=} H(W_2|Q, [W_2]_s)+H(W_1|Q, [W_1]_s)- H(W_1|Q, [W_1]_s,[W_2]_s, W_1+W_2)\\
&\stackrel{(b)}{=} H(W_2|Q, [W_2]_s)+I(W_1; W_1+W_2 |Q, [W_1]_s, [W_2])\\
&\geq H(W_2|Q, [W_2]_s),
\end{align*}
where $(a)$ and $(b)$ hold because of the Markov chain $W_1 \leftrightarrow Q\leftrightarrow W_2$. Similarly, we can show that $H(W_1+W_2|Q, [W_1+W_2]_s) \geq H(W_1|Q, [W_1]_s).$

Finally, using (\ref{eq: covering_dist_simplified}) and (\ref{eq: packing_dist}) the following holds
\begin{align}\label{eq: dist_achievable}
R_i\geq  \log_2 p^r-  \min_{0 \leq s \leq r-1}  \frac{H(W_i|Q)}{H(W_1+W_2| Q, [W_1+W_2]_{s})} (\log_2 p^{(r-s)}-H(Z|[Z]_{s})),
\end{align}
where  we minimize the above bound over all PMFs of the form $P_{QW_1V_1W_2V_2}= P_Q \prod_i \left(P_{V_i|Q}P_{W_i|Q}\right)$, such that $p(q)$ is a rational number for all $q\in \mathcal{Q}$. Since rational numbers are dense in $\RR$, one can consider arbitrary PMF $p(q)$.  Lastly, in the next lemma, we show that the cardinality bound  $|\mathcal{Q}|\leq r$ is sufficient to optimize \eqref{eq: dist_achievable}. 

\begin{lem}\label{lem: cardinality of Q}
The cardinality of $\mathcal{Q}$ is bounded by $|\mathcal{Q}|\leq r$. 
\end{lem}

\begin{proof}
Note that (\ref{eq: packing_dist}) and  (\ref{eq: covering_dist_simplified}) give an alternative characterization of the achievable region. Using these equations, observe that this region is convex in $\RR^2$. As a result, we can characterize the achievable region by its supporting hyperplanes. Let $\bar{R}_i:= \log_2 p^r -R_i, i=1,2$. Using (\ref{eq: dist_achievable}) for any $0\leq \alpha \leq 1 $ the corresponding supporting hyperplane is characterized by 
\begin{align}\nonumber
\big(\alpha \bar{R}_1 +(1-\alpha)\bar{R}_2 \big)& H(W| Q, [W]_{s}) \\ \label{eq: support hyperplan dist}
&- \Big(\alpha H(W_1|Q) + (1-\alpha)H(W_2|Q)\Big) \Big(\log_2 p^{(r-s)}-H(Z|[Z]_{s})\Big)\leq 0, 
\end{align}
where $s\in [0,r-1]$. We use the support lemma for the above inequalities to bound $|\mathcal{Q}|$. To this end, we first show that the left-hand side of these inequalities are continuous functions of conditional  PMF's of $W_1$ and $W_2$ given $Q$. Let $\mathscr{P}_r$ denote the set of all product PMF's on $\ZZ_{p^r}\times \ZZ_{p^r}$. Note $\mathscr{P}_r$  is a compact set. Fix $q\in \mathcal{Q}$. Denote $ f(p(w_1|q)p(w_2|q))=\alpha H(W_1|Q=q) + (1-\alpha)H(W_2|Q=q)$ and $g_s(p(w_1|q)p(w_2|q))=  H(W_1+W_2| Q=q, [W_1+W_2]_{s})$, where $s\in [0:r-1]$. We show that  $f(\cdot), g_s(\cdot)$ are real valued continuous functions of $ \mathscr{P}_r$. Since the entropy function is continuous then so is $f$. We can write  $g_s(p(w_1|q)p(w_2|q))= H(W_1+W_2|Q=q)-H([W_1+W_2]_s|Q=q)$. Note that $[\cdot]_s$ is a continuous function from $ \mathscr{P}_r$ to $\mathscr{P}_r$. This implies that  $H([\cdot]_s)$ is also continuous. So $g_s$ is continuous. As a result, the left-hand side of the bounds in (\ref{eq: support hyperplan dist}) are real valued continuous functions of $\mathscr{P}_r$. Therefore, we can apply the support lemma \cite{ElGamal-book}. Since there are $r$ bounds for different values of $s$, then $|\mathcal{Q}|\leq r$. 
\end{proof}



\section{Proof of Theorem \ref{thm: comp MACnon uniform}}\label{sec: proof of comp mac nonuniform}
We need to find conditions for which the probability of the error events $E_1, E_2$ and $E_d$ approach zero. Suppose $\mathbf{G}$ is the generator matrix, and $\mathbf{b}$ is the translation of $\mathcal{C}_{I,1}$ and $\mathcal{C}_{I,2}$.  In addition, suppose $\mathbf{\bar{G}}$ is the generator matrix and $\mathbf{\bar{b}}_i$ is the translation defined for $\bar{\mathcal{C}}_i, i=1,2$. 
{  For any $\mathbf{a}\in \ZZ_{p^r}^k$ and $\mathbf{\bar{a}}\in \ZZ_{p^r}^l$ define the map $\phi(\mathbf{a}, \mathbf{\bar{a}})=\mathbf{a} \mathbf{G}+\bar{\mathbf{a}}\mathbf{\bar{G}}$. By $\Phi(\cdot, \cdot)$ denote the map $\phi$ whose matrices are selected randomly and uniformly. }

\subsection{Analysis of $E_1, E_2$}
 For any sequence $\mathbf{v}_i \in \mathcal{V}_i$ define $$\lambda_i(\mathbf{v}_i)=\sum_{\mathbf{w}_i \in \mathcal{W}_i } \sum_{\mathbf{x}_i \in A_\epsilon^{(n)}(X_i)}  \mathbbm{1}\{\mathbf{x}_i=\phi(\mathbf{w}_i, \mathbf{v}_i)+\mathbf{b}+\mathbf{\bar{b}}_i\},$$ where $i=1,2$. Therefore, $E_i$ occurs if $\lambda_i(\mathbf{v}_i)=0$. For more convenience, we weaken the definition of event $E_i$. We say $E_i$ occurs, if $\lambda_i(\mathbf{v}_i)< \frac{1}{2}E(\lambda_i(v_i))$. Using Lemma \ref{lem: covering} we can show that $P(E_i)\rightarrow 0$ as $n\rightarrow \infty$, if 
 \begin{align}\label{eq comp: covering}
 \frac{k}{n}H([W_i]_t|Q) \geq \log_2 p^s -H([X_i]_t) + \gamma(\epsilon), ~ i=1,2, ~ 1\leq t \leq r,
 \end{align}
where $\lim_{\epsilon \rightarrow 0} \gamma(\epsilon)=0$.
{
\subsection{Analysis of $E_c \cap E^c_1 \cap E^c_2 $}
 Define the set 
\begin{align*}
\mathcal{E} \triangleq \{   (\mathbf{x}_1, \mathbf{x}_2) \in A_\epsilon^{(n)}(X_1) \times A_\epsilon^{(n)}(X_2) : (\mathbf{x}_1, \mathbf{x}_2) \in A_\epsilon^{(n)}(X_1, X_2)     \}.
\end{align*}
Therefore, probability of $E_c$ can be written as  
\begin{align*}
P(E_c \cap E^c_1 \cap E^c_2 ) =  \sum_{(\mathbf{x}_1, \mathbf{x}_2)\in \mathcal{E}} P( e_1(\Theta_1)=\mathbf{x}_1,e_2(\Theta_2)=\mathbf{x}_2 ),
\end{align*}
where $e_i$ is the output of the $i$th encoder, and $\Theta_i$ is the random message to be transmitted by encoder $i$, where $i=1,2$. By the definition of $\phi_1(\cdot)$ and $\phi_2(\cdot)$, we have
\begin{align*}
P(E_c \cap E^c_1 \cap E^c_2)&=\sum_{(\mathbf{x}_1, \mathbf{x}_2)\in \mathcal{E}} \prod_{i=1}^2 \left[  \sum_{ \mathbf{v}_i \in \mathcal{V}_i}   \sum_{ \mathbf{w}_i \in \mathcal{W}_i } \frac{1}{|\mathcal{V}_i|}  \mathbbm{1}\Big\{\lambda_i(\mathbf{v_i})\geq 1/2 ~ E(\lambda_i(\mathbf{v}_i)) \Big\} \mathbbm{1}\Big\{\phi_i(\mathbf{w}_i, \mathbf{v}_i)+\mathbf{b}+\bar{\mathbf{b}}_i \Big\} \right]
\end{align*}
We remove the indicator function on $\{\lambda_i(\mathbf{v_i})\geq 1/2 ~ E(\lambda_i(\mathbf{v}_i))\}$. This gives an upper-bound the above expression. Next, we taking expectation over all $\phi_1$ and $\phi_2$. We have
\begin{align*}
\EE\{P(E_c \cap E^c_1 \cap E^c_2)\}& \leq \sum_{(\mathbf{x}_1, \mathbf{x}_2)\in \mathcal{E}} \sum_{ \mathbf{v}_i \in \mathcal{V}_i, i=1,2 }   \sum_{ \mathbf{w}_i \in \mathcal{W}_i, i=1,2 } \frac{1}{|\mathcal{V}_1| |\mathcal{V}_2|} P\{\Phi_i(\mathbf{w}_i, \mathbf{v}_i)+\mathbf{B}+\bar{\mathbf{B}}_i, i=1,2\}\\
& \stackrel{(a)}{=} \sum_{(\mathbf{x}_1, \mathbf{x}_2)\in \mathcal{E}} \sum_{ \mathbf{v}_i \in \mathcal{V}_i, i=1,2 }   \sum_{ \mathbf{w}_i \in \mathcal{W}_i, i=1,2 } \frac{1}{|\mathcal{V}_1| |\mathcal{V}_2|} p^{-2nr}\\
&=  \sum_{(\mathbf{x}_1, \mathbf{x}_2)\in \mathcal{E}}  |\mathcal{W}_1||\mathcal{W}_2| p^{-2nr}.
\end{align*}
Note that  $(a)$ is because $\mathbf{B}_1$ and $\mathbf{B}_2$ are independent random vectors with uniform distribution over $\ZZ_{p^r}^n$.  Using the proof of Lemma \ref{lem: covering}, we provide a tighter than the one in \eqref{eq comp: covering}. We have
\begin{align*}
|\mathcal{W}_i|^{-1}|A_\epsilon^{(n)}(X_i)|^{-1}p^{nr} \leq 2^{-n\gamma(\epsilon)}, ~ i=1,2, 
\end{align*}
where $\gamma$ is any function of $\epsilon$, such that $\lim_{\epsilon\rightarrow 0} \gamma(\epsilon)=0$. This function will be determined. Therefore, we have 

\begin{align*}
\EE\{P(E_c \cap E^c_1 \cap E^c_2)\} \leq   \sum_{(\mathbf{x}_1, \mathbf{x}_2)\in \mathcal{E}} |A_\epsilon^{(n)}(X_1)|^{-1}|A_\epsilon^{(n)}(X_2)|^{-1} 2^{n2\gamma(\epsilon)}
\end{align*}
For any $\mathbf{x}_i \in A_\epsilon^{(n)}(X_i)$, we have $P^n_{X_1}(\mathbf{x}_1) \geq |A_\epsilon^{(n)}(X_i)|^{-1}$. Thus,
\begin{align*}
\EE\{P(E_c \cap E^c_1 \cap E^c_2)\} \leq   \sum_{(\mathbf{x}_1, \mathbf{x}_2)\in \mathcal{E}} 	 P^n_{X_1}(\mathbf{x}_1)P^n_{X_2}(\mathbf{x}_2) 2^{n2\gamma(\epsilon)} \leq 2^{n2\gamma(\epsilon)} P^n_{X_1X_2}( \mathcal{E})\leq 2^{-n( \delta(\epsilon)-2\gamma(\epsilon))}.
\end{align*}
Thus, if $\gamma < \frac{1}{2} \delta(\epsilon)$, then $\EE\{P(E_c \cap E^c_1 \cap E^c_2)\} \rightarrow 0$ as $n \rightarrow \infty$.
}

\subsection{Analysis of $E_d \cap ( E_1^c \cup E_2^c \cup E_c)^c$}
In what follows, we redefine the decoding operation. Suppose $\mathbf{x}_i=\phi(\mathbf{w}_i, \mathbf{v}_i)+\mathbf{b}+\bar{\mathbf{b}}_i$, is the codeword transmitted by encoder $i, i=1,2$. We require the decoder to decode $\mathbf{w}_1+\mathbf{w}_2$ and $\mathbf{v}_1+\mathbf{v}_2$. Upon receiving $\mathbf{y}$, the decoder finds {  $\tilde{\mathbf{w}} \in A_\epsilon^{(n)}(W_1 + W_2) $} and { $\tilde{\mathbf{v}} \in A_\epsilon^{(n)}(V_1 + V_2)$} such that $\phi(\tilde{\mathbf{w}}, \tilde{\mathbf{v}})+2\mathbf{b}+\bar{\mathbf{b}}_1+\bar{\mathbf{b}}_2$ is jointly typical with $\mathbf{y}$ with respect to $P_{X_1+X_2,Y}$.   Therefore, the new $E_d$ occurs,  if  $\tilde{\mathbf{w}}$ or $\tilde{\mathbf{v}}$ is not unique. This is a stronger condition, but it is more convenient for error analysis.  Fix $\phi, \mathbf{b}$ and $\mathbf{\bar{b}}_i, i=1,2$.
By $P(\mathbf{v}_i, \mathbf{w}_i, \mathbf{x}_i)$ denote the probability that $(\mathbf{v}_i, \mathbf{w}_i, \mathbf{x}_i)$ is selected at the $i$th encoder.  Then, $P(\mathbf{v}_i, \mathbf{w}_i, \mathbf{x}_i)= \frac{1}{|\mathcal{V}_i|}    \frac{1}{\lambda_i(\mathbf{v}_i)} \mathbbm{1}\{\phi(\mathbf{w}_{ i}, \mathbf{v}_i)+\mathbf{b}+ \mathbf{\bar{b}}_i=\mathbf{x}_i	\}.$

Then the probability of $E_d \cap ( E_1^c \cup E_2^c \cup E_c)^c$ equals
\begin{align*}
P(E_d \cap ( E_1^c \cup E_2^c \cup E_c)^c )&= \left[ \prod_{j=1}^{2}   \sum_{\mathbf{v}_j \in \mathcal{V}_j} \sum_{\mathbf{w}_{ j} \in \mathcal{	W}_j} \mathbbm{1}\Big\{\lambda_i(\mathbf{v_i})\geq 1/2 ~ E(\lambda_i(\mathbf{v}_i)), i=1,2  \Big\} \right]\\
& \sum_{(\mathbf{x}_1, \mathbf{x}_2) \in A_\epsilon^{(n)}(X_1,X_2)}  \sum_{\mathbf{y}\in \mathcal{Y}^n} P(\mathbf{v}_i, \mathbf{w}_{ i},\mathbf{x}_i, i=1,2)\\
&   P^n_{Y|X_1X_2}(\mathbf{y}|\mathbf{x_1, x_2})  P(E_d ~ | ~ ( E_1^c \cup E_2^c \cup E_c)^c, \mathbf{y}, \mathbf{x}_i, \mathbf{v}_i, \mathbf{w}_i, i=1,2)
\end{align*}
Next, we bound $ P(E_d ~ | ~ ( E_1^c \cup E_2^c \cup E_c)^c, \mathbf{y}, \mathbf{x}_i, \mathbf{v}_i, \mathbf{w}_i, i=1,2)$, and  $ P(\mathbf{v}_i \mathbf{w}_{ i},\mathbf{x}_i, i=1,2)$.
\begin{align*}
P(E_d ~ & | ~ ( E_1^c \cup E_2^c \cup E_c)^c, \mathbf{y}, \mathbf{x}_i, \mathbf{v}_i, \mathbf{w}_i, i=1,2)=\\\nonumber
&\mathbbm{1}\{ \exists ~ (\mathbf{\tilde{\mathbf{w}}},\tilde{\mathbf{v}})  \in  \mathcal{W} \times \mathcal{V} :  (\mathbf{\tilde{\mathbf{w}}},\tilde{\mathbf{v}})  \neq (\mathbf{w}_1+\mathbf{w}_2, \mathbf{v}_1+\mathbf{v}_2),   \phi(\mathbf{\tilde{\mathbf{w}}},\tilde{\mathbf{v}}) +2\mathbf{b}+\bar{\mathbf{b}}_1+\bar{\mathbf{b}}_2 \in A_{\epsilon'}^n(Z|\mathbf{y}) \},
\end{align*}
where  $\mathcal{W}\triangleq A_\epsilon^{(n)}(W_1 + W_2), \mathcal{V}\triangleq A_\epsilon^{(n)}(V_1+ V_2)$, and $Z \triangleq X_1+X_2$. Using the union bound, we have
\begin{align}\label{eq comp: bound on E_d given y, x, v, w}
P(E_d ~  | ~& ( E_1^c \cup E_2^c \cup E_c)^c, \mathbf{y}, \mathbf{x}_i, \mathbf{v}_i, \mathbf{w}_i, i=1,2)  \leq \\\nonumber
&  \sum_{\substack{ \tilde{\mathbf{w}} \in \mathcal{W}\\  \tilde{\mathbf{w}} \neq \mathbf{w}_1+\mathbf{w}_2}} \sum_{\substack{ \tilde{\mathbf{v}} \in \mathcal{V}\\  \tilde{\mathbf{v}} \neq \mathbf{v}_1+\mathbf{v}_2}} \sum_{\substack{ \tilde{\mathbf{z}}\in A_{\epsilon'}^{(n)}(Z|\mathbf{y})}} \mathbbm{1}\{\phi(\tilde{\mathbf{w}}, \tilde{\mathbf{v}})+2\mathbf{b}+\bar{\mathbf{b}}_1+\bar{\mathbf{b}}_2=\tilde{\mathbf{z}}\}
\end{align}

Note that $P(\mathbf{v}_i, \mathbf{w}_{ i},\mathbf{x}_i, i=1,2)=\prod_{i=1,2}P(\mathbf{v}_i, \mathbf{w}_{ i},\mathbf{x}_i)$. Since there is no encoding error, $\lambda_i(\mathbf{v}_i) \geq \frac{1}{2}E(\lambda_i(\mathbf{v}_i))$. As a result, 
\begin{align}\label{eq comp: bound on P(w_i, v_i, x_i)}
P(\mathbf{v}_i, \mathbf{w}_{ i},\mathbf{x}_i) \leq \frac{1}{|\mathcal{V}_i|}  \frac{2}{E(\lambda_i(\mathbf{v}_i))}\mathbbm{1}\{\phi(\mathbf{w}_{ i}, \mathbf{v}_i)+\mathbf{b}+ \mathbf{\bar{b}}_i=\mathbf{x}_i	\}
\end{align}

Therefore, using  \eqref{eq comp: bound on P(w_i, v_i, x_i)}, we have
\begin{align}\nonumber
P(E_d \cap & ( E_1^c \cup E_2^c \cup E_c)^c )\leq  \sum_{(\mathbf{x}_1, \mathbf{x}_2) \in A_\epsilon^{(n)}(X_1,X_2)}  \Big[ \prod_{j=1}^{2}   \sum_{\mathbf{v}_j \in \mathcal{V}_j} \sum_{\mathbf{w}_{ j} \in \mathcal{	W}_j} \mathbbm{1}\Big\{\lambda_j(\mathbf{v_j})\geq 1/2 ~ E(\lambda_j(\mathbf{v}_j))\Big\} \\\nonumber
& \frac{1}{|\mathcal{V}_j|}  \frac{2}{E(\lambda_i(\mathbf{v}_j))}\mathbbm{1}\{\phi(\mathbf{w}_{ j}, \mathbf{v}_j)+\mathbf{b}+ \mathbf{\bar{b}}_j=\mathbf{x}_j	\} \Big] \\\nonumber
&   \sum_{\mathbf{y}\in \mathcal{Y}^n}  P^n_{Y|X_1X_2}(\mathbf{y}|\mathbf{x_1, x_2})  P(E_d ~  | ~ ( E_1^c \cup E_2^c \cup E_c)^c, \mathbf{y}, \mathbf{x}_i, \mathbf{v}_i, \mathbf{w}_i, i=1,2)\\\nonumber
& \leq  \sum_{(\mathbf{x}_1, \mathbf{x}_2) \in A_\epsilon^{(n)}(X_1,X_2)}  \left[ \prod_{j=1}^{2}   \sum_{\mathbf{v}_j \in \mathcal{V}_j} \sum_{\mathbf{w}_{ j} \in \mathcal{	W}_j}  \frac{1}{|\mathcal{V}_j|}  \frac{2}{E(\lambda_i(\mathbf{v}_j))}\mathbbm{1}\{\phi(\mathbf{w}_{ j}, \mathbf{v}_j)+\mathbf{b}+ \mathbf{\bar{b}}_j=\mathbf{x}_j	\} \right] \\ \label{eq: bound 43}
&   \sum_{\mathbf{y}\in \mathcal{Y}^n}  P^n_{Y|X_1X_2}(\mathbf{y}|\mathbf{x_1, x_2})  P(E_d ~  | ~ ( E_1^c \cup E_2^c \cup E_c)^c, \mathbf{y}, \mathbf{x}_i, \mathbf{v}_i, \mathbf{w}_i, i=1,2)
\end{align}

The last inequality follows by eliminating the indicator function on $\{\lambda_i(\mathbf{v_i})\geq 1/2 ~ E(\lambda_i(\mathbf{v}_i)), i=1,2  \Big\}$. Note that for jointly $\epsilon$-typical sequences $\mathbf{x}_1, \mathbf{x}_2$ and large enough $n$, we have  $P(\mathbf{Y}^n \notin A_{\tilde{\epsilon}}^{(n)}(Y|\mathbf{x}_1, \mathbf{x}_2)) \leq \frac{c}{n\tilde{\epsilon}^2}$, where $c$ is a constant. This follows from the standard arguments on typical sets.  Thus, using  \eqref{eq: bound 43} and  \eqref{eq comp: bound on E_d given y, x, v, w}  we get
 
 \begin{align*}
P(E_d &	 \cap  ( E_1^c \cup E_2^c \cup E_c)^c )\leq \frac{c}{n\tilde{\epsilon}^2}+\\
& \sum_{(\mathbf{x}_1, \mathbf{x}_2) \in A_\epsilon^{(n)}(X_1,X_2)}  \left[ \prod_{j=1}^{2}   \sum_{\mathbf{v}_j \in \mathcal{V}_j} \sum_{\mathbf{w}_{ j} \in \mathcal{	W}_j}  \frac{1}{|\mathcal{V}_j|}  \frac{2}{E(\lambda_i(\mathbf{v}_j))}\mathbbm{1}\{\phi(\mathbf{w}_{ j}, \mathbf{v}_j)+\mathbf{b}+ \mathbf{\bar{b}}_j=\mathbf{x}_j	\} \right] \\\nonumber 
&\sum_{\mathbf{y}\in A_{\tilde{\epsilon}}^n(Y|\mathbf{x}_1,\mathbf{x}_2)}  P^n_{Y|X_1X_2}(\mathbf{y}|\mathbf{x_1, x_2})  \sum_{\substack{ \tilde{\mathbf{w}} \in \mathcal{W}\\  \tilde{\mathbf{w}} \neq \mathbf{w}_1+\mathbf{w}_2}} \sum_{\substack{ \tilde{\mathbf{v}} \in \mathcal{V}\\  \tilde{\mathbf{v}} \neq \mathbf{v}_1+\mathbf{v}_2}} \sum_{\substack{ \tilde{\mathbf{z}}\in A_{\epsilon'}^{(n)}(Z|\mathbf{y})}} \mathbbm{1}\{\phi(\tilde{\mathbf{w}}, \tilde{\mathbf{v}})+2\mathbf{b}+\bar{\mathbf{b}}_1+\bar{\mathbf{b}}_2=\tilde{\mathbf{z}}\}
\end{align*} 
Next, we take the average of the above expression over all maps $\phi$, and  all vectors $\mathbf{b}, \bar{\mathbf{b}}_i, i=1,2$. 
\begin{align*}
 \EE\{P(E_d \cap & ( E_1^c \cup E_2^c \cup E_c)^c )\} \leq \frac{c}{n\tilde{\epsilon}^2}+   \left[ \prod_{j=1}^{2}   \sum_{\mathbf{v}_j \in \mathcal{V}_j} \sum_{\mathbf{w}_{ j} \in \mathcal{	W}_j}   \frac{1}{|\mathcal{V}_j|}  \frac{2}{E(\lambda_j(\mathbf{v}_j))} \right] \\
& \sum_{(\mathbf{x}_1, \mathbf{x}_2, \mathbf{y}) \in A_{\bar{\epsilon}}^{(n)}(X_1,X_2,Y)}   P^n_{Y|X_1X_2}(\mathbf{y}|\mathbf{x_1, x_2})
  \sum_{\substack{ \tilde{\mathbf{w}} \in \mathcal{W}\\  \tilde{\mathbf{w}} \neq \mathbf{w}_1+\mathbf{w}_2}} \sum_{\substack{ \tilde{\mathbf{v}} \in \mathcal{V}\\  \tilde{\mathbf{v}} \neq \mathbf{v}_1+\mathbf{v}_2}} \sum_{\substack{ \tilde{\mathbf{z}}\in A_{\epsilon'}^{(n)}(Z|\mathbf{y})}} 
\\ &P\{\tilde{z}=\Phi(\mathbf{\tilde{w}}, \mathbf{\tilde{v}})+2\mathbf{B}+\bar{\mathbf{B}}_1+\bar{\mathbf{B}}_1, x_1 = \Phi(\mathbf{w}_1, \mathbf{v}_1)+\mathbf{B}+\bar{\mathbf{B}}_1, x_2= \Phi(\mathbf{w}_2, \mathbf{v}_2)+\mathbf{B}+\bar{\mathbf{B}}_1\}
\end{align*}

Notice that $ \mathbf{B},  \bar{\mathbf{B}}_1$, and are $\bar{\mathbf{B}}_1$ are uniform over $\ZZ_{p^r}^n$ and independent of other random variables. Hence, the innermost term in the above summations is simplified to
\begin{align}\label{eq: pf comp over MAC - probability}
p^{-2nr} P\{\mathbf{\tilde{z}-x_1-x_2}= \Phi(\mathbf{\tilde{w}}-(\mathbf{w_1+w_2}), \mathbf{\tilde{v}}-(\mathbf{v_1+v_2}))\}
\end{align}
Using  Lemma \ref{lem: P(phi)}, if $\mathbf{\tilde{w}}-(\mathbf{w_1+w_2}) , \mathbf{\tilde{v}}-(\mathbf{v_1+v_2}) \in H^k_s \backslash H_{s+1}^k$ the expression in \eqref{eq: pf comp over MAC - probability}  equals $$p^{-2nr} p^{-n(r-s)}\mathbbm{1}\{\tilde{z}-\mathbf{x_1-x_2} \in H_s^n\},$$ where $ 0\leq s \leq r-1$. Therefore, $\EE\{P(E_d \cap  ( E_1^c \cup E_2^c \cup E_c)^c )\}$ is upper-bounded as 
\begin{align} \nonumber
\EE\{P(E_d \cap & ( E_1^c \cup E_2^c \cup E_c)^c )\}  \leq \frac{c}{n\tilde{\epsilon}^2}+  \\ \nonumber
& \left[ \prod_{j=1}^{2}   \sum_{\mathbf{v}_j \in \mathcal{V}_j} \sum_{\mathbf{w}_{ j} \in \mathcal{	W}_j}   \frac{1}{|\mathcal{V}_j|}  \frac{2}{E(\lambda_j(\mathbf{v}_j))} \right] \sum_{(\mathbf{x}_1, \mathbf{x}_2, \mathbf{y}) \in A_{\bar{\epsilon}}^{(n)}(X_1,X_2,Y)}P^n_{Y|X_1X_2}(\mathbf{y}|\mathbf{x_1, x_2})\\\label{eq: pe_1}
&   \sum_{s=0}^{r-1}
  \sum_{\substack{ \tilde{\mathbf{w}} \in \mathcal{W}\\  \mathbf{\tilde{w}}-(\mathbf{w_1+w_2}) \in H^{k}_{s}}} \sum_{\substack{ \tilde{\mathbf{v}} \in \mathcal{V}\\  \mathbf{\tilde{v}}-(\mathbf{v_1+v_2}) \in H^{k}_{s}}}\sum_{\substack{\tilde{z}\in A_\epsilon^n(Z|y)\\ \mathbf{\tilde{z}-x_1-x_2} \in H^n_{s}}} p^{-2nr} p^{-n(r-s)}
\end{align}

Note the most inner term in the above summations does not depend on the value of $\mathbf{\tilde{z}, \tilde{v}}$ and $\mathbf{\tilde{w}}$. Hence, we replace those summations by the size of the corresponding subsets. Using Lemma \ref{lem: typical set intersection subgroup} we can bound the size of these subsets and get the following bound on the probability of error
\begin{align*}
\EE\{P(E_d \cap&  ( E_1^c \cup E_2^c \cup E_c)^c )\}  \leq \frac{c}{n\tilde{\epsilon}^2}+\\
&  \left[ \prod_{j=1}^{2}   \sum_{\mathbf{v}_j \in \mathcal{V}_j} \sum_{\mathbf{w}_{ j} \in \mathcal{	W}_j}   \frac{1}{|\mathcal{V}_j|}  \frac{2}{E(\lambda_j(\mathbf{v}_j))} \right] \sum_{(\mathbf{x}_1, \mathbf{x}_2, \mathbf{y}) \in A_{\bar{\epsilon}}^{(n)}(X_1,X_2,Y)}P^n_{Y|X_1X_2}(\mathbf{y}|\mathbf{x_1, x_2})\\
&\sum_{s=0}^{r-1} 2^{k(H(W|Q, [W]_s)+\eta_1(\epsilon))}2^{l(H(V|Q, [V]_s)+\eta_2(\epsilon))}~ 2^{n(H(Z|Y[Z]_{s})+\eta_3(\epsilon))}p^{-2nr} p^{-n(r-s)},
\end{align*}
where $W=W_1+W_2, V=V_1+V_2$, and $\lim_{\epsilon \rightarrow 0} \eta_i(\epsilon)=0, i=1,2,3$. Note that $E({\lambda_i}(\mathbf{v}_i))=|\mathcal{W}_i||A_\epsilon^{(n)}(X_i)|p^{-nr}, i=1,2$.
As the terms in the above expression do not depend on the values of $\mathbf{w}_i, \mathbf{v}_i, \mathbf{x}_i,i=1,2$ and  $\mathbf{y}$, we can replace the summations over them with the corresponding sets.  As a result, we have 
\begin{align*}
\EE\{P(E_d \cap  ( E_1^c \cup E_2^c \cup E_c)^c )\} & \leq \frac{c}{n\epsilon^2} + 4  \sum_{s=0}^{r-1}  p^{-n(r-s)}  2^{kH(W|Q, [W]_s)}2^{lH(V|Q, [V]_s)}~ 2^{n(H(Z|Y[Z]_{s})+\delta'(\epsilon))},
\end{align*}
where $\lim_{\epsilon \rightarrow 0}  \delta'(\epsilon)=0$.  Therefore, the right-hand side of the above inequality approaches zero as $n\rightarrow \infty$, if the following bounds hold:
\begin{align}\label{equ: bound simple form}
\frac{k}{n} H(W|Q, [W]_s) +\frac{l}{n} H(V|Q, [V]_s)\leq   \log_2 p^{r-s}- H(Z|Y[Z]_{s})-\delta(\epsilon), \quad \mbox{for} ~ 0\leq s\leq r-1.
\end{align}
Next, we apply the Fourier-Motzkin technique \cite{ElGamal-book} to eliminate $\frac{k}{n}$ from \eqref{eq comp: covering} and \eqref{equ: bound simple form}. We get
\begin{align*}
\frac{l}{n} H(V|Q, [V]_s) \leq \log_2 p^{r-s}- H(Z|Y[Z]_{s}) -\frac{H(W|Q, [W]_s)}{H([W_i]_t|Q)} (\log_2 p^t - H([X_i]_t)) -o(\epsilon),
\end{align*}
 where $i=1,2, ~~0 \leq s \leq r-1$, and $1\leq t \leq r$.  Note by definition $$R_i=\frac{1}{n}\log_2 |\bar{\mathcal{C}}_i|\leq \frac{1}{n}\log_2 |\mathcal{V}_i| \leq \frac{l}{n} H(V_i|Q)  .$$ Therefore, we obtain the bounds in the theorem. Using the same argument as in Lemma \ref{lem: cardinality of Q}, we can bound the cardinality of $Q$ by  $|\mathcal{Q}| \leq r^2$. This completes the proof.

\section{Proof of Lemma \ref{lem: suboptimality of Gelfand-Pinsker}}\label{sec: proof of Lemma MAC with states}
\begin{proof}
Consider the bound on the sum-rate given in (\ref{eq: sum-rate}). The set of all $(R_1, R_2)$ satisfying only this bound is an outer-bound for $\mathscr{R}_{GP}$. The time-sharing random variable $Q$ is trivial for this outer-bound, because there is only one inequality on the rates, and because of the cost constraints $\EE\{c_i(X_i)\}=0, i=1,2$.  For any distribution $P\in \mathscr{P}_{GP}$, we obtain
\begin{align}\nonumber
R_1+R_2&\leq  I(U_1 U_2; Y)-I(U_1;S_1)-I(U_2;S_2)\\\nonumber
&=H(Y)-H(Y|U_1U_2)-H(S_1)+H(S_1|U_1)-H(S_2)+H(S_2|U_2)\\\nonumber
&\leq H(S_1|U_1)+H(S_2|U_2)-H(Y|U_1U_2)-2\\\label{eq: sum-rate last bound}
&=\max_{P \in \mathscr{P}_{GP}} \sum_{u_1\in \mathcal{U}_1}\sum_{u_2\in \mathcal{U}_2} p(u_1,u_2) \Big( H(S_1|u_1)+H(S_2|u_2)-H(Y|u_1u_2)-2\Big)
\end{align}
where the second inequality holds, as $H(Y)\leq 2$, and $H(S_i)=2$ for $i=1,2$.  In the next step, we relax the conditions in $\mathscr{P}_{GP}$, and provide an upper-bound on (\ref{eq: sum-rate last bound}). For $i=1,2$, and any $u_i\in \mathcal{U}_i$, define   $\mathscr{P}_{u_i}$ as the collection of all conditional PMFs $ p(s_i,x_i|u_i)$ on $\ZZ^2_4$ such that 
\begin{enumerate}
\item $X_i=f_i(S_i,u_i)$ for some function $f_i$,
\item  $E(c_i(X_i)|u_i)=0.$
\end{enumerate}
In the first condition, given $u_i$,  $f_i(s_i, u_i)$ can be thought as a function $g_{u_i}$ of $s_i$. For different $u_i$'s we have different functions $g_{u_i}(s_i)$.  The second condition is implied from the cost constraint $E(c_i(X_i))=0$, because without loss of generality we assume $p(u_i)>0$ for all $u_i\in \mathcal{U}_i$. Also, note that we removed the condition that $S_i$ is uniform over $\ZZ_4$. 
Hence, $\mathscr{P}_{GP}$ is a subset of the set of all PMFs of the form $P=\prod_{i=1}^2 p(u_i)p(s_i,x_i|u_i)$, where $p(s_i,x_i|u_i)\in \mathscr{P}_{u_i}, i=1,2$. 
As a result, (\ref{eq: sum-rate last bound}) is upper-bounded by 
\begin{align}
&R_1+R_2\\
&\leq  \max_{p(u_1), p(u_2)}\max_{\substack{ p(s_i,x_i|u_i)\in \mathscr{P}_{u_i}\\ i=1,2 }} \sum_{u_1\in \mathcal{U}_1}\sum_{u_2\in \mathcal{U}_2} p(u_1,u_2) \Big( H(S_1|u_1)+H(S_2|u_2)-H(Y|u_1u_2)-2\Big)\\
&\leq\max_{u_1 \in \mathcal{U}_1, u_2\in \mathcal{U}_2} \max_{\substack{ p(s_i,x_i|u_i)\in \mathscr{P}_{u_i}\\ i=1,2 }}\Big( H(S_1|u_1)+H(S_2|u_2)-H(Y|u_1u_2)-2\Big)
\end{align}
Fix $u_2\in \mathcal{U}_2$ and $p(s_2,x_2|u_2)\in \mathscr{P}_{u_2}$. We maximize over all $u_1\in \mathcal{U}_1$ and $p(s_1,x_1|u_1)\in \mathscr{P}_{u_1}$. Let $N=X_2+ S_2$, where $X_2$ and $S_2$ are distributed according to  $p(s_2,x_2|u_2)$. For fixed $u_2\in \mathcal{U}_2$, by $Q_{u_2}\in \mathscr{P}_{u_2}$ denote the PMF $p(s_2,x_2|u_2)$. This maximization problem is equivalent to finding  
\begin{align}\label{eq: p.t.p}
 R(u_2, Q_{u_2})\triangleq H(S_2|u_2)+ \max_{u_1\in \mathcal{U}_1} \max_{p(s_1,x_1|u_1)\in \mathscr{P}_{u_1} }H(S_1|u_1)-H(X_1+ S_1+ N |u_1)-2.
\end{align}
 Consider the problem of PtP channel with state, where the channel is $Y=X_1+ S_1 + N$.  It can be shown that $ R(u_2, Q_{u_2})-H(S_2|u_2)$ is an upper-bound on the capacity of this problem.  We proceed by the following lemma.
\begin{lem}\label{lem: R(u_2, Q)< 0.32}
The following bound holds $R(u_2, Q_{u_2})<1 $ for all $u_2\in \mathcal{U}_2$ and $Q_{u_2}\in \mathscr{P}_{u_2}$ . 
\end{lem}
\begin{proof}
The proof is given in Appendix \ref{sec: proof of lem R(u_2, Q)< 0.32}.
\end{proof}
 Finally, as a result of the above lemma the proof is completed.
\end{proof}

\section{Proof of Lemma \ref{lem: R(u_2, Q)< 0.32}}\label{sec: proof of lem R(u_2, Q)< 0.32}
\begin{proof}
Note that for any fixed $u_2\in \mathcal{U}_2$, the distribution of $N$ depends on  the conditional PMF $p(s_1|u_1)$, and the function $x_1=f_1(s_1,u_1)$.  For any $u\in \mathcal{U}_2$ define  
$$\mathcal{L}_u:=\{f_2(u,s)+ s: s\in \ZZ_4\}.$$
For any given $i\in \{1,2,3,4\}$,  define
$$\mathcal{B}_i \triangleq \{u\in \mathcal{U}_2:    |\mathcal{L}_u|=i   \}.$$
Note that $\mathcal{B}_i$'s are disjoint and $\mathcal{U}_2=\bigcup_i \mathcal{B}_i$. Depending on $u_2$, we consider four cases. In what follows, for each case, we derive an upper bound on  (\ref{eq: p.t.p}). Consider the PMF $p(\omega)$ on $\ZZ_4$. For brevity, we represent this PMF by the vector $\mathbf{p}:=(p(0), p(1), p(2), p(3))$. 

\subsection*{ Case 1: $u_2\in \mathcal{B}_1$}
Since $|\mathcal{L}_{u_2}|=1$, then for all $s_2\in \ZZ_4$ the equality $s_2+ f_2(s_2,u_2)=a$ holds, where $a\in\ZZ_4$ is a constant that only depends on $u_2$. This implies that conditioned on $u_2$, $X_2+ S_2$ equals to a constant $a$, with probability one. Therefore,
\begin{align*}
H(X_1+ S_1 + X_2 + S_2|u_2 u_1)=H(X_1+ S_1+ a|u_1u_2)=H(X_1+ S_1|u_1)
\end{align*}
Moreover, $$H(S_2|u_2)=H(a \ominus X_2|u_2)=H(X_2|u_2).$$ By assumption $p(u_2)>0$. Therefore, the cost constraint $\EE(c_2(X_2))=0$ implies that  $\EE(c_2(X_2)|U_2=u_2)=0$. Hence, given $U_2=u_2$, the random variable $X_2$ takes at most two values with positive probabilities.  As a result,  $H(X_2|u_2)\leq 1$. Given this inequality, we obtain
\begin{align*}
R(u_2, Q_{u_2}) \leq H(S_1|u_1)-H(X_1+ S_1|u_1)-1 \leq 0
\end{align*}
where the last inequality follows by Lemma \ref{lem mac state: bound on H(S)-H(X+S+N)} in Appendix \ref{sec: useful lemmas}.

\subsection*{ Case 2: $u_2\in \mathcal{B}_2$}
For any fixed $u_2\in \mathcal{B}_2$, $f_2(s_2, u_2) + s_2$ takes two values for all $s_2\in \ZZ_4$. Assume these values are $a, b\in \ZZ_4$, where $a\neq b$. Given $u_2$ the random variable  $X_2+ S_2$ is distributed over $\{a,b\}$. Therefore, $X_2 + S_2\ominus a$ is distributed over $\{0, b\ominus a\}$, and  
\begin{align*}
H(X_1+ S_1 + X_2 + S_2|u_2 u_1)=H(X_1+ S_1 + X_2 + S_2\ominus a |u_2 u_1).
\end{align*}
As a result, the case $\{a,b\}$ gives the same bound as $\{0, b\ominus a\}$, and we need to consider only the case in which $a=0$. For the case in which $a=0$, and  $b=3$, consider $X_2 + S_2+ 1$. Using a similar argument as above, we can show that when $b=3$, we get the same bound when $b=1$. Therefore, we only need to consider the cases in which $a=0$, and $b\in \{1, 2\}$. We address these cases in the next Claim.
%
\begin{claim}\label{claim: case 2}
Let $P(X_2 + S_2 =0 |u_1)=p_0$. The following holds: 

1) If $b=2$, then
\begin{align*}
R(u_2, Q_{u_2}) &\leq \beta(H(S_1|u_1)-H(X_1+ S_1 + N_{(2/3,0,1/3,0)}|u_1))\\
&+(1-\beta)(H(S_1|u_1)-H(X_1+ S_1+ N_{(1/3,0,2/3,0)}|u_1))+H(S_2|u_2)-2
\end{align*}

2) If $b=1$, then 
\begin{align*}
R(u_2, Q_{u_2})& \leq  \beta(H(S_1|u_1)-H(X_1+ S_1 + N_{(2/3,1/3,0,0)}|u_1))\\
&+(1-\beta)(H(S_1|u_1)-H(X_1+ S_1+ N_{(1/3,2/3,0,0)}|u_1))+H(S_2|u_2)-2
\end{align*}
\end{claim}
\begin{proof}
The proof is given in Appendix \ref{sec: proof of claim case 2}.
\end{proof}

Using the claim and applying Lemma \ref{lem mac state: bound on H(S)-H(X+S+N)},  we have
\begin{align*}
R(u_2, Q_{u_2})&<  1 + H(S_2|u_2)-2\leq  1.
\end{align*}

 \subsection*{ Case 3: $u_2\in \mathcal{B}_3$}
We need only to consider the case when $\mathbf{p}=(p_0, p_1, p_2, 0)$. We proceed by the following claim.

\begin{claim}
If $u_2\in \mathcal{B}_3$, the following bound holds
\begin{align*}
R(u_2, Q_{u_2}) & \leq \beta_0 (H(S_1|u_1)-H(X_1+ S_1 + N_{(2/4,1/4,1/4, 0)}|u_1))\\
&+\beta_1(H(S_1|u_1)-H(X_1+ S_1+ N_{ (1/4,2/4,1/4,0)}|u_1))\\
&+\beta_2 (H(S_1|u_1)-H(X_1+ S_1+ N_{ (1/4,1/4,2/4,0)}|u_1))+H(S_2|u_2)-2,
\end{align*}
 where $\beta_i=4p_i-1, ~i=0,1,2$.
\end{claim}

\begin{proof}
Similar to Claim \ref{claim: case 2}, we can write $\mathbf{p}$ as a linear combination of three distributions of the form $$\mathbf{p}=\beta_0 (2/4,1/4,1/4, 0)+\beta_1 (1/4,2/4,1/4,0)+\beta_2 (1/4,1/4,2/4,0),$$ where $\beta_i=4p_i-1, ~i=0,1,2$. The proof then follows from the concavity of the entropy.
\end{proof} 
Therefore, by Lemma \ref{lem mac state: bound on H(S)-H(X+S+N)}, we obtain 
 \begin{align*}
R(u_2, Q_{u_2}) &<   1 +H(S_2|u_2)-2\leq 1.
\end{align*}
 
 \subsection*{Case 4: $u_2\in \mathcal{B}_4$}
 In this case, there is a 1-1 correspondence between $x_2(s_2,u_2)+ s_2$ and $s_2$. Therefore $H(S_2|u_1,u_2)=H(S_2+ X_2|u_1,u_2)$, and we obtain
\begin{align*}
H(S_2|u_1,u_2)-H(X_1+ S_1 + X_2 + S_2|u_1,u_2)&=H(S_2+ X_2|u_1, u_2)-H(X_1+ S_1 + X_2 + S_2|u_1, u_2)\\&\leq 0
\end{align*} 
Therefore 
$H(S_1|u_1)+H(S_2|u_2)-H(Y|u_1u_2)-2\leq  H(S_1|u_1)-2\leq 0.$

Finally, considering all four cases $R(u_2, Q_{u_2}) <1$ for all $u_2 \in \mathcal{U}_2$. This completes the proof.  
\end{proof}

\section{Useful Lemmas}\label{sec: useful lemmas}
{
\begin{lem}\label{lem: sum of typical sets }
Let $X$ and $Y$ be independent random variables with marginal distributions $P_X$ and $P_Y$, respectively. Suppose $X$ and $Y$ take values from a group $\ZZ_{m}$. Then 
\begin{align*}
A_{\epsilon/2}^{(n)}(X+ Y) \subseteq A_\epsilon^{(n)}(X) + A_\epsilon^{(n)}(Y)
\end{align*}
\end{lem}

\begin{proof}
Let $\mathbf{z}\in A_{\epsilon/2}^{(n)}(X + Y)$. Select $\mathbf{y} \in A_{\epsilon/2}^{(n)}(Y|\mathbf{z})$. Since $\mathbf{z}$ is $\epsilon/2$- typical, then so is $\mathbf{y}$. In addition, $(\mathbf{z, \mathbf{y}}) \in A_\epsilon^{(n)}(X+ Y , Y )$. Let $\mathbf{x}=\mathbf{z} \ominus \mathbf{y}$. Then $(\mathbf{x}, \mathbf{y}) \in A_\epsilon^{(n)}(X, Y)$, and $\mathbf{x}+ \mathbf{y}=\mathbf{z}$. Note that $A_\epsilon^{(n)}(X, Y) \subseteq A_\epsilon^{(n)}(X) \times A_\epsilon^{(n)}(Y). $ This completes the proof.
\end{proof}
}


\begin{lem}[\cite{Aria_group_codes}] \label{lem: P(phi)}
Suppose that $\mathbf{G}$ is a $k\times n$ matrix with elements generated randomly and uniformly from $\ZZ_{p^r}$. If $\mathbf{u} \in H^k_s\backslash H^k_{s+1}$, then 
$$P\{ \mathbf{u} \mathbf{G}_i=\mathbf{x}\}= p^{-n(r-s)}  \11\{x\in H_s^n\}.$$
\end{lem}

\begin{lem} \label{lem: typical set intersection subgroup}
{Given $(X,Y)\sim P_{XY}$, and sequences $\mathbf{x}, \mathbf{y}$ such that $([\mathbf{x}]_s, \mathbf{y}) \in A_{\epsilon}^{(n)}([X]_s,Y)$, let $\mathcal{A}=\{\mathbf{x}' ~ | ~ (\mathbf{x}', \mathbf{y})\in A_\epsilon^n(XY), \mathbf{x}'-\mathbf{x} \in H^n_{s}\}.$
Then 
\begin{align*}
A_{c_1 \epsilon} ^{(n)}(X|[\mathbf{x}]_s, \mathbf{y}) \subseteq& \mathcal{A} \subseteq A_{c_2 \epsilon}^{(n)}(X| [\mathbf{x}]_s, \mathbf{y}),
\end{align*}
and we have,
\begin{align*}
(1-c_1\epsilon)2^{n(H(X|Y [X]_{s})-c_1\delta(\epsilon))}\leq &|\mathcal{A}| \leq 2^{n(H(X|Y [X]_{s})+c_2\delta (\epsilon))},
\end{align*}
where $\delta(\epsilon)=\frac{\epsilon}{|\mathcal{Y}|} \sum_{a\in \mathcal{X}} \sum_{b \in \mathcal{Y}: p(b|a)>0}\log_2 p(b|a)$,  $c_1=\frac{1}{|\mathcal{X}|+|\mathcal{Y}|}$, and $c_2=p^{r-s}  \frac{|\mathcal{X}|+1}{|\mathcal{Y}|}$.}
\end{lem}  
\begin{proof}
Suppose $\mathbf{x}' \in \mathcal{A}$. Then $ \mathbf{x}'-\mathbf{x} \in H^n_{s}$, which implies $[\mathbf{x'}]_s=[\mathbf{x}]_s$. In addition, $(\mathbf{x}', \mathbf{y}) \in A_{\epsilon}^{(n)}(X,Y)$. Therefore, $(\mathbf{x}', [\mathbf{x}]_s, \mathbf{y}) \in A_{\epsilon'}^{(n)}(X,[X],Y)$, where $\epsilon'=\epsilon p^{r-s}$. Thus, $\mathbf{x}' \in A_{\epsilon''}^{(n)}(X|  [\mathbf{x}]_s, \mathbf{y})$, where $\epsilon''=\frac{|\mathcal{X}|+1}{|\mathcal{Y}|}\epsilon'$. On the other hand, if $\mathbf{x'}\in A_{\tilde{\epsilon}}^{(n)}(X|[\mathbf{x}]_s \mathbf{y})$, then $[\mathbf{x'}]_s=[\mathbf{x}]_s$, and $\mathbf{x'}\in A_\epsilon^{(n)}(X|\mathbf{y})$, where $\epsilon =\tilde{\epsilon}(|\mathcal{X}|+|\mathcal{Y}|)$. 
\end{proof}

\begin{lem}\label{lem: H(X+Y)=H(X)}
Let $X$ and $Y$ be two independent random variables over $\ZZ_{m}$ with distributions $\mathbf{p}=(p_0, p_1,...,p_{m-1})$ and $\mathbf{q}=(q_0, q_1, ..., q_{m-1})$, respectively. Then $H(X\oplus_m Y)=H(Y)$ if and only if there exists $i\in [1:m]$ such that $\mathbf{p}  \circledast_m \mathbf{q}= \pi^i(\mathbf{q})$ , where $\pi((q_0, q_1, ..., q_{m-1}))=(q_{m-1},q_0, q_1, ..., q_{m-2})$, and $\pi^i$ is the composition of the function $\pi$ with itself for $i$ times. 
\end{lem}
\begin{proof}
First note that as $X$ is independent of $Y$, we have $ H(X\oplus_m Y)-H(Y)=I(X; X\oplus_m Y)\geq 0$. We find all distributions $\mathbf{p}$ and $\mathbf{q}$ for which the right-hand side equals zero. We first fix a distribution $\mathbf{q}$ and find all $\mathbf{p}$ such that the equality holds. The is equivalent to the solution of the following minimization problem:
\begin{align}\label{eq: minimizing over p}
\min_{\mathbf{p}\in \Delta_m} H( \mathbf{p} \circledast_m \mathbf{q})-H(\mathbf{q}),
\end{align}
where $\Delta_m\triangleq \{(q_0, q_1, ..., q_{m-1})\in \RR^m:    \sum_{i=0}^{m-1} q_i =1,  ~ q_i\geq 0, ~ i\in[0:m-1]\}$. Note that $\Delta_m$ is a $m-1$-dimensional simplex in $\RR^m$. Define the map $\varphi_{\mathbf{q}}: \Delta_m \mapsto \Delta_m, ~  \varphi_{\mathbf{q}}(\mathbf{p})= \mathbf{p} \circledast_m \mathbf{q}$ for all $\mathbf{p}, \mathbf{q}\in \Delta_m$. Note that $\varphi_{\mathbf{q}}$ is a linear map. Let $\varphi_{\mathbf{q}}(\Delta_m)$ denote the image of $\Delta_m$ under $\varphi_{\mathbf{q}}$. Since $\varphi_{\mathbf{q}}$ is a linear map, $\varphi_{\mathbf{q}}(\Delta_m)$ is a simplex. Therefore, \eqref{eq: minimizing over p} is equivalent to $\min_{\mathbf{p'}\in \varphi_{\mathbf{q}}(\Delta_m)} H( \mathbf{p'})-H(\mathbf{q})$. It is well-known that the entropy function is strictly concave. Hence, the minimum points are the extreme points of the simplex $\varphi_{\mathbf{q}}(\Delta_m)$. Extreme points of $\varphi_{\mathbf{q}}(\Delta_m)$ are the image of the extreme points of $\Delta_m$. Define the map $\pi: \Delta_m \mapsto \Delta_m$ as in the statement of the lemma. Extreme points of $\varphi_{\mathbf{q}}(\Delta_m)$ are characterized by $\pi^i(\mathbf{q}), i\in[1:m]$, where $\pi^i$ is the composition of $\pi$ with itself for $i$ times. 
Therefore, the minimum points of \eqref{eq: minimizing over p} are described as  $\bigcup_{i=1}^m \varphi_{\mathbf{q}}^{-1}(\pi^i(\mathbf{q}))$, where $\varphi^{-1}(\mathbf{a})$ is the pre-image of $\mathbf{a}, \forall \mathbf{a}\in \Delta_m$. 

Next, we range over all $\mathbf{q}\in \Delta_m$. Define the set $$\mathcal{A}_i \triangleq \{ (\mathbf{p}, \mathbf{q}) \in \Delta_m \times \Delta_m:  \mathbf{p} \circledast_m \mathbf{q}= \pi^i(\mathbf{q})\}.$$
Then, the set of all $(\mathbf{p,q})$ such that $H( \mathbf{p} \circledast_m \mathbf{q})=H(\mathbf{q})$ is characterized by the set $ \bigcup_{i=1}^m \mathcal{A}_i$. This is equivalent to the statement of the lemma.
\end{proof}

\begin{lem}\label{lem mac state: bound on H(S)-H(X+S+N)}
Suppose $S$ and $N_{\mathbf{p}}$ are independent random variables over $\ZZ_4$, where $\mathbf{p}$ is the distribution of $N_{\mathbf{p}}$. Let $f: \ZZ_4 \mapsto \ZZ_4$ be a function of $S$, and denote $X\triangleq f(S)$. If $\EE\{c_1(X)\}=0$, then the following bounds hold:
\begin{align*}
H(S)-H(X+ S)&\leq 1\\
H(S)-H(X+ S + N_{\mathbf{p}})&<1,
\end{align*}
where $\mathbf{p} \in \{(1/3, 0, 2/3,0), (1/3, 2/3,0,0),  (1/4,1/4,1/2,0)\}$.
\end{lem}
\begin{proof}
For the first equality, we start with the following equalities 
\begin{align*}
H(X+ S)&=H(X, S)-H(X|X+ S)\\
&=H(S)-H(X|X+ S).
\end{align*}
Therefore, we obtain
\begin{align*}
H(S)-H(X+ S)=H(X|X + S)\leq H(X)\stackrel{(a)}{\leq} 1.
\end{align*}
Note $(a)$ is true, because  $X$ takes at most two values with positive probabilities.

For the second inequality we have
\begin{align}\nonumber
H(S)-H(X+ S + N_{\mathbf{p}})&=H(S)-H(X+ S)+H(X+ S)-H(X+ S + N_{\mathbf{p}})\\\label{eq: H(s)-H(X+S+N)}
&\leq 1-(H(X+ S + N_{\mathbf{p}})-H(X+ S))\leq 1.
\end{align}
Let $\mathbf{q}$ be the distribution of $X+ S$. We find the conditions on  $\mathbf{p}$ and $\mathbf{q}$ for which  $H(X+ S + N_{\mathbf{p}})-H(X+ S)=0$. Since $N_{\mathbf{p}}$ is independent of $X+ S$,   we can use Lemma \ref{lem: H(X+Y)=H(X)} in which $Y=N_{\mathbf{p}}$ and $X=X+ S$. Therefore, $H(X+ S + N_{\mathbf{p}})=H(X+ S)$, if and only if  $\mathbf{p} \circledast_4 \mathbf{q}= \pi^i(\mathbf{q})$ for some $i\in [1:4]$. For fixed $i$ and $\mathbf{p}$,  the map defined by $\mathbf{q} \mapsto  \mathbf{p} \circledast_4 \mathbf{q}- \pi^i(\mathbf{q})$ is a linear map. In addition, the null space of this  map characterizes the set of all $\mathbf{q}$ that satisfies the equality in Lemma \ref{lem: H(X+Y)=H(X)}. 
For $\mathbf{p}=(1/3, 0, 2/3, 0)$ this map can be represented by the matrix 
%
\begin{align*}
A_{i, (1/3, 0, 2/3, 0)}=
\begin{bmatrix}
-\frac{2}{3} & 0 & \frac{2}{3} & 0\\
0 & -\frac{2}{3} & 0 & \frac{2}{3}\\
\frac{2}{3} & 0 & -\frac{2}{3} & 0 \\
0 & \frac{2}{3} & 0 & -\frac{2}{3}
\end{bmatrix}
\end{align*}
The null space of $\mathbf{A}_{i, (1/3, 0, 2/3, 0)}$ is the subspace spanned by $(1/2, 0, 1/2, 0)$ and $(1/4, 1/4, 1/4 , 1/4)$. Using the same approach, we can show that for any $i\in[1:4]$ and $$\mathbf{p} \in \{(1/3, 0, 2/3,0), (1/3, 2/3,0,0),  (1/4,1/4,1/2,0)\},$$ the null space of $\mathbf{A}_{i, \mathbf{p}}$ is contained in the subspace spanned by $(1/2, 0, 1/2, 0)$ and $(1/4, 1/4, 1/4 , 1/4)$. This implies that $q_0=q_2$ and $q_1=q_3$.  

\begin{table}[h]
\caption {The conditions on $x(\cdot)$ and $S$.}\label{tab: X+S and X, S}
\begin{center}
\begin{tabular}{|c|c|c|c|c|}
\hline
$X+ S$ & 0 & 1 & 2 & 3\\
\hline
$( s, x(s))$ & $(0,0), (2, 2)$ & $(1,0), (3, 2)$ & $(0,2), (2, 0)$ & $(1,2), (3, 0)$\\
\hline
\end{tabular}
\end{center}
\end{table}

Note $\mathbf{q}$ is the distribution of $x(S)+ S$. Next, we find all functions $x(\cdot )$ and random variables $S$ such that $q_0=q_2$ and $q_1=q_3$.   For each $a\in \ZZ_4$, we characterize $(s, x(s))$ such that $x(s)+ s=a$, where  $x(s)\in \{0,2\}$. We present such characterization in Table \ref{tab: X+S and X, S}. Using Table \ref{tab: X+S and X, S}, if $q_0>0$, then $p(S=0)=p(S=2)=q_0$ and $x(0)=x(2) $. Similarly, if $q_1>0$, then $p(S=1)=p(S=3)=q_1$ and $x(1)=x(3) $. Therefore, if $q_0, q_1>0$, the distribution of $S$ equals to $\mathbf{q}=(q_0, q_1,q_0, q_1)$.  If $q_0=0$, then $q_1=1/2$. This implies $p(S=1)=p(S=3)=1/2$. Similarly, If $q_1=0$, then $p(S=0)=p(S=2)=q_1=1/2$. As a result of this argument, $H(S)=H(X+ S)$. Also by Lemma \ref{lem: H(X+Y)=H(X)}, the equality  $H(X+ S)=H(X + S + N_{\mathbf{p}})$ holds. Therefore, in this case, $H(S)-H(X + S + N_{\mathbf{p}})=0$. To sum-up, we proved that if $\mathbf{p} \in \{(1/3, 0, 2/3,0), (1/3, 2/3,0,0),  (1/4,1/4,1/2,0)\}$ and $H(X+ S)=H(X + S + N_{\mathbf{p}})$, then $H(S)-H(X + S + N_{\mathbf{p}})=0$. Therefore, using this argument and \eqref{eq: H(s)-H(X+S+N)}, we proved that if  $\mathbf{p} \in \{(1/3, 0, 2/3,0), (1/3, 2/3,0,0),  (1/4,1/4,1/2,0)\}$, then $H(X+ S)-H(X + S + N_{\mathbf{p}})<1$. 
\end{proof}

\section{Proof of Claim \ref{claim: case 2}}\label{sec: proof of claim case 2}
\begin{proof}
\paragraph*{1)}
 Let $a=0, b=2$, and  $P(X_2 + S_2 =0 |u_1)=p_0$, and $P(X_2 + S_2 =2 |u_1)=1-p_0$. We represent this PMF by the vector $\mathbf{p}=(p_0, 0, 1-p_0, 0)$. This probability distribution is a linear combination of the form 
\begin{align}\label{eq: p case 1 lin comb}
\mathbf{p}=\beta (2/3, 0, 1/3, 0)+(1-\beta)(1/3,0,2/3,0),
\end{align} 
where $\beta=3p_0-1$.
 \begin{remark} \label{rem: circular conv and bilinear}
Let $Z=X+ Y$, where the PMF of $X$ is $\mathbf{p}=(p_0, p_1, p_2,p_3)$, and the PMF of $Y$ is $\mathbf{q}=(q_0, q_1, q_2,q_3)$. If $\mathbf{t}$ is the PMF of $Z$, then $\mathbf{t}=\mathbf{p} \circledast_4 \mathbf{q} $, where $ \circledast_4$  is the circular convolution in $\ZZ_4$. In addition, the map $(\mathbf{p} , \mathbf{q})\longmapsto \mathbf{p} \circledast_4 \mathbf{q} $ is a bi-linear map.
\end{remark}
Let $t_i=p(X_1+ S_1 + X_2+ S_2=i|u_1u_2)$ and $q_i=p(X_1+ S_1 =i |u_1)$ for all $i\in \ZZ_4$. Also denote $\mathbf{q}=(q_0, q_1, q_2,q_3)$, and $\mathbf{t}=(t_0, t_1, t_2, t_3)$. Using  Remark \ref{rem: circular conv and bilinear} and equation (\ref{eq: p case 1 lin comb}) we obtain 
\begin{align*}
\mathbf{t}&=\beta \big((2/3, 0, 1/3, 0)\circledast_4 \mathbf{q}\big)+(1-\beta)\big((1/3, 0, 2/3, 0)\circledast_4 \mathbf{q}\big).
\end{align*}
%
This implies that, $\mathbf{t}$ is also a linear combination of two PMFs. From the concavity of entropy, we get the following lower-bound:
\begin{align*}
H(X_1+ S_1 &+ X_2+ S_2|u_1u_2)=H(\mathbf{t})\\
&=H(\beta \big((2/3, 0, 1/3, 0)\circledast_4 \mathbf{q}\big)+(1-\beta)\big((1/3, 0, 2/3, 0)\circledast_4 \mathbf{q}\big))\\
& \geq \beta H((2/3, 0, 1/3, 0)\circledast_4 \mathbf{q}) +(1-\beta)H((1/3, 0, 2/3, 0)\circledast_4 \mathbf{q})\\  
& = \beta H(X_1+ S_1 + N_{(2/3, 0, 1/3, 0)}| u_1) +(1-\beta)H(X_1+ S_1+ N_{(1/3, 0, 2/3, 0)} | u_1),  
\end{align*}
where in the last equality, $N_{(\lambda_0, \lambda_1,\lambda_2, \lambda_3)}$ denotes a random variable with PMF $(\lambda_0, \lambda_1,\lambda_2, \lambda_3)$ that is also independent of $u_1$ and $X_1+ S_1$. As a result of the above argument, equation (\ref{eq: sum-rate last bound}) is bounded by
\begin{align*}
H(S_1|u_1)&+H(S_2|u_2)-H(Y|u_1u_2)-2\\& \leq H(S_1|u_1)+H(S_2|u_2) - \beta H(X_1+ S_1 + N_{(2/3, 0, 1/3, 0)}| u_1)\\ &-(1-\beta)H(X_1+ S_1+ N_{(1/3, 0, 2/3, 0)} | u_1)-2\\
&= \beta(H(S_1|u_1)-H(X_1+ S_1 + N_{(2/3,0,1/3,0)}|u_1))\\
&+(1-\beta)(H(S_1|u_1)-H(X_1+ S_1+ N_{(1/3,0,2/3,0)}|u_1))+H(S_2|u_2)-2
\end{align*}

\paragraph*{2)}
  Let $a=0, b=2$, and $P(X_2 + S_2 =0 |u_1)=p_0$, and $P(X_2 + S_2 =1 |u_1)=1-p_0$. In this case $\mathbf{p}=(p_0, 1-p_0, 0, 0)$. Also,  $$\mathbf{p}=\beta (2/3, 1/3, 0, 0)+(1-\beta)(1/3,2/3,0,0),$$ where $\beta=3p_0-1$. Similar to case 1), we use Remark \ref{rem: circular conv and bilinear} and the concavity of the entropy to get,
\begin{align*}
H(S_1|u_1)&+H(S_2|u_2)-H(Y|u_1u_2)-2\\& \leq  \beta(H(S_1|u_1)-H(X_1+ S_1 + N_{(2/3,1/3,0,0)}|u_1))\\
&+(1-\beta)(H(S_1|u_1)-H(X_1+ S_1+ N_{(1/3,2/3,0,0)}|u_1))+H(S_2|u_2)-2
\end{align*}
\end{proof}

\bibliographystyle{IEEEtran}
\bibliography{IEEEabrv,test2}

\end{document}